\documentclass[11pt, reqno]{amsart}

\usepackage{amsmath, amsthm, amssymb}
\usepackage{array}
\usepackage{enumitem}
\usepackage{pdflscape}
\usepackage{caption}
\usepackage{bm}

\usepackage{ifpdf}
\ifpdf
\usepackage[pdftex]{graphicx}
\else
\usepackage[dvips]{graphicx}
\fi
\usepackage{tikz}
 	 \usetikzlibrary{arrows,backgrounds}
   \usetikzlibrary{decorations.pathmorphing}
   \usetikzlibrary{knots}

\usepackage[all]{xy}

\usepackage{multicol}
\usepackage{mathtools} 
\usepackage{tocvsec2}
\usepackage{bbm}
\usepackage{stmaryrd} 

\input xy
\xyoption{all}

\usepackage[pdftex,plainpages=false,hypertexnames=false,pdfpagelabels]{hyperref}
\newcommand{\arxiv}[1]{\href{http://arxiv.org/abs/#1}{\tt arXiv:\nolinkurl{#1}}}
\newcommand{\arXiv}[1]{\href{http://arxiv.org/abs/#1}{\tt arXiv:\nolinkurl{#1}}}
\newcommand{\doi}[1]{\href{http://dx.doi.org/#1}{{\tt DOI:#1}}}

\newcommand{\googlebooks}[1]{(preview at \href{http://books.google.com/books?id=#1}{google books})}

\newcommand{\RN}[1]{\uppercase\expandafter{\romannumeral #1\relax}}

\usepackage{xcolor}
\definecolor{dark-red}{rgb}{0.7,0.25,0.25}
\definecolor{dark-blue}{rgb}{0.15,0.15,0.55}
\definecolor{medium-blue}{rgb}{0,0,.8}
\definecolor{DarkGreen}{RGB}{0,150,0}
\definecolor{darkgreen}{rgb}{0,150,0}
\definecolor{rho}{named}{red}
\hypersetup{
   colorlinks, linkcolor={purple},
   citecolor={medium-blue}, urlcolor={medium-blue}
}

\usepackage{longtable}
\usepackage{fullpage}

\theoremstyle{plain}
\newtheorem{thm}{Theorem}[section]
\newtheorem*{thm*}{Theorem}

\newtheorem{cor}[thm]{Corollary}

\newtheorem*{cor*}{Corollary}

\newtheorem*{conj*}{Conjecture}
\newtheorem{lem}[thm]{Lemma}

\newtheorem{prop}[thm]{Proposition}

\newtheorem*{quest*}{Question}
\newtheorem*{claim*}{Claim}

\theoremstyle{definition}
\newtheorem{defn}[thm]{Definition}

\newtheorem{alg}[thm]{Algorithm}

\newtheorem{sub-ex}[thm]{Sub-Example}
\newtheorem{rem}[thm]{Remark}
\newtheorem*{rem*}{Remark}

\DeclareMathOperator{\Absorb}{Absorb}

\DeclareMathOperator{\coev}{coev}

\DeclareMathOperator{\End}{End}
\DeclareMathOperator{\Fib}{Fib}
\DeclareMathOperator{\ev}{ev}

\DeclareMathOperator{\ONB}{ONB}

\DeclareMathOperator{\spann}{span}

\DeclareMathOperator{\id}{id}
\DeclareMathOperator{\Tube}{Tube}

\DeclareMathOperator{\im}{im}
\DeclareMathOperator{\Irr}{Irr}
\DeclareMathOperator{\loc}{loc}

\DeclareMathOperator{\tr}{tr}

\newcommand{\comment}[1]{}

\newcommand{\be}{\begin{enumerate}[label=(\arabic*)]}
\newcommand{\ee}{\end{enumerate}}

\newcommand{\Z}{\mathbb{Z}}

\newcommand{\C}{\mathbb{C}}

\newcommand{\set}[2]{\left\{#1 \middle| #2\right\}}

\renewcommand{\a}{\mathfrak{a}}

\newcommand{\PreTubeAC}{\widetilde{\Tube}_A(\mathcal{X})}
\newcommand{\Cl}{C_{\redl}}
\newcommand{\redl}{\textcolor{red}{\ell}}

\def\semicolon{;}
\def\applytolist#1{
    \expandafter\def\csname multi#1\endcsname##1{
        \def\multiack{##1}\ifx\multiack\semicolon
            \def\next{\relax}
        \else
            \csname #1\endcsname{##1}
            \def\next{\csname multi#1\endcsname}
        \fi
        \next}
    \csname multi#1\endcsname}

\def\calc#1{\expandafter\def\csname c#1\endcsname{{\mathcal #1}}}
\applytolist{calc}QWERTYUIOPLKJHGFDSAZXCVBNM;
\def\bbc#1{\expandafter\def\csname bb#1\endcsname{{\mathbb #1}}}
\applytolist{bbc}QWERTYUIOPLKJHGFDSAZXCVBNM;
\def\bfc#1{\expandafter\def\csname bf#1\endcsname{{\mathbf #1}}}
\applytolist{bfc}QWERTYUIOPLKJHGFDSAZXCVBNM;
\def\sfc#1{\expandafter\def\csname s#1\endcsname{{\sf #1}}}
\applytolist{sfc}QWERTYUIOPLKJHGFDSAZXCVBNM;
\def\fc#1{\expandafter\def\csname f#1\endcsname{{\mathfrak #1}}}
\applytolist{fc}QWERTYUIOPLKJHGFDSAZXCVBNM;

\newcommand{\Rep}{{\sf Rep}}

\renewcommand{\Vec}{{\sf Vec}}

\newcommand{\Hilb}{{\sf Hilb}}
\newcommand{\fdHilb}{{\sf Hilb_{fd}}}

\newcommand{\noshow}[1]{}
\renewcommand{\MR}[1]{}

\newcommand{\Cstar}{{\rm C^*}}

\usetikzlibrary{shapes}
\usetikzlibrary{cd}
\usetikzlibrary{backgrounds}
\usetikzlibrary{decorations,decorations.pathreplacing,decorations.markings}
\usetikzlibrary{fit,calc,through}
\usetikzlibrary{external}
\usetikzlibrary{arrows}
\tikzset{vertex/.style = {shape=circle,draw,fill=black,inner sep=0pt,minimum size=5pt}}
\tikzset{edge/.style = {->,> = latex', bend right}}
\tikzset{
	super thick/.style={line width=3pt}
}
\tikzset{
    quadruple/.style args={[#1] in [#2] in [#3] in [#4]}{
        #1,preaction={preaction={preaction={draw,#4},draw,#3}, draw,#2}
    }
}
\tikzstyle{shaded}=[fill=red!10!blue!20!gray!30!white]
\tikzstyle{unshaded}=[fill=white]
\tikzstyle{empty box}=[circle, draw, thick, fill=white, opaque, inner sep=2mm]
\tikzstyle{annular}=[scale=.7, inner sep=1mm, baseline]
\tikzstyle{rectangular}=[scale=.75, inner sep=1mm, baseline=-.1cm]
\tikzstyle{late>}=[decoration={markings, mark=at position 0.75 with {\arrow{>}}}, postaction={decorate}]
\tikzstyle{late<}=[decoration={markings, mark=at position 0.75 with {\arrow{<}}}, postaction={decorate}]
\tikzstyle{mid>}=[decoration={markings, mark=at position 0.5 with {\arrow{>}}}, postaction={decorate}]
\tikzstyle{mid<}=[decoration={markings, mark=at position 0.5 with {\arrow{<}}}, postaction={decorate}]
\tikzstyle{over}=[double, draw=white, super thick, double=]
\tikzstyle{box} = [rectangle,draw,rounded corners=5pt,very thick]

\tikzstyle{knot}=[preaction={super thick, white, draw}]

\newcommand{\roundNbox}[6]{
	\draw[rounded corners=5pt, very thick, #1] ($#2+(-#3,-#3)+(-#4,0)$) rectangle ($#2+(#3,#3)+(#5,0)$);
	\coordinate (ZZa) at ($#2+(-#4,0)$);
	\coordinate (ZZb) at ($#2+(#5,0)$);
	\node at ($1/2*(ZZa)+1/2*(ZZb)$) {#6};
}

\newcommand{\basisCircle}[3][(0,0)]{
 \filldraw[fill=#3] #1 circle (#2);
}

\newcommand{\mBasisCircle}[2]{
 \tikzmath{\basisCircle{#1}{#2}}
}

\newcommand{\tvttt}{\tikzmath{\draw(0:0)--(30:.4);\draw(0:0)--(150:.4);\draw(0:0)--(-90:.4);}}
\newcommand{\tvtti}{\tikzmath{\draw[dotted](0:0)--(30:.4);\draw(0:0)--(150:.4);\draw(0:0)--(-90:.4);}}
\newcommand{\tvitt}{\phi^{-1/2}\tikzmath{\draw(0:0)--(30:.4);\draw(0:0)--(150:.4);\draw[dotted](0:0)--(-90:.4);}}

\newcommand{\tikzmath}[2][]
     {\vcenter{\hbox{\begin{tikzpicture}[#1]#2
                     \end{tikzpicture}}}
     }

\newcommand{\halfDottedEllipse}[4][]{
	\draw[#1] #2 arc(-180:0:{#3} and {#4});
	\draw[#1, dotted] ($ #2 + 2*(#3,0)$) arc(0:180:{#3} and {#4});
}

\newcommand{\nhex}[5][]{
\coordinate (center) at (#2, #3);
\ifthenelse{\equal{#1}{}}{}{
\coordinate (pointD) at (canvas polar cs:angle=-30,radius=.866*#4cm);
\coordinate (pointE) at (canvas polar cs:angle=-30,radius=.288*#4cm);
\draw[#1] ($(center)+(pointD)$) -- +($#4*(-.333, .333)$);
}%
\foreach \hexSideCounter in {1,2,3,4,5,6} {
\coordinate (pointB) at (canvas polar cs:angle={60*\hexSideCounter},radius=#4cm);
\coordinate (pointA) at (canvas polar cs:angle={60*(\hexSideCounter - 1)},radius=#4cm);
\draw[#5] ($(center)+(pointA)$) -- ($(center)+(pointB)$);
}
}

\newcommand{\levinHex}[5][]{
\coordinate (center) at (#2, #3);
\ifthenelse{\equal{#1}{}}{}{
\coordinate (pointD) at (canvas polar cs:angle=-30,radius=.866*#4cm);
\coordinate (pointE) at (canvas polar cs:angle=-30,radius=.288*#4cm);
\draw[#1] ($(center)+(pointD)$) -- +($#4*(-.333, .333)$);
}
\foreach \hexSideCounter in {1,2,3,4,5,6} {
\coordinate (pointA) at (canvas polar cs:angle={60*(\hexSideCounter - 1)},radius=#4cm);
\coordinate (pointB) at (canvas polar cs:angle={60*\hexSideCounter},radius=#4cm);
\coordinate (pointC) at (canvas polar cs:angle={60*(\hexSideCounter - 1)},radius=1.5*#4cm);
\draw[#5] ($(center)+(pointA)$) -- ($(center)+(pointB)$); \draw[#5] ($(center)+(pointA)$) -- ($(center)+(pointC)$);
}
}

\newcommand{\levinHexOpen}[6][]{
\coordinate (center) at (#2, #3);
\ifthenelse{\equal{#1}{}}{}{
\coordinate (pointD) at (canvas polar cs:angle=-30,radius=.866*#4cm);
\coordinate (pointE) at (canvas polar cs:angle=-30,radius=.288*#4cm);
\draw[#1] ($(center)+(pointD)$) -- +($#4*(-.333, .333)$);
}
\foreach \hexSideCounter in {1,2,3,4,5,6} {
\coordinate (pointA) at (canvas polar cs:angle={60*(\hexSideCounter - 1)},radius=#4cm);
\coordinate (pointB) at (canvas polar cs:angle={60*\hexSideCounter},radius=#4cm);
\coordinate (pointC) at (canvas polar cs:angle={60*(\hexSideCounter - 1)},radius=1.5*#4cm);
\ifthenelse{#6=\hexSideCounter}{}{\draw[#5] ($(center)+(pointA)$) -- ($(center)+(pointB)$);}
\draw[#5] ($(center)+(pointA)$) -- ($(center)+(pointC)$);
}
}

\newcommand{\levinHexOpenShort}[7][]{
\coordinate (center) at (#2, #3);
\ifthenelse{\equal{#1}{}}{}{
\coordinate (pointD) at (canvas polar cs:angle=-30,radius=.866*#4cm);
\coordinate (pointE) at (canvas polar cs:angle=-30,radius=#7*#4cm);
\draw[#1] ($(center)+(pointD)$) -- ($(center)+(pointE)$);
}
\foreach \hexSideCounter in {1,2,3,4,5,6} {
\coordinate (pointA) at (canvas polar cs:angle={60*(\hexSideCounter - 1)},radius=#4cm);
\coordinate (pointB) at (canvas polar cs:angle={60*\hexSideCounter},radius=#4cm);
\coordinate (pointC) at (canvas polar cs:angle={60*(\hexSideCounter - 1)},radius=1.5*#4cm);
\ifthenelse{#6=\hexSideCounter}{}{\draw[#5] ($(center)+(pointA)$) -- ($(center)+(pointB)$);}
\draw[#5] ($(center)+(pointA)$) -- ($(center)+(pointC)$);
}
}

\newcommand{\arHex}[7][0]
{
\coordinate (center) at (#2, #3);
\foreach \hexSideCounter in {1,2,3,4,5,6} {
\coordinate (pointA) at (canvas polar cs:angle={60*(\hexSideCounter - 1)},radius=#4cm);
\coordinate (pointB) at (canvas polar cs:angle={60*\hexSideCounter},radius=#4cm);
\ifthenelse{#1=\hexSideCounter}{}{ 
\ifthenelse{#6=\hexSideCounter}
{\ifthenelse{#7=1} 
{\draw[#5] ($(center)+(pointB)$) -- node[sloped,allow upside down,pos=0.5] {\arrowIn[#5]} ($(center)+(pointA)$);} 
{\draw[#5] ($(center)+(pointA)$) -- node[sloped,allow upside down,pos=0.5] {\arrowIn[#5]} ($(center)+(pointB)$);} 
}
{\draw[#5] ($(center)+(pointA)$) -- ($(center)+(pointB)$);} 
}
}
}

\newcommand{\arHexKnot}[7][0]
{
\coordinate (center) at (#2, #3);
\foreach \hexSideCounter in {1,2,3,4,5,6} {
\coordinate (pointA) at (canvas polar cs:angle={60*(\hexSideCounter - 1)},radius=#4cm);
\coordinate (pointB) at (canvas polar cs:angle={60*\hexSideCounter},radius=#4cm);
\ifthenelse{#1=\hexSideCounter}{}{ 
\ifthenelse{#6=\hexSideCounter}
{\ifthenelse{#7=1} 
{\draw[#5,knot] ($(center)+(pointB)$) -- node[sloped,allow upside down,pos=0.5] {\arrowIn[#5]} ($(center)+(pointA)$);} 
{\draw[#5,knot] ($(center)+(pointA)$) -- node[sloped,allow upside down,pos=0.5] {\arrowIn[#5]} ($(center)+(pointB)$);} 
}
{\draw[#5,knot] ($(center)+(pointA)$) -- ($(center)+(pointB)$);} 
}
}
}

\newcommand{\levinHexRow}[6][]{
\foreach \hexColumnCounter in {1,...,#4}
{
\levinHex[#1]{{#2 + (1.5 * #5)*(\hexColumnCounter -1)}}{{#3 + (.866 * #5)*(\hexColumnCounter -1)}}{#5}{#6}
}
}
\newcommand{\levinHexGrid}[7][]{
\foreach \a in {1,...,#4}
{
\levinHexRow[#1]{#2 + (1.5 * #6)*(\a - 1)}{#3 - (.866 * #6)*(\a - 1)}{#5}{#6}{#7}
}
}

\newcommand{\arrowIn}[1][black]{
\tikz \draw[-stealth,#1] (-1pt,0) -- (1pt,0);
}

\newcommand{\arrowInR}[1][black]{
\tikz \draw[-stealth,#1] (1pt,0) -- (-1pt,0);
}

\begin{document}
\title{A lattice model for condensation in Levin-Wen systems}
\author{Jessica Christian, David Green, Peter Huston, and David Penneys}
\date{\today}
\begin{abstract}
 Levin-Wen string-net models provide a construction of (2+1)D topologically ordered phases of matter with anyonic localized excitations described by the {Drinfeld} center of a unitary fusion category.
 Anyon condensation is a mechanism for phase transitions between (2+1)D topologically ordered phases.
 We construct an extension of Levin-Wen models in which tuning a parameter implements anyon condensation.
 We also describe the classification of anyons in Levin-Wen models via representation theory of the tube algebra, and use a variant of the tube algebra to classify low-energy localized excitations in the condensed phase.
\end{abstract}
\maketitle

\section{Introduction}
Since their introduction in \cite{PhysRevB.71.045110}, string-net lattice models \cite{MR3204497,PhysRevB.103.195155} have been used as tractable examples of systems exhibiting (2+1)D topological order.
A (2+1)D string-net model is determined by the data of a unitary fusion category (UFC) $\mathcal{X}$, and exhibits $Z(\mathcal{X})$-topological order, in the sense that the modular tensor category $Z(\mathcal{X})$ classifies the types of quasiparticle excitations and gives their fusion and braiding statistics \cite{MR2942952}.

Systems exhibiting (2+1)D topological order can include domain walls between regions in different topological phases.
One source of topological domain walls is anyon condensation, where topological order is described by the unitary modular tensor category (UMTC) $\mathcal{C}$ on one side of the wall, while a bosonic condensable algebra $A\in\mathcal{C}$ is condensed on the other side \cite{PhysRevB.79.045316,MR2516228,MR3246855}.
The condensed algebra now plays the role of the vacuum, and anyons $s$ from the region where $A$ is not condensed can split as domain wall excitations, which may either be confined to the domain wall or able to pass into the condensed region, according to the structure of the fusion channels between $s$ and $A$.
Wall excitations are then described by the fusion category $\mathcal{C}_A$ of $A$-modules in $\mathcal{C}$, while the UMTC $\mathcal{C}_A^{\loc}$ of local $A$-modules describes anyons in the condensed phase \cite{MR3246855,MR3039775}.

Anyon condensation also gives rise to phase transitions between topological phases, which can be thought of as topological Wick rotations \cite{1912.01760} of the spatial domain walls.
In this article, we will describe a class of modified Levin-Wen models, due to Corey Jones, in which a chosen condensable algebra $A\in Z(\mathcal{X})$ may be condensed by tuning a parameter, driving a system with $Z(\mathcal{X})$ topological order through a phase transition to $Z(\mathcal{X})_A^{\loc}$ topological order.
In particular, a domain wall of the form described in \cite{MR3246855} can be created by choosing different values of the parameter on each side of the wall.
Our models closely track some existing constructions, such as that of \cite{PhysRevB.84.125434}, where models for the condensation of an Abelian plaquette excitation were constructed and analyzed, \cite{PhysRevB.94.235136}, which describes a procedure for ungauging a symmetry that is equivalent to our model for condensing an algebra of the form $\C^G$, or \cite{2209.12750}, which analyzes in great detail the case of condensing an algebra in $Z(\mathsf{Ising})$ to create a spatial boundary to $\mathbb{Z}/2$-toric code.
However, our models will allow for an arbitrary choice of UFC $\mathcal{X}$ and condensable algebra $A\in Z(\cX)$, and the modifications to the Hamiltonian come directly from the data of the condensable algebra.
In particular, $Z(\mathcal{\cX})$ and $A$ may be non-Abelian, and the fusion rules of $\mathcal{X}$ can have multiplicity.

The structure of this paper is as follows.
In Section \ref{sec:levin-wen}, we review Levin-Wen models in detail, including a description of string operators, hopping operators, and how the type of a topological excitation can be determined locally via representations of the tube algebra \cite{MR1782145,MR1966525}.
In Section \ref{sec:condensationModels}, we describe a class of models, parameterized by a unitary fusion category $\mathcal{
 X}$ and a condensable algebra $A\in Z(\mathcal{X})$, which permit the condensation of $A$ via tuning a parameter $t$.
When $t=0$, these models will reduce to the usual Levin-Wen Hamiltonian associated to $\mathcal{X}$, and when $t=1$, the algebra $A$ is condensed.
We describe a variant tube algebra of local operators and string operators in the condensed phase, and adapt the analysis from Section \ref{sec:levin-wen} to show that anyons in the condensed phase are described by the UMTC $Z(\cX)_A^{\loc}$, as argued in \cite{MR3246855}.
We also discuss the effect of the phase transition on the space of ground states in Section \ref{ssec:GSD}.
Finally, Section \ref{sec:examples} contains additional examples of the models described in \S~\ref{sec:condensationModels}, including the non-Abelian example of condensing the Lagrangian algebra in $Z(\Fib)$.

\textbf{Note added} Shortly before completing this work, we became aware of \cite{2303.07291}, which studies anyon condensation in the case of Abelian bosons.

\subsection*{Acknowledgements}
This project began as an undergraduate research project for Jessica Christian in Summer 2020 led by Peter Huston and David Green.
It then evolved into a chapter of Peter Huston's PhD thesis from 2022.
The authors would like to thank Corey Jones for suggesting this project and for many important ideas.
The authors would also like to thank 
Dave Aasen,
Maissam Barkeshli,
Jacob Bridgeman,
Fiona Burnell, and
Yuan-Ming Lu
for helpful comments and discussions.
All the authors were all supported by NSF grant DMS 1654159.
David Green and David Penneys were additionally supported by NSF grant DMS 2154389.

\section{String-net models in (2+1)D}
\label{sec:levin-wen}
In this section, we investigate (2+1)D string-net models for topological order, which were introduced in \cite{PhysRevB.71.045110}.
We follow the treatment of \cite{MR2942952,MR3204497,PhysRevB.103.195155,YanbaiZhang,0907.2204}.
We begin by introducing the commuting projector local Hamiltonian of a string-net model associated to the UFC $\cX$, in \S~\ref{ssec:lwModel}.
The goal of our analysis is to identify the space of states containing an isolated topological excitation at a particular location as a representation of the tube algebra $\Tube(\cX)$, extending the work of \cite{PhysRevB.97.195154}.
We accomplish this goal in \S~\ref{ssec:tubeImplementation}.
As setup, we introduce in \S~\ref{ssec:stringOperators} notions of string operators and hopping operators for anyons in $Z(\cX)$ which are slightly more general than those that appear in \cite{PhysRevB.97.195154,PhysRevB.103.195155}, so that string operators can realize all elements of the $\Tube(\cX)$ representation.
Aside from providing details on well-known properties of string-net models, the exposition in this section provides the blueprint for our analysis of the condensed phase in \S~\ref{sec:condensationModels}.

\subsection{Background: The Levin-Wen system}
\label{ssec:lwModel}
We begin by explicitly describing the string-net model associated to a unitary fusion category $\cX$.
Whenever possible, we suppress notation such as tensor products, associators, and unitors.

As in \cite{MR3204497}, we use a regular hexagonal 2D lattice which we view as being oriented \emph{left to right}, although this choice of geometry is not necessary.
$$
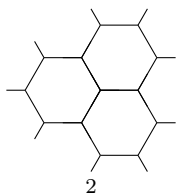
\tikzmath{
\levinHex{0}{0}{.5}{black}
\levinHex{.75}{.433}{.5}{black}
\levinHex{.75}{-.433}{.5}{black}
}
$$
We assign a Hilbert space to each vertex of the lattice, where there are two different types of vertices:
\begin{align*}
\tikzmath{
\draw (0:.5cm) node[below, xshift=.1cm] {$\scriptstyle v$} -- (60:.5cm);
\draw (0:.5cm) -- (-60:.5cm);
\draw (0:.5cm) -- (0:1cm);
}
\qquad
&\longleftrightarrow
\qquad
\cH_v := \bigoplus_{x,y,z\in \Irr(\cX)}\cX(x y \to z)
\\
\tikzmath{
\draw (0:0) node[below, xshift=-.1cm] {$\scriptstyle v$} -- (-60:.5cm);
\draw (0:0) -- (180:.5cm);
\draw (0:0) -- (60:.5cm);
}
\qquad
&\longleftrightarrow
\qquad
\cH_v := \bigoplus_{x,y,z\in \Irr(\cX)}\cX(x\to y  z)
\end{align*}
Spaces of morphisms, such as $\cX(x\to y z)$, carry several different inner products.
There are two which we consider in this work.
The first is the \textit{isometry inner product},
determined by the formula
\begin{equation}
 \label{eq:isometryIP}
\langle f,g\rangle \id_a
:=
g^\dag\circ f
\qquad\qquad
\forall f,g\in \cX(x\to y z)
\end{equation}
This gives a canonical identification of $\cX(y z \to x)$ with the dual Hilbert space $\overline{\cX(x\to y z)}$.

The isometry inner product appears naturally when computing compositions of morphisms in $\cX$.
However, the isometry inner product is ill-behaved in the sense that the isomorphisms $\bigoplus_{x,y,z}\cX(x\to y z)\to\bigoplus_{x,y,z}\cX(y z\to x)$ coming from pivotality of $\cX$ are not unitary; this inner product is not rotationally invariant.
Therefore, for the Hilbert spaces $\cH_v$, we choose a different inner product $\langle\cdot|\cdot\rangle_v$, where the inner product on the summand $\cX(x\to yz)$ is given by
\begin{equation}
 \label{eq:vertexIP}
 \langle g|f\rangle_v=\sqrt{\frac{d_x}{d_yd_z}}\langle g|f\rangle.
\end{equation}
In this paper, most inner products that are computed are actually the isometry inner product \eqref{eq:isometryIP}, because they arise from the comparison of operators defined in terms of the graphical calculus of $\cX$.
The importance of using the rotationally invariant inner product \eqref{eq:vertexIP} for the lattice Hilbert spaces is that the plaquette term $B_p$ which we define below will actually be self-adjoint.

\begin{rem}
 \label{rem:whyIPV}
 We can see that the inner product \eqref{eq:vertexIP} is rotationally invariant by relating it to the pivotal trace.
 Since $\cX$ is a unitary fusion category, $\cX$ has a canonical unitary spherical structure \cite{MR1444286,MR2091457,MR4133163}, giving a pivotal trace $\tr$ such that $\tr(\id_x)=d_x$.
 Therefore, $\langle g|f\rangle_v=\frac{1}{\sqrt{d_xd_yd_z}}\tr(g^\dag f)$.
 Since $d_x=d_{\overline{x}}$, the scalar $\sqrt{d_xd_yd_z}$ is obviously rotationally invariant; so is the value of $\tr$.
\end{rem}

We will sometimes depict a vector $f\in\mathcal{X}(x y\to z)\subseteq\mathcal{H}_v$ as a picture where $v$ is labeled by the morphism $f$, and the links incident to $v$ are labeled by $x$, $y$, and $z$:
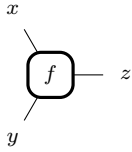
\[\tikzmath{
\draw (0:0) -- (0:.7cm);
\draw (0:0) -- (120:.7cm);
\draw (0:0) -- (-120:.7cm);
\roundNbox{fill=white}{(0,0)}{.3}{0}{0}{$\scriptstyle f$};
\node at (0:1cm) {$\scriptstyle z$};
\node at (120:1cm) {$\scriptstyle x$};
\node at (-120:1cm) {$\scriptstyle y$};
}\]
States on a finite chunk of the lattice where all links are assigned the same object in $\Irr(\cX)$ by all incident vertices can therefore be interpreted as linear combinations of string diagrams in $\cX$, read from left to right.

Sometimes, it will be more convenient to use other orientations of links, so we adopt the convention that
\[\tikzmath{
 \draw (0,0) -- (.5,0);
 \node at (.25,.2) {$\scriptstyle x$};
}=\tikzmath{
 \draw[mid>] (0,0) -- (.5,0);
 \node at (.25,.2) {$\scriptstyle x$};
}=\tikzmath{
 \draw[mid<] (0,0) -- (.5,0);
 \node at (.25,.2) {$\scriptstyle\overline{x}$};
}\]
For example, on a trivalent vertex corresponding to a hom space, we have
$$
\tikzmath{
\draw[mid<] (0:.5cm) -- node[left] {$\scriptstyle x$} (60:.5cm);
\draw[mid<] (0:.5cm) -- node[left] {$\scriptstyle y$} (-60:.5cm);
\draw[mid>] (0:.5cm) -- node[above] {$\scriptstyle z$} (0:1cm);
}
=
\cX(xy\to z)
\qquad\qquad
\tikzmath{
\draw[mid<] (0:.5cm) -- node[left] {$\scriptstyle x$} (60:.5cm);
\draw[mid<] (0:.5cm) -- node[left] {$\scriptstyle y$} (-60:.5cm);
\draw[mid<] (0:.5cm) -- node[above] {$\scriptstyle z$} (0:1cm);
}
=
\cX(xy\to \overline{z}).
$$

The Hamiltonian has two terms: link and plaquette.
The link term $A_\ell$ for a link $\ell$ connecting vertices $u$ and $v$ projects onto the subspace of $\cH_u\otimes \cH_v$ where the labels assigned to the link $\ell$ match.
Thus, $A_\ell$ terms commute with one another, and the ground states of $-\sum_\ell A_\ell$ can be locally interpreted as linear combinations of string diagrams in $\cX$ on the 1-skeleton of our lattice.
Note, however, that two string diagrams which give the same morphism in $\cX$ may be distinct as ground states of $-\sum_\ell A_\ell$.

Following \cite{MR3204497}, for $s\in\Irr(\mathcal{C})$, we define an operator $B_p^s$ which glues a closed $s$-loop into the plaquette $p$.
\begin{equation}
 \label{eq:Bps}
 \tikzmath{
  \pgfmathsetmacro{\hexIn}{.5};
  \pgfmathsetmacro{\hexOut}{.75};
  \draw[mid>] (0:\hexIn) -- (0:\hexOut);
  \draw[mid<] (0:\hexIn) -- (60:\hexIn);
  \draw[mid>] (60:\hexIn) -- (60:\hexOut);
  \draw[mid<] (60:\hexIn) -- (120:\hexIn);
  \draw[mid<] (120:\hexIn) -- (120:\hexOut);
  \draw[mid<] (120:\hexIn) -- (180:\hexIn);
  \draw[mid<] (180:\hexIn) -- (180:\hexOut);
  \draw[mid>] (180:\hexIn) -- (240:\hexIn);
  \draw[mid<] (240:\hexIn) -- (240:\hexOut);
  \draw[mid>] (240:\hexIn) -- (300:\hexIn);
  \draw[mid>] (300:\hexIn) -- (300:\hexOut);
  \draw[mid>] (300:\hexIn) -- (360:\hexIn);
  \node at (0:1cm) {$\scriptstyle a_1$};
  \node at (60:1cm) {$\scriptstyle a_2$};
  \node at (120:1cm) {$\scriptstyle a_3$};
  \node at (180:1cm) {$\scriptstyle a_4$};
  \node at (240:1cm) {$\scriptstyle a_5$};
  \node at (300:1cm) {$\scriptstyle a_6$};
  \node at (30:.7cm) {$\scriptstyle c_1$};
  \node at (90:.7cm) {$\scriptstyle c_2$};
  \node at (150:.7cm) {$\scriptstyle c_3$};
  \node at (-150:.7cm) {$\scriptstyle c_4$};
  \node at (-90:.7cm) {$\scriptstyle c_5$};
  \node at (-30:.7cm) {$\scriptstyle c_6$};
 }\mapsto
 \tikzmath{
  \pgfmathsetmacro{\hexIn}{.5};
  \pgfmathsetmacro{\hexOut}{.75};
  \draw[mid>] (0:\hexIn) -- (0:\hexOut);
  \draw[mid<] (0:\hexIn) -- (60:\hexIn);
  \draw[mid>] (60:\hexIn) -- (60:\hexOut);
  \draw[mid<] (60:\hexIn) -- (120:\hexIn);
  \draw[mid<] (120:\hexIn) -- (120:\hexOut);
  \draw[mid<] (120:\hexIn) -- (180:\hexIn);
  \draw[mid<] (180:\hexIn) -- (180:\hexOut);
  \draw[mid>] (180:\hexIn) -- (240:\hexIn);
  \draw[mid<] (240:\hexIn) -- (240:\hexOut);
  \draw[mid>] (240:\hexIn) -- (300:\hexIn);
  \draw[mid>] (300:\hexIn) -- (300:\hexOut);
  \draw[mid>] (300:\hexIn) -- (360:\hexIn);
  \node at (0:1cm) {$\scriptstyle a_1$};
  \node at (60:1cm) {$\scriptstyle a_2$};
  \node at (120:1cm) {$\scriptstyle a_3$};
  \node at (180:1cm) {$\scriptstyle a_4$};
  \node at (240:1cm) {$\scriptstyle a_5$};
  \node at (300:1cm) {$\scriptstyle a_6$};
  \node at (30:.7cm) {$\scriptstyle c_1$};
  \node at (90:.7cm) {$\scriptstyle c_2$};
  \node at (150:.7cm) {$\scriptstyle c_3$};
  \node at (-150:.7cm) {$\scriptstyle c_4$};
  \node at (-90:.7cm) {$\scriptstyle c_5$};
  \node at (-30:.7cm) {$\scriptstyle c_6$};
  \draw[mid>] (0,0) circle (.3cm);
  \node at (.1,0) {$\scriptstyle s$};
 }
\end{equation}
The operator $B_p^s$ is only defined on the ground states of the $A_\ell$ terms for links $\ell$ of $p$; if $A_\ell|\phi\rangle\neq|\phi\rangle$ for one of those links, we define $B_p^s|\phi\rangle=0$.
We interpret \eqref{eq:Bps} as an operator on our Hilbert space, as described in \cite[Appendix C]{PhysRevB.71.045110}, using the following relation.
\begin{equation}
 \label{eq:fusionDecomp}
 \id_{xs}
 =
 \tikzmath{
 \draw (0,0) -- node[left]{$\scriptstyle x$} (.5,.5);
 \draw (1,0) -- node[right]{$\scriptstyle s$} (.5,.5);
 \draw (.5,.5) -- node[right]{$\scriptstyle y$} (.5,1);
 \draw (.5,1) -- node[left]{$\scriptstyle x$} (0,1.5);
 \draw (.5,1) -- node[right]{$\scriptstyle s$} (1,1.5);
 \filldraw[fill=red] (.5,.5) circle (.05cm);
 \filldraw[fill=red] (.5,1) circle (.05cm);
}
:=
 \sum_{\substack{
 y\in \Irr(\cC)
 \\
 \alpha \in \ONB(xs \to y)
 }}
 \tikzmath{
 \draw (0,0) -- node[left]{$\scriptstyle x$} (.5,.5);
 \draw (1,0) -- node[right]{$\scriptstyle s$} (.5,.5);
 \draw (.5,.5) -- node[right]{$\scriptstyle y$} (.5,1);
 \draw (.5,1) -- node[left]{$\scriptstyle x$} (0,1.5);
 \draw (.5,1) -- node[right]{$\scriptstyle s$} (1,1.5);
 \filldraw[fill=red] (.5,.5) circle (.05cm) node[below]{$\scriptstyle \alpha$};
 \filldraw[fill=red] (.5,1) circle (.05cm) node[above, yshift=.1]{$\scriptstyle \alpha^\dag$};
}.
\end{equation}
Here, we adapt the notation from \cite[Eq.~(3)]{MR3663592} and write a pair of nodes labelled by $\tikzmath{\draw[fill=red] (0,0) circle (.05cm);}$ to denote summing over an isometry orthonormal basis of $\bigoplus_{y\in \Irr(\cX)}\mathcal{X}(xs\to y)$ and its adjoint; the sum is independent of the choice of basis.
Applying equation \eqref{eq:fusionDecomp} six times allows us to rewrite $B_p^s$ as
\begin{equation}
\label{eq:bpExpansion}
\tikzmath{
\pgfmathsetmacro{\hexIn}{.5};
\pgfmathsetmacro{\hexOut}{.75};
\draw[mid>] (0:\hexIn) -- (0:\hexOut);
\draw[mid<] (0:\hexIn) -- (60:\hexIn);
\draw[mid>] (60:\hexIn) -- (60:\hexOut);
\draw[mid<] (60:\hexIn) -- (120:\hexIn);
\draw[mid<] (120:\hexIn) -- (120:\hexOut);
\draw[mid<] (120:\hexIn) -- (180:\hexIn);
\draw[mid<] (180:\hexIn) -- (180:\hexOut);
\draw[mid>] (180:\hexIn) -- (240:\hexIn);
\draw[mid<] (240:\hexIn) -- (240:\hexOut);
\draw[mid>] (240:\hexIn) -- (300:\hexIn);
\draw[mid>] (300:\hexIn) -- (300:\hexOut);
\draw[mid>] (300:\hexIn) -- (360:\hexIn);
\node at (0:1cm) {$\scriptstyle a_1$};
\node at (60:1cm) {$\scriptstyle a_2$};
\node at (120:1cm) {$\scriptstyle a_3$};
\node at (180:1cm) {$\scriptstyle a_4$};
\node at (240:1cm) {$\scriptstyle a_5$};
\node at (300:1cm) {$\scriptstyle a_6$};
\node at (30:.7cm) {$\scriptstyle c_1$};
\node at (90:.7cm) {$\scriptstyle c_2$};
\node at (150:.7cm) {$\scriptstyle c_3$};
\node at (-150:.7cm) {$\scriptstyle c_4$};
\node at (-90:.7cm) {$\scriptstyle c_5$};
\node at (-30:.7cm) {$\scriptstyle c_6$};
\draw[mid>] (0,0) circle (.3cm);
\node at (.1,0) {$\scriptstyle s$};
}
\mapsto
\tikzmath{
\pgfmathsetmacro{\hexIn}{.5};
\pgfmathsetmacro{\hexOut}{.75};
\draw[mid>] (0:\hexIn) -- (0:\hexOut);
\draw[mid>] (0:\hexIn) -- (60:\hexIn);
\draw[mid>] (60:\hexIn) -- (60:\hexOut);
\draw[mid>] (60:\hexIn) -- (120:\hexIn);
\draw[mid>] (120:\hexIn) -- (120:\hexOut);
\draw[mid>] (120:\hexIn) -- (180:\hexIn);
\draw[mid>] (180:\hexIn) -- (180:\hexOut);
\draw[mid>] (180:\hexIn) -- (240:\hexIn);
\draw[mid>] (240:\hexIn) -- (240:\hexOut);
\draw[mid>] (240:\hexIn) -- (300:\hexIn);
\draw[mid>] (300:\hexIn) -- (300:\hexOut);
\draw[mid>] (300:\hexIn) -- (360:\hexIn);
\node at (0:1cm) {$\scriptstyle a_1$};
\node at (60:1cm) {$\scriptstyle a_2$};
\node at (120:1cm) {$\scriptstyle \overline{a_3}$};
\node at (180:1cm) {$\scriptstyle \overline{a_4}$};
\node at (240:1cm) {$\scriptstyle \overline{a_5}$};
\node at (300:1cm) {$\scriptstyle a_6$};
\node at (30:.7cm) {$\scriptstyle \overline{c_1}$};
\node at (90:.7cm) {$\scriptstyle \overline{c_2}$};
\node at (150:.7cm) {$\scriptstyle \overline{c_3}$};
\node at (-150:.7cm) {$\scriptstyle c_4$};
\node at (-90:.7cm) {$\scriptstyle c_5$};
\node at (-30:.7cm) {$\scriptstyle c_6$};
\draw[mid>] (0,0) circle (.3cm);
\node at (.1,0) {$\scriptstyle s$};
}
\mapsto
\tikzmath{
\pgfmathsetmacro{\hexIn}{1};
\pgfmathsetmacro{\hexOut}{1.5};
\draw[mid>] (0:\hexIn) -- (0:\hexOut);
\draw[mid>] (0:\hexIn) -- (60:\hexIn);
\draw[mid>] (60:\hexIn) -- (60:\hexOut);
\draw[mid>] (60:\hexIn) -- (120:\hexIn);
\draw[mid>] (120:\hexIn) -- (120:\hexOut);
\draw[mid>] (120:\hexIn) -- (180:\hexIn);
\draw[mid>] (180:\hexIn) -- (180:\hexOut);
\draw[mid>] (180:\hexIn) -- (240:\hexIn);
\draw[mid>] (240:\hexIn) -- (240:\hexOut);
\draw[mid>] (240:\hexIn) -- (300:\hexIn);
\draw[mid>] (300:\hexIn) -- (300:\hexOut);
\draw[mid>] (300:\hexIn) -- (360:\hexIn);
\draw[draw=none, mid>] (0:\hexIn) -- (20:.88cm);
\draw[draw=none, mid>] (40:.88cm) -- (60:\hexIn);
\draw[draw=none, mid>] (60:\hexIn) -- (80:.88cm);
\draw[draw=none, mid>] (100:.88cm) -- (120:\hexIn);
\draw[draw=none, mid>] (120:\hexIn) -- (140:.88cm);
\draw[draw=none, mid>] (160:.88cm) -- (180:\hexIn);
\draw[draw=none, mid>] (180:\hexIn) -- (200:.88cm);
\draw[draw=none, mid>] (220:.88cm) -- (240:\hexIn);
\draw[draw=none, mid>] (240:\hexIn) -- (260:.88cm);
\draw[draw=none, mid>] (280:.88cm) -- (300:\hexIn);
\draw[draw=none, mid>] (300:\hexIn) -- (320:.88cm);
\draw[draw=none, mid>] (340:.88cm) -- (360:\hexIn);
\draw[mid>] (40:.88cm) arc (-60: -180:.35cm);
\draw[mid>] (100:.88cm) arc (0: -120:.35cm);
\draw[mid>] (160:.88cm) arc (60: -60:.35cm);
\draw[mid>] (220:.88cm) arc (120: 0:.35cm);
\draw[mid>] (280:.88cm) arc (180: 60:.35cm);
\draw[mid>] (340:.88cm) arc (240: 120:.35cm);
\filldraw[fill=red] (20:.88cm) circle (.05cm);
\filldraw[fill=red] (40:.88cm) circle (.05cm);
\filldraw[fill=orange] (80:.88cm) circle (.05cm);
\filldraw[fill=orange] (100:.88cm) circle (.05cm);
\filldraw[fill=yellow] (140:.88cm) circle (.05cm);
\filldraw[fill=yellow] (160:.88cm) circle (.05cm);
\filldraw[fill=DarkGreen] (-160:.88cm) circle (.05cm);
\filldraw[fill=DarkGreen] (-140:.88cm) circle (.05cm);
\filldraw[fill=blue] (-100:.88cm) circle (.05cm);
\filldraw[fill=blue] (-80:.88cm) circle (.05cm);
\filldraw[fill=purple] (-40:.88cm) circle (.05cm);
\filldraw[fill=purple] (-20:.88cm) circle (.05cm);
\node at (0:1.7cm) {$\scriptstyle a_1$};
\node at (60:1.7cm) {$\scriptstyle a_2$};
\node at (120:1.7cm) {$\scriptstyle \overline{a_3}$};
\node at (180:1.7cm) {$\scriptstyle \overline{a_4}$};
\node at (240:1.7cm) {$\scriptstyle \overline{a_5}$};
\node at (300:1.7cm) {$\scriptstyle a_6$};
\node at (0:.45cm) {$\scriptstyle s$};
\node at (60:.45cm) {$\scriptstyle s$};
\node at (120:.45cm) {$\scriptstyle s$};
\node at (180:.45cm) {$\scriptstyle s$};
\node at (240:.45cm) {$\scriptstyle s$};
\node at (300:.45cm) {$\scriptstyle s$};
\node at (12:1.2cm) {$\scriptstyle \overline{c_1}$};
\node at (30:1.2cm) {$\scriptstyle \overline{d_1}$};
\node at (48:1.2cm) {$\scriptstyle \overline{c_1}$};
\node at (72:1.2cm) {$\scriptstyle \overline{c_2}$};
\node at (90:1.2cm) {$\scriptstyle \overline{d_2}$};
\node at (108:1.2cm) {$\scriptstyle \overline{c_2}$};
\node at (132:1.2cm) {$\scriptstyle \overline{c_3}$};
\node at (150:1.2cm) {$\scriptstyle \overline{d_3}$};
\node at (168:1.2cm) {$\scriptstyle \overline{c_3}$};
\node at (192:1.2cm) {$\scriptstyle c_4$};
\node at (210:1.2cm) {$\scriptstyle d_4$};
\node at (228:1.2cm) {$\scriptstyle c_4$};
\node at (252:1.2cm) {$\scriptstyle c_5$};
\node at (270:1.2cm) {$\scriptstyle d_5$};
\node at (288:1.2cm) {$\scriptstyle c_5$};
\node at (312:1.2cm) {$\scriptstyle c_6$};
\node at (330:1.2cm) {$\scriptstyle d_6$};
\node at (348:1.2cm) {$\scriptstyle c_6$};
}
\mapsto
\tikzmath{
\pgfmathsetmacro{\hexIn}{1};
\pgfmathsetmacro{\hexOut}{1.5};
\draw[mid>] (0:\hexIn) -- (0:\hexOut);
\draw[mid<] (0:\hexIn) -- (60:\hexIn);
\draw[mid>] (60:\hexIn) -- (60:\hexOut);
\draw[mid<] (60:\hexIn) -- (120:\hexIn);
\draw[mid<] (120:\hexIn) -- (120:\hexOut);
\draw[mid<] (120:\hexIn) -- (180:\hexIn);
\draw[mid<] (180:\hexIn) -- (180:\hexOut);
\draw[mid>] (180:\hexIn) -- (240:\hexIn);
\draw[mid<] (240:\hexIn) -- (240:\hexOut);
\draw[mid>] (240:\hexIn) -- (300:\hexIn);
\draw[mid>] (300:\hexIn) -- (300:\hexOut);
\draw[mid>] (300:\hexIn) -- (360:\hexIn);
\draw[draw=none, mid<] (0:\hexIn) -- (20:.88cm);
\draw[draw=none, mid<] (40:.88cm) -- (60:\hexIn);
\draw[draw=none, mid<] (60:\hexIn) -- (80:.88cm);
\draw[draw=none, mid<] (100:.88cm) -- (120:\hexIn);
\draw[draw=none, mid<] (120:\hexIn) -- (140:.88cm);
\draw[draw=none, mid<] (160:.88cm) -- (180:\hexIn);
\draw[draw=none, mid>] (180:\hexIn) -- (200:.88cm);
\draw[draw=none, mid>] (220:.88cm) -- (240:\hexIn);
\draw[draw=none, mid>] (240:\hexIn) -- (260:.88cm);
\draw[draw=none, mid>] (280:.88cm) -- (300:\hexIn);
\draw[draw=none, mid>] (300:\hexIn) -- (320:.88cm);
\draw[draw=none, mid>] (340:.88cm) -- (360:\hexIn);
\draw[mid>] (40:.88cm) arc (-60: -180:.35cm);
\draw[mid>] (100:.88cm) arc (0: -120:.35cm);
\draw[mid>] (160:.88cm) arc (60: -60:.35cm);
\draw[mid>] (220:.88cm) arc (120: 0:.35cm);
\draw[mid>] (280:.88cm) arc (180: 60:.35cm);
\draw[mid>] (340:.88cm) arc (240: 120:.35cm);
\filldraw[fill=red] (20:.88cm) circle (.05cm);
\filldraw[fill=red] (40:.88cm) circle (.05cm);
\filldraw[fill=orange] (80:.88cm) circle (.05cm);
\filldraw[fill=orange] (100:.88cm) circle (.05cm);
\filldraw[fill=yellow] (140:.88cm) circle (.05cm);
\filldraw[fill=yellow] (160:.88cm) circle (.05cm);
\filldraw[fill=DarkGreen] (-160:.88cm) circle (.05cm);
\filldraw[fill=DarkGreen] (-140:.88cm) circle (.05cm);
\filldraw[fill=blue] (-100:.88cm) circle (.05cm);
\filldraw[fill=blue] (-80:.88cm) circle (.05cm);
\filldraw[fill=purple] (-40:.88cm) circle (.05cm);
\filldraw[fill=purple] (-20:.88cm) circle (.05cm);
\node at (0:1.7cm) {$\scriptstyle a_1$};
\node at (60:1.7cm) {$\scriptstyle a_2$};
\node at (120:1.7cm) {$\scriptstyle a_3$};
\node at (180:1.7cm) {$\scriptstyle a_4$};
\node at (240:1.7cm) {$\scriptstyle a_5$};
\node at (300:1.7cm) {$\scriptstyle a_6$};
\node at (0:.45cm) {$\scriptstyle s$};
\node at (60:.45cm) {$\scriptstyle s$};
\node at (120:.45cm) {$\scriptstyle s$};
\node at (180:.45cm) {$\scriptstyle s$};
\node at (240:.45cm) {$\scriptstyle s$};
\node at (300:.45cm) {$\scriptstyle s$};
\node at (12:1.2cm) {$\scriptstyle c_1$};
\node at (30:1.2cm) {$\scriptstyle d_1$};
\node at (48:1.2cm) {$\scriptstyle c_1$};
\node at (72:1.2cm) {$\scriptstyle c_2$};
\node at (90:1.2cm) {$\scriptstyle d_2$};
\node at (108:1.2cm) {$\scriptstyle c_2$};
\node at (132:1.2cm) {$\scriptstyle c_3$};
\node at (150:1.2cm) {$\scriptstyle d_3$};
\node at (168:1.2cm) {$\scriptstyle c_3$};
\node at (192:1.2cm) {$\scriptstyle c_4$};
\node at (210:1.2cm) {$\scriptstyle d_4$};
\node at (228:1.2cm) {$\scriptstyle c_4$};
\node at (252:1.2cm) {$\scriptstyle c_5$};
\node at (270:1.2cm) {$\scriptstyle d_5$};
\node at (288:1.2cm) {$\scriptstyle c_5$};
\node at (312:1.2cm) {$\scriptstyle c_6$};
\node at (330:1.2cm) {$\scriptstyle d_6$};
\node at (348:1.2cm) {$\scriptstyle c_6$};
}
\end{equation}
Note that we switch orientations in the first and third arrows for ease of applying equation \eqref{eq:fusionDecomp}.
The pairs of colored vertices refer to summing over an orthonormal basis and dual basis, as in \eqref{eq:fusionDecomp}, while labels $f_{1\cdots 6}$ for the six vertices have been omitted to avoid clutter.
The second arrow requires the use of the associator/$F$-matrices to re-associate in order to apply \eqref{eq:fusionDecomp}.
Explicitly,
$$
\id_s\otimes f
=
\tikzmath{
\draw (-.8,.6) node[left]{$\scriptstyle x$} -- (.8,.6) node[right]{$\scriptstyle x$} ;
\draw (0,.2) -- (.8,.2) node[right]{$\scriptstyle a$};
\draw (0,-.2) -- (.8,-.2) node[right]{$\scriptstyle b$};
\draw (0,0) -- (-.8,0) node[left]{$\scriptstyle c$};
\roundNbox{fill=white}{(0,0)}{.4}{0}{0}{$f$}
}
=
\tikzmath{
\draw (1,.8) -- (-.4,.8) arc (90:270:.4cm);
\draw (-.8,.4) -- node[above]{$\scriptstyle d$} (-1.2,.4);
\draw (-1.6,.8) node[left]{$\scriptstyle s$} arc (90:-90:.4cm) node[left]{$\scriptstyle c$};
\draw (0,.2) -- (1,.2);
\draw (3.8,.2) -- (4.2,.2) node[right]{$\scriptstyle a$};
\draw (0,-.2) -- (4.2,-.2) node[right]{$\scriptstyle b$};
\draw (1.8,.8) -- (2,.8) arc (90:-90:.2cm) -- (1.8,.4);
\draw (4.2,.8) node[right]{$\scriptstyle s$} -- (2.8,.8) arc (90:270:.2cm) -- (3,.4);
\draw (2.2,.6) -- node[above]{$\scriptstyle e$} (2.6,.6);
\roundNbox{fill=white}{(0,0)}{.4}{0}{0}{$f$}
\roundNbox{fill=white}{(1.4,.3)}{.7}{-.3}{-.3}{$\alpha$}
\roundNbox{fill=white}{(3.4,.3)}{.7}{-.3}{-.3}{$\alpha^{-1}$}
\filldraw[fill=DarkGreen] (2.2,.6) circle (.05cm);
\filldraw[fill=DarkGreen] (2.6,.6) circle (.05cm);
\filldraw[fill=blue] (-1.2,.4) circle (.05cm);
\filldraw[fill=blue] (-.8,.4) circle (.05cm);
\node at (2.4,0) {$\scriptstyle b$};
\node at (.7,0) {$\scriptstyle b$};
\node at (.7,.4) {$\scriptstyle a$};
\node at (-.6,-.2) {$\scriptstyle c$};
\node at (0,1) {$\scriptstyle s$};
}
$$
Thus, in the final diagram of \eqref{eq:bpExpansion}, each vertex of $p$ is now labelled by the composition of several morphisms, yielding a new morphism in $\mathcal{H}_v$.
For example, in terms of the sum where $\tikzmath{\filldraw[fill=blue] (0:0) circle (.05cm);}=\phi$, $\tikzmath{\filldraw[fill=DarkGreen] (0:0) circle (.05cm);}=\psi$, the lower right vertex is now labelled by
an element of $\mathcal{X}(d_5\to d_6a_6)$:
\[\tikzmath{
 \draw (1,.8) -- (-1,.8);
 \draw (-.8,.4) -- (-1.2,.4);
 \draw (-1.8,.4) node[left]{$\scriptstyle d_5$} -- (-1,.4);
 \draw (0,.2) -- (1,.2);
 \draw (1.4,.2) -- (2.3,.2) node[right]{$\scriptstyle c_6$};
 \draw (0,-.2) -- (3,-.2) node[right]{$\scriptstyle a_6$};
 \draw (-1.2,.8) -- (2.3,.8);
 \draw (0,0) -- (-1.2,0);
 \draw (2.4,.6) --  (3,.6) node[right]{$\scriptstyle d_6$};
 \roundNbox{fill=white}{(0,0)}{.4}{0}{0}{$f$}
 \roundNbox{fill=white}{(1.4,.3)}{.7}{-.3}{-.3}{$\alpha$}
 \roundNbox{fill=white}{(-1.2,.4)}{.6}{-.2}{-.2}{$\phi^\dag$}
 \roundNbox{fill=white}{(2.4,.55)}{.55}{-.15}{-.15}{$\psi$}
 \node at (.7,0) {$\scriptstyle a_6$};
 \node at (.7,.4) {$\scriptstyle c_6$};
 \node at (-.6,.2) {$\scriptstyle c_5$};
 \node at (0,1) {$\scriptstyle s$};
}\]

The plaquette term $B_p$ is then given by
\[B_p=\frac{1}{D}\sum_{s\in\Irr(\mathcal{X})}d_sB_p^s\text{,}\]
where $D=\sum_sd_S^2$ is the global dimension \cite[Defn.~2.5]{MR1966524} of $\cX$.
One uses the associativity of $\mathcal{X}$ to check that $B_p$ is an idempotent \cite[\S5]{YanbaiZhang}.
The computation that $B_p$ is self-adjoint is somewhat involved, but appears in \cite[Theorem 5.0.1]{0907.2204}.
The inner product $\langle\cdot|\cdot\rangle_v$ of \eqref{eq:vertexIP} is chosen so that the notion of Hermitian operator used in the proof in \cite{0907.2204} matches the notion in our Hilbert space.

There are several interpretations for $B_p$.
One is that after applying $B_p$, strings can now be deformed across the plaquette $p$ \cite[Appendix C]{PhysRevB.71.045110}.
Another is that the operator $B_p$ is the orthogonal projector onto the trivial representation of the algebra $\Tube(\mathcal{X})$, as described in Section \ref{ssec:stringOperators}.
Still another is that $B_p$ amounts to contracting $p$ to a vertex, and then restoring it \cite{PhysRevB.85.075107}.
Categorically, this corresponds to applying the map $E^\dag E$ (up to a factor of the global dimension $D$), where $E$ is the map which uses the composition of $\cX$ to replace a string diagram labelling $p$ with a single morphism in $\cX$ \cite{MR3204497}.
This last interpretation thus provides an alternative definition of $B_p$ which is manifestly self-adjoint.

Finally, the overall Hamiltonian on the lattice Hilbert space $\otimes_v\cH_v$ is given by
\begin{equation}
 \label{eq:LWHamiltonian}
 H=-\sum_vA_v-\sum_pB_p
\end{equation}

\subsection{Topological Excitations from String Operators}
\label{ssec:stringOperators}
Topological excitations in a Levin-Wen model based on the UFC $\mathcal{X}$ are classified by simple objects in $Z(\mathcal{X})$, the Drinfeld center of $\mathcal{C}$ \cite{MR2942952}.
One way of deriving $Z(\cX)$ from $\cX$ is by means of the tube algebra $\Tube(\cX)$, a finite dimensional $\Cstar$ algebra whose category of representations is equivalent to $Z(\cX)$ \cite{MR1966525,MR1782145}.
An action of the tube algebra as local operators at the site of topological excitation is described in \cite{PhysRevB.97.195154}, providing a natural and local way to identify topological excitations with objects in $Z(\cX)$.
Topological excitations are created by string operators, families of operators determined by an object $c\in\Irr(Z(\cX))$ and a path $p$ through the lattice which create excitations of types $c$ and $\overline{c}$ at the ends of $p$.
The string operator preserves the ground state in the middle of $p$, and if $p$ and $q$ are homotopic paths with the same endpoints, then string operators along $p$ and $q$ agree as long as the part of lattice through which the homotopy must pass is in the ground state \cite[Appendix C]{PhysRevB.71.045110}.

In \cite{PhysRevB.97.195154}, as in many sources, the authors provide a unique string operator for every object in $\Irr(Z(\cX))$, and consequently do not identify the particular representation of $\Tube(\cX)$ at each site.
In this section, we generalize their notion of string operator, so that string operators can absorb elements of $\Tube(\cX)$ acting locally on each end.
We apply this more general notion of string operator in \S\ref{ssec:tubeImplementation}, to determine the exact local representation of $\Tube(\cX)$.
In \S\ref{ssec:condensedPhase}, we will apply and generalize the constructions of \S\ref{ssec:tubeImplementation} to our model of anyon condensation in order to prove Theorem \ref{thm:correctMTC}.

An object in $Z(\mathcal{C})$ is an object $X\in\mathcal{C}$, together with a \textit{half-braiding}, i.e., a unitary natural isomorphism $\rho:X\otimes-\Rightarrow-\otimes X$ which satisfies certain coherences.
Given the data $s=(X,\rho)\in Z(\mathcal{X})$, we can define a class of string operators $\sigma^s_p(\phi,\psi)$, parameterized by a choice of oriented path $p$ between two potential location of excitations and vectors $\phi,\psi\in\oplus_{y\in\Irr(\mathcal{X})}\mathcal{X}(X\to y)$.

We begin by discussing the potential choices of the path $p$.
The location of an excitation is determined by which terms of the Hamiltonian the excitation violates.
In the case of our hexagonal model, the location of an excitation is therefore a pair $(q,\ell)$, where $q$ is a plaquette and $\ell$ is an edge of $q$.
Therefore, the string determined by the path $p$ must begin and end at points inside a plaquette, near a specific boundary edge, and must also avoid the center of each plaquette, as well as the vertices of the lattice.
Thus, $p$ consists of a list of pairs $(q_i,\ell_i)$, where each $\ell_i$ is an edge of $q_i$, either $q_{i+1}=q_i$ or $\ell_{i+1}=\ell_i$, and if $\ell_{i+1}\neq\ell_i$, then the two edges share a common vertex.
An example of such a $p$ appears in Figure \ref{fig:stringOperator}; compare \cite[Fig. 19]{PhysRevB.71.045110}.
\begin{figure}\[
\tikzmath{
\coordinate (a) at (-150:.53cm);
\coordinate (b) at (-90:.55cm);
\coordinate (c) at (-50:.55cm);
\coordinate (d) at (-20:.78cm);
\coordinate (e) at (.95,-.1);
\coordinate (f) at (1.3,.1);
\coordinate (g) at (1.55,.27);
\coordinate (h) at (1.85,.4);
\coordinate (i) at (2.15,.55);
\coordinate (j) at (2.45,.75);
\coordinate (z) at (2.75,1.1);
\filldraw[red] (a) circle (.05cm);
\filldraw[red] (z) circle (.05cm);
\draw[red, thick] (a) .. controls ++(-45:.2cm) and ++(180:.4cm) .. (b) node [sloped,allow upside down,pos=0.5] {\arrowIn[red]}
                            .. controls ++(0:.2cm) and ++(-135:.2cm) .. (c) node [sloped,allow upside down,pos=0.5] {\arrowIn[red]}
                            .. controls ++(45:.2cm) and ++(-135:.2cm) .. (d)
                            .. controls ++(45:.1cm) and ++(180:.1cm) .. (e) node [sloped,allow upside down,pos=0.5] {\arrowIn[red]}
                            .. controls ++(0:.2cm) and ++(180:.2cm) .. (f)
                            .. controls ++(0:.1cm) and ++(-135:.1cm) .. (g) node [sloped,allow upside down,pos=0.5] {\arrowIn[red]}
                            .. controls ++(45:.2cm) and ++(-135:.2cm) .. (h)
                            .. controls ++(45:.2cm) and ++(180:.2cm) .. (i) node [sloped,allow upside down,pos=0.5] {\arrowIn[red]}
                            .. controls ++(0:.2cm) and ++(180:.2cm) .. (j)
                            .. controls ++(0:.2cm) and ++(-135:.2cm) .. (z) node [sloped,allow upside down,pos=0.5] {\arrowIn[red]}
                            ;
\fill[white] (.55,-.35) circle (.07cm);
\fill[white] (1.14,0) circle (.07cm);
\fill[white] (1.7,.35) circle (.07cm);
\fill[white] (2.3,.65) circle (.07cm);
\levinHexGrid[]{0}{0}{2}{3}{.75}{black}
}\]
 \caption{A path for a string operator on the hexagonal lattice.}
 \label{fig:stringOperator}
\end{figure}

The string operator $\sigma^s_p(\phi,\psi)$ is not defined on the entire Hilbert space of our lattice model, but only on the subspace where
\begin{enumerate}[label=(S\arabic*)]
 \item\label{S:Unexcited} 
 every lattice link along $p$, except perhaps the links at the endpoints of $p$, is unexcited, and
 \item\label{S:LabelledBy1}
 at the vertices of the initial and final links of $p$ which do not lie along $p$, the initial and final links are labelled by $1$.
\end{enumerate}
Condition \ref{S:Unexcited} means that string operators cannot pass through the locations of excitations of $A_\ell$ terms, although they may begin or end on such links.
Consequently, in the middle of a string operator, we may use we use the graphical calculus and \eqref{eq:fusionDecomp} to implement tensoring with $s$ and braiding over $s$ as linear combinations of operators on the individual vertex Hilbert spaces, as we did when defining $B_p^s$.
Consequently, when defining a string operator, we may take advantage of the graphical calculus away from the endpoints, using \eqref{eq:fusionDecomp} and the half-braiding of $s$ to rewrite the string operator as linear combinations of products of operators on the individual vertex Hilbert spaces, as we previously did for $B_p^s$.

Because the plaquette term permits deforming strings in $\mathcal{C}$ across plaquettes, if $p$ and $q$ are homotopic paths with the same endpoints, then $\sigma^s_p$ and $\sigma^s_q$ agree on ground states of our Hamiltonian \cite[Appendix C]{PhysRevB.71.045110}.
More generally, if $|\eta\rangle$ is a state containing localized excitations, then $\sigma^s_p(\phi,\psi)|\eta\rangle$ and $\sigma^s_q(\phi,\psi)|\eta\rangle$ differ by an application of certain terms of the half-braidings, depending on which excitations are crossed during a homotopy from $p$ to $q$; if the homotopy never passes through the location of an excitation, $\sigma^s_p(\phi,\psi)|\eta\rangle=\sigma^s_q(\phi,\psi)|\eta\rangle$.
The excitations created by these string operators are therefore topological, in the sense that the effect of moving excitations via string operators depends only on the topology of the movement, and not the exact path taken.

Condition \ref{S:LabelledBy1} means that on states where a string operator $\sigma^s_p(\phi,\psi)$ is nonzero, at the endpoints of $p$, we can view the two hexagonal plaquettes containing the final edge $\ell$ of $p$ as a single decagonal plaquette.
Applying $\sigma^s_p(\phi,\psi)$ then turns $\ell$ into an additional edge interior to the decagonal plaquette which supports the excitation, similar to the edges added to plaquettes in the extended Levin-Wen model of \cite{PhysRevB.97.195154}.
If $w$ is the final vertex along $p$, then basis states where $\sigma^s_p(\phi,\psi)$ is nonzero are those in the image of the projection
\[
\pi_{\ell,w}^1:
\mathcal{H}_w
\cong
\bigoplus_{a,b,c\in\Irr(\mathcal{X})}\mathcal{X}(ab\to c)\to\bigoplus_{a,c\in\Irr(\mathcal{X})}\mathcal{X}(a1\to c)
\]
onto the space of states where the morphism labeling $w$ assigns the simple $1$ to the edge $\ell$.
(See also Definition \ref{def:selectEdgeLabel}.)
However, as a convention, we define $\sigma^1_p=\id$, rather than as a product $\pi^1_{\ell,w}\pi^1_{m,v}$ to enforce condition \ref{S:LabelledBy1}.
This choice will later be justified by Lemma \ref{lem:restoration}.

At the initial vertex of $p$, we define $\sigma_p^s(\phi,\psi)$ by tensoring with $s$ along the link $\ell$ where $p$ begins, and composing with the morphism $\phi^\dag$.
At the final vertex while at the final vertex of $p$, we tensor with $s$, and then compose with the morphism $\psi$.
An example computation of a string operator in terms of operators local to each vertex appears in Figure \ref{fig:stringExpansion}.
If we let $\overline{p}$ denote the path $p$ with the orientation reversed, then by construction, our string operators satisfy
\begin{equation}
 \label{eq:stringOpReverse}
 \sigma_p^s(\phi,\psi)=\sigma_{\overline{p}}^{\overline{s}}\left(\overline{\psi},\overline{\phi}\right)\text{.}
\end{equation}

\begin{figure}
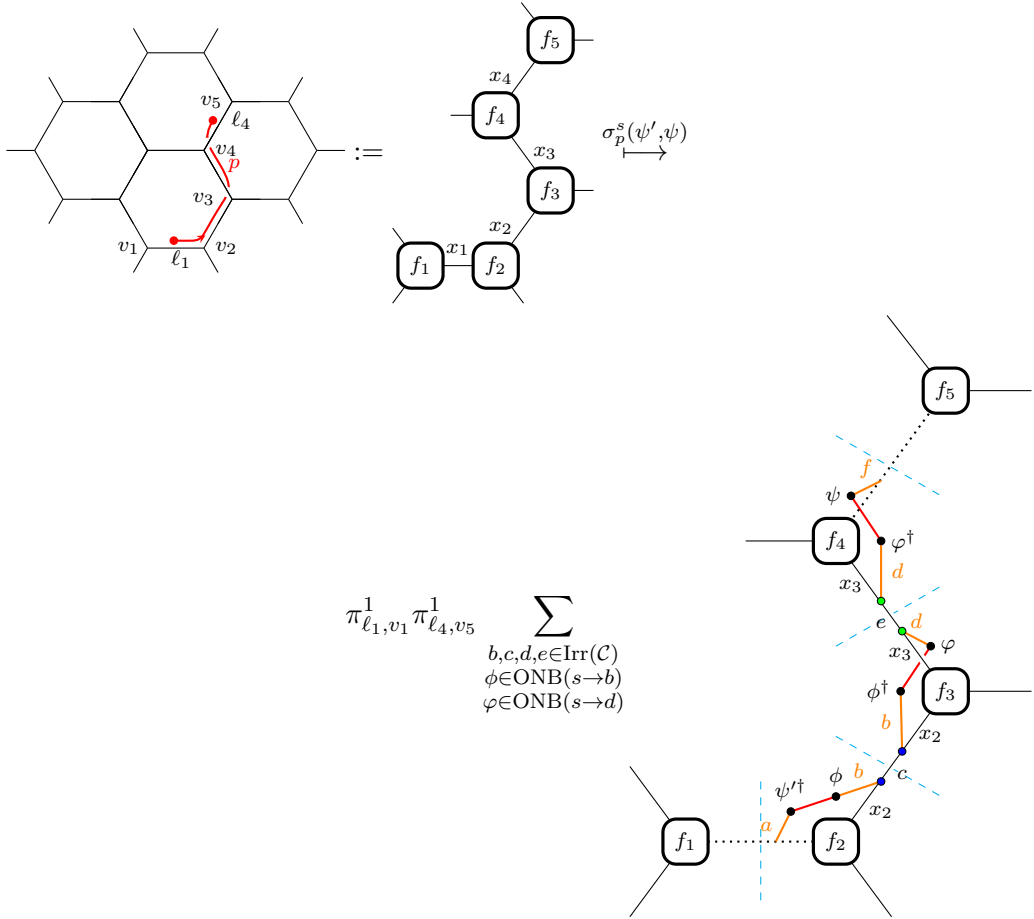

\begin{align*}
\tikzmath{
\coordinate (a) at (1.1,-1.2);
\coordinate (b) at (1.62,-.9);
\coordinate (c) at (1.72,-.2);
\coordinate (z) at (1.62,.4);
\filldraw[red] (a) circle (.05cm);
\filldraw[red] (z) circle (.05cm);
\draw[red, thick] (a) .. controls ++(0:.4cm) and ++(-120:.4cm) .. (b) node [sloped,allow upside down,pos=0.5] {\arrowIn[red]}
                            .. controls ++(60:.4cm) and ++(-60:.4cm) .. (c)
                            .. controls ++(120:.3cm) and ++(-120:.3cm) .. (z) node [sloped,allow upside down,pos=0.5] {\arrowIn[red]}
                            ;
\node[red] at (1.9,-.2) {$\scriptstyle p$};
\fill[white] (1.8,-.55) circle (.07cm);
\fill[white] (1.55,.1) circle (.07cm);
\levinHexGrid[]{0}{0}{2}{2}{.75}{black}
\node at (.5,-1.3) {$\scriptstyle v_1$};
\node at (1.2,-1.45) {$\scriptstyle \ell_1$};
\node at (1.8,-1.3) {$\scriptstyle v_2$};
\node at (1.5,-.65) {$\scriptstyle v_3$};
\node at (1.8,0) {$\scriptstyle v_4$};
\node at (2,.4) {$\scriptstyle \ell_4$};
\node at (1.6,.65) {$\scriptstyle v_5$};
}
&:=
\tikzmath{
\draw (0,0) --node[above]{$\scriptstyle x_1$} (1,0) -- node[left]{$\scriptstyle x_2$} (1.73,1) -- node[right]{$\scriptstyle x_3$} (1,2) -- node[left]{$\scriptstyle x_4$} (1.73,3);
\draw (-.37,-.5) -- (0,0);
\draw (-.37,.5) -- (0,0);
\draw (1.37,-.5) -- (1,0);
\draw (2.3,1) -- (1.73,1);
\draw (.4,2) -- (1,2);
\draw (2.3,3) -- (1.73,3);
\draw (1.34,3.5) -- (1.73,3);
\roundNbox{fill=white}{(0,0)}{.3}{0}{0}{$\scriptstyle f_1$}
\roundNbox{fill=white}{(1,0)}{.3}{0}{0}{$\scriptstyle f_2$}
\roundNbox{fill=white}{(1.73,1)}{.3}{0}{0}{$\scriptstyle f_3$}
\roundNbox{fill=white}{(1,2)}{.3}{0}{0}{$\scriptstyle f_4$}
\roundNbox{fill=white}{(1.73,3)}{.3}{0}{0}{$\scriptstyle f_5$}
}
\overset{\sigma^s_p(\psi',\psi)}{\longmapsto}
\\
&
 \pi^1_{\ell_1,v_1}\pi^1_{\ell_4,v_5}
\sum_{\substack{
b,c,d,e\in \Irr(\cC)
\\
\phi \in \ONB(s\to b)
\\
\varphi \in \ONB(s\to d)
}}
\tikzmath[scale=2]{
\draw[dashed, cyan] (.5,.4) -- (.5,-.4);
\draw[dashed, cyan] (1,.7) -- (1.7,.3);
\draw[dashed, cyan] (1,1.3) -- (1.7,1.7);
\draw[dashed, cyan] (1,2.7) -- (1.7,2.3);
\coordinate (begin) at (.6,0);
\coordinate (b1) at (1.3,.4);
\coordinate (b2) at (1.44,.6);
\coordinate (g1) at (1.44,1.4);
\coordinate (g2) at (1.3,1.6);
\coordinate (phi1) at (1,.3);
\coordinate (phi2) at (1.43,1);
\coordinate (phi3) at (1.63,1.3);
\coordinate (phi4) at (1.3,2);
\coordinate (psi1) at (.7,.2);
\coordinate (psi2) at (1.1,2.3);
\coordinate (end) at (1.3,2.4);
\draw[thick, red] (psi1) -- (phi1);
\draw[thick, red] (phi2) -- (phi3);
\draw[thick, red] (phi4) -- (psi2);
\draw[dotted, thick] (0,0) -- (1,0);
\draw[knot] (1,0) -- (1.73,1) -- (1,2);
\draw[dotted, thick] (1,2) -- (1.73,3);
\draw (-.37,-.5) -- (0,0);
\draw (-.37,.5) -- (0,0);
\draw (1.37,-.5) -- (1,0);
\draw (2.3,1) -- (1.73,1);
\draw (.4,2) -- (1,2);
\draw (2.3,3) -- (1.73,3);
\draw (1.34,3.5) -- (1.73,3);
\draw[thick, orange] (psi1) -- node[left]{$\scriptstyle a$} (begin) ;
\draw[thick, orange] (phi1) -- node[above]{$\scriptstyle b$} (b1);
\draw[thick, orange] (phi2) -- node[left]{$\scriptstyle b$} (b2);
\draw[thick, orange] (phi3) -- node[above]{$\scriptstyle d$} (g1);
\draw[thick, orange] (phi4) -- node[right]{$\scriptstyle d$} (g2);
\draw[thick, orange] (psi2) -- node[above]{$\scriptstyle f$} (end);
\filldraw[fill=blue] (b1) circle (.025cm);
\filldraw[fill=blue] (b2) circle (.025cm);
\filldraw[fill=green] (g1) circle (.025cm);
\filldraw[fill=green] (g2) circle (.025cm);
\filldraw (phi1) circle (.025cm) node[above]{$\scriptstyle \phi$};
\filldraw (phi2) circle (.025cm) node[left]{$\scriptstyle \phi^\dag$};
\filldraw (phi3) circle (.025cm) node[right]{$\scriptstyle \varphi$};
\filldraw (phi4) circle (.025cm) node[right]{$\scriptstyle \varphi^\dag$};
\filldraw (psi1) circle (.025cm) node[above]{$\scriptstyle \psi'^\dag$};
\filldraw (psi2) circle (.025cm) node[left]{$\scriptstyle \psi$};
\roundNbox{fill=white}{(0,0)}{.15}{0}{0}{$\scriptstyle f_1$}
\roundNbox{fill=white}{(1,0)}{.15}{0}{0}{$\scriptstyle f_2$}
\roundNbox{fill=white}{(1.73,1)}{.15}{0}{0}{$\scriptstyle f_3$}
\roundNbox{fill=white}{(1,2)}{.15}{0}{0}{$\scriptstyle f_4$}
\roundNbox{fill=white}{(1.73,3)}{.15}{0}{0}{$\scriptstyle f_5$}
\node at (1.3,.2) {$\scriptstyle x_2$};
\node at (1.44,.45) {$\scriptstyle c$};
\node at (1.63,.7) {$\scriptstyle x_2$};
\node at (1.43,1.25) {$\scriptstyle x_3$};
\node at (1.3,1.45) {$\scriptstyle e$};
\node at (1.1,1.7) {$\scriptstyle x_3$};
}
\end{align*}
 \caption{
  \label{fig:stringExpansion}
  A string operator is resolved into local operators on individual vertex Hilbert spaces.
  The morphism $f_i$ labels the vertex $v_i$, and the object $x_i$ labels the link $\ell_i$.
  The effect of applying $\pi^1_{\ell_1,v_1}$ is to ensure $x_1=1$, and the effect of applying $\pi^1_{\ell_4,v_5}$ is to ensure $x_4=1$.
  We represent the object 1 on the right hand side by dotted edges.
  The entire edge is dotted, rather than just the part parallel to the string, because a string operator ending at $\ell$ is only defined on the ground state of $A_\ell$.
 }
\end{figure}

\begin{rem}
 Condition \ref{S:LabelledBy1} may appear unnatural, especially since, without redefining $\sigma^1_p$ by special case, string operators $\sigma^1_p$ corresponding to the vacuum would actually excite ground states of our model.
 This is not as bad as it seems, because Lemma \ref{lem:restoration} will show that in situations where the removed edge $\ell$ of the plaquette $q$ where a string operator terminates does not host an excitation, applying $B_r$, where $r$ is the plaquette which borders $q$ along $\ell$, will undo the effects of $\pi^1_{\ell,w}$ and restore the ground state.
 Aside from states obtained from string operators $\sigma^1$, this can occur when two string operators have created anti-particles at the same location, leaving the vacuum as one of the fusion channels.
 The benefit of imposing condition \ref{S:LabelledBy1} is that we can apply the results of \cite{PhysRevB.97.195154}.

 This approach is not the only option, though.
 One could instead extend string operators to the whole ground space of $A_\ell$ terms, moving the applications of $\pi_1$ and $B_p$ to definitions of hopping operators and tube-algebra representations.
 Alternatively, one could view string operators as actually changing the lattice.
 The results of \cite{PhysRevB.85.075107}, Lemma \ref{lem:restoration}, and indeed, the original conceptual description of the string-net in \cite{PhysRevB.71.045110}, support the point of view that the particular choice of lattice is not important, as does the fact that the topological field theory describing the low-energy behavior of the string-net model is topological, i.e., depending only the choice of manifold.
 We are so strict about working in a single consistent Hilbert space only because the results Proposition \ref{prop:tubeActionCorrect} and Theorem \ref{thm:correctMTC} are part of the work of checking that our lattice model for anyon condensation realizes the expected topological phases.
\end{rem}

\begin{rem}
 \label{rem:stringOperatorDesign}
Many articles, such as \cite{PhysRevB.97.195154,PhysRevB.103.195155},
only consider $\sigma^s_p(\psi,\psi)$ where $\psi$ is the sum over an orthonormal basis.
One advantage of our approach is that in \S\ref{ssec:tubeImplementation} below, we will recover the entire $\Tube(\mathcal{X})$ representation from a choice of $X\in\Irr(Z(\mathcal{X}))$, rather than simply computing which minimal central projection preserves a certain state.
 This will demonstrate that any excited state containing finitely many excitations which are separated from one another can be achieved via linear combinations of string operators.
\end{rem}

\begin{rem}
 \label{rem:nonsimpleStringOperators}
 In defining the operator $\sigma^s_p(\phi,\psi)$, we do not make use of the fact that $s\in Z(\mathcal{X})$ is a simple object.
 However, there are some reasons to do so.
 First, if $s$ is simple and $p$ crosses any links of the lattice, then $\sigma^s_p(\phi,\psi)$ is always nonzero.
 This can be shown using the tube algebra techniques introduced in \S~\ref{ssec:tubeImplementation} below.
 If $\sigma^s_p(\phi,\psi)=0$, then by an application of Proposition \ref{prop:tubeActionCorrect}, $\phi$ and $\psi$ would lie in orthogonal summands of the tube algebra representation $\rho_s$ introduced in \eqref{eq:tubeRep}, showing that $s$ is decomposable in $Z(\mathcal{X})$.

 Also, if some isotypic component in $Z(\mathcal{X})$ of $s$ is not simple, e.g. $s\cong t\oplus t$ where $t\in\Irr(Z(\mathcal{X}))$, then different choices of $(\phi,\psi)$ will produce the same string operator.
 For this reason, the hopping operators defined below in \S\ref{sssec:hopping} do not make sense for such $s$, and we only define them for simple objects in $Z(\mathcal{X})$.
\end{rem}

\subsubsection{String Operators on Excited States}
\label{sssec:excitedStringOperators}
As defined above, the string operators $\sigma^s_p(\phi,\psi)$ only make sense on states which are locally the ground state near $p$.
We can also define string operators on states where the one or both endpoints of $p$ host an anyon.
Our construction will make use of the local tube algebra action described in \S\ref{ssec:tubeImplementation} below,
but the use of the tube algebra action reduces to the proofs of Corollary \ref{cor:tubeKetbra}, which the reader may presently treat as a black box.
The statement of Corollary \ref{cor:tubeKetbra} is technical, but the upshot is that there are local operators which can detect and change the choice of $\psi$ in $\sigma_p(\psi',\psi)$.
Indeed, since anyon types are by definition superselection sectors under the action of such local operators \cite{cond-mat/0506438}, Corollary \ref{cor:tubeKetbra} must hold so long as we have correctly identified the UMTC $Z(\cX)$ and string operators $\sigma^s$ associated to each $s\in\Irr(Z(\cX))$.
One should therefore view the Corollary as justifying the definition of $\sigma^s$.

\begin{cor*}[\ref{cor:tubeKetbra}]
 For any link $\ell$, vertex $v$ of $\ell$, anyon $s\in\Irr(Z(\mathcal{X})$, and morphisms $\phi,\psi\in\bigoplus_x\mathcal{X}(s\to x)$, there is a local operator $T^s_{\ell,v}(\phi,\psi)$ such that, if $|\omega\rangle=A_\ell|\omega\rangle=B_r|\omega\rangle$ for plaquettes $r$ containing $\ell$, then
 \[
  T^s_{\ell,v}(\phi,\psi)\sigma^s_p(\eta',\eta)|\omega\rangle=\langle\phi|\eta\rangle\sigma^s_p(\eta',\psi)|\omega\rangle\text{,}
 \]
 and if $t\neq s$,
 \[T^s_{\ell,v}(\phi,\psi)\sigma^t_p(\eta',\eta)|\omega\rangle=0\text{.}\]
\end{cor*}

Based on our construction of string operators, the location of an excitation in a state $|\phi\rangle$ actually consists of three pieces of information: a link $\ell$, a plaquette $r$ containing $\ell$, and a vertex $v$ of $\ell$ such that $\pi^1_{\ell,v}|\phi\rangle=|\phi\rangle$.
Suppose that $p$ is the potential path of a string operator, as in Figure \ref{fig:stringOperator}, with endpoints $(\ell_0,r_0,v_0)$ and $(\ell_1,r_1,v_1)$.
Suppose that $|\phi\rangle$ is a state containing excitations of types $t_0$ and $t_1$ at the endpoints of $p$, so that $\pi_{\ell_0,v_0}^1|\phi\rangle$ and $\pi_{\ell_1,v_1}^1|\phi\rangle$ are nonzero, and no other excitations near $p$.
Given morphisms $\psi_0\in\mathcal{X}(t_0s\to x)$ and $\psi_1\in\mathcal{X}(t_1s\to y)$, where $x,y\in\Irr(\mathcal{X})$, we can define a string operator
\[\sigma^s_p[t_0,t_1](\psi_0,\psi_1)\]
which creates excitations of types $\overline{s}$ and $s$ at the two endpoints.
In the middle of $p$, the definition is the same as for $\sigma^s_p$, and is once again independent of the choices of $t_i$ and $\psi_i$; at the endpoints of $p$, we use the chosen $t_i$ and $\psi_i$ to end the string on the lattice.
Explicitly, we define
\[\sigma^s_p[t_0,t_1](\psi_0,\psi_1)=\sum_{\substack{x,y\in\Irr(\mathcal{X})\\\eta_0\in\ONB(t_0,x)\\\eta_1\in\ONB(t_1,y)}}S^s_p[\eta_0,\eta_1](\psi_0,\psi_1)T^{t_0}_{\ell_0,v_0}(\eta_0,\eta_0)T^{t_1}_{\ell_1,v_1}(\eta_1,\eta_1)\text{,}\]
where the $T$ operators are the local operators from Corollary \ref{cor:tubeKetbra},
and an explicit description of $S$ appears in Figure \ref{fig:excitedStringExpansion}.

\begin{figure}
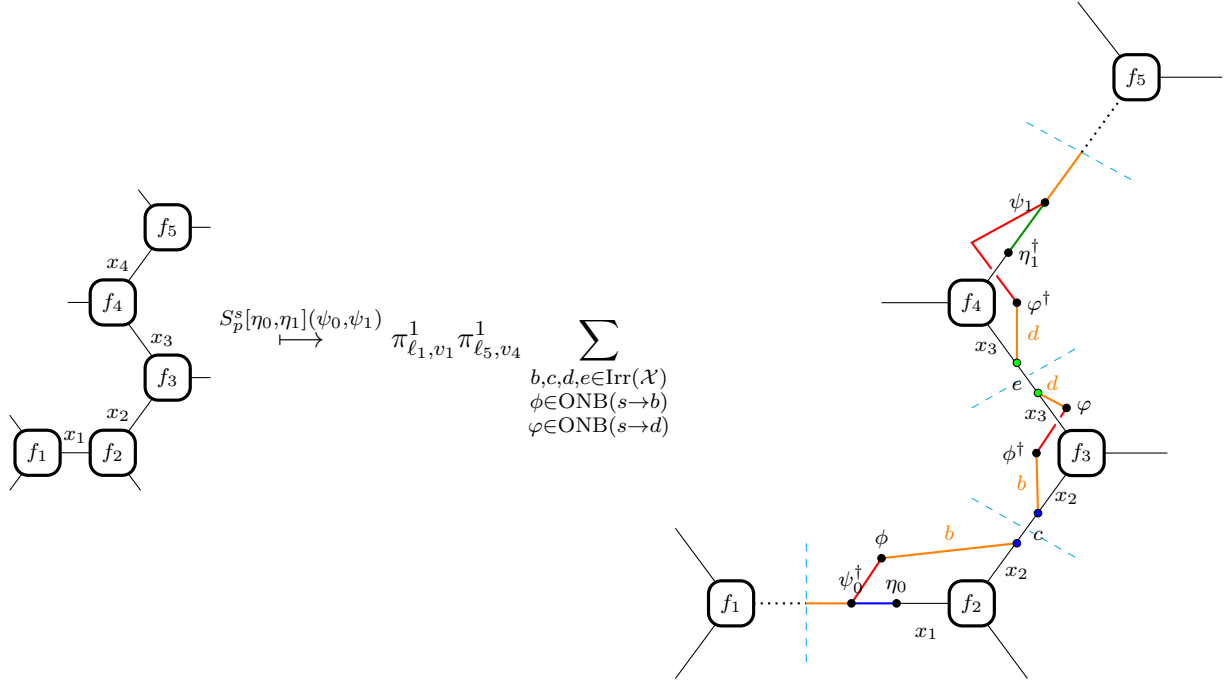

 \[
  \tikzmath{
   \draw (0,0) --node[above]{$\scriptstyle x_1$} (1,0) -- node[left]{$\scriptstyle x_2$} (1.73,1) -- node[right]{$\scriptstyle x_3$} (1,2) -- node[left]{$\scriptstyle x_4$} (1.73,3);
   \draw (-.37,-.5) -- (0,0);
   \draw (-.37,.5) -- (0,0);
   \draw (1.37,-.5) -- (1,0);
   \draw (2.3,1) -- (1.73,1);
   \draw (.4,2) -- (1,2);
   \draw (2.3,3) -- (1.73,3);
   \draw (1.34,3.5) -- (1.73,3);
   \roundNbox{fill=white}{(0,0)}{.3}{0}{0}{$\scriptstyle f_1$}
   \roundNbox{fill=white}{(1,0)}{.3}{0}{0}{$\scriptstyle f_2$}
   \roundNbox{fill=white}{(1.73,1)}{.3}{0}{0}{$\scriptstyle f_3$}
   \roundNbox{fill=white}{(1,2)}{.3}{0}{0}{$\scriptstyle f_4$}
   \roundNbox{fill=white}{(1.73,3)}{.3}{0}{0}{$\scriptstyle f_5$}
  }
  \overset{S^s_p[\eta_0,\eta_1](\psi_0,\psi_1)}{\longmapsto}\pi^1_{\ell_1,v_1}\pi^1_{\ell_5,v_4}
  \sum_{\substack{
  b,c,d,e\in \Irr(\cX)
  \\
  \phi \in \ONB(s\to b)
  \\
  \varphi \in \ONB(s\to d)
  }}
  \tikzmath[scale=2]{
   \draw[dashed, cyan] (-.1,.4) -- (-.1,-.4);
   \draw[dashed, cyan] (1,.7) -- (1.7,.3);
   \draw[dashed, cyan] (1,1.3) -- (1.7,1.7);
   \draw[dashed, cyan] (1.365,3.2) -- (2.095,2.8);
   \coordinate (begin) at (.6,0);
   \coordinate (b1) at (1.3,.4);
   \coordinate (b2) at (1.44,.6);
   \coordinate (g1) at (1.44,1.4);
   \coordinate (g2) at (1.3,1.6);
   \coordinate (phi1) at (.4,.3);
   \coordinate (phi2) at (1.43,1);
   \coordinate (phi3) at (1.63,1.3);
   \coordinate (phi4) at (1.3,2);
   \coordinate (psi0) at (.2,0);
   \coordinate (psiP) at (1.0,2.4);
   \coordinate (psi1) at (1.4867,2.6667);
   \coordinate (eta0) at (.5,0);
   \coordinate (eta1) at (1.2434,2.3333);
   \coordinate (end) at (1.3,2.4);
   \draw[thick, red] (psi0) -- (phi1);
   \draw[thick, red] (phi2) -- (phi3);
   \draw[thick, red] (phi4) -- (psiP) -- (psi1);
   %
   \draw[dotted, thick] (-.6,0) -- (-.1,0);
   \draw[thick, orange] (-.1,0) -- (psi0);
   \draw[thick, blue] (psi0) -- (eta0);
   \draw (eta0) -- (1,0);
   \draw[knot] (1,0) -- (1.73,1) -- (1,2);
   \draw[knot] (1,2) -- (eta1);
   \draw[dotted, thick] (2.46,4);
   \draw (-.97,-.5) -- (-.6,0);
   \draw (-.97,.5) -- (-.6,0);
   \draw (1.37,-.5) -- (1,0);
   \draw (2.3,1) -- (1.73,1);
   \draw (.4,2) -- (1,2);
   \draw (2.665,3.5) -- (2.095,3.5);
   \draw (1.705,4) -- (2.095,3.5);
   \draw[thick,DarkGreen] (eta1) -- (psi1);
   \draw[thick, orange] (psi1) -- (1.73,3);
   \draw[dotted, thick] (1.73,3) -- (2.095,3.5);
   %
   \draw[thick, orange] (phi1) -- node[above]{$\scriptstyle b$} (b1);
   \draw[thick, orange] (phi2) -- node[left]{$\scriptstyle b$} (b2);
   \draw[thick, orange] (phi3) -- node[above]{$\scriptstyle d$} (g1);
   \draw[thick, orange] (phi4) -- node[right]{$\scriptstyle d$} (g2);
   %
   \filldraw[fill=blue] (b1) circle (.025cm);
   \filldraw[fill=blue] (b2) circle (.025cm);
   \filldraw[fill=green] (g1) circle (.025cm);
   \filldraw[fill=green] (g2) circle (.025cm);
   \filldraw (phi1) circle (.025cm) node[above]{$\scriptstyle \phi$};
   \filldraw (phi2) circle (.025cm) node[left]{$\scriptstyle \phi^\dag$};
   \filldraw (phi3) circle (.025cm) node[right]{$\scriptstyle \varphi$};
   \filldraw (phi4) circle (.025cm) node[right]{$\scriptstyle \varphi^\dag$};
   \filldraw (psi0) circle (.025cm) node[above]{$\scriptstyle \psi_0^\dag$};
   \filldraw (psi1) circle (.025cm) node[left]{$\scriptstyle \psi_1$};
   \filldraw (eta0) circle (.025cm) node[above]{$\scriptstyle \eta_0$};
   \filldraw (eta1) circle (.025cm) node[right]{$\scriptstyle \eta_1^\dag$};
   \roundNbox{fill=white}{(-.6,0)}{.15}{0}{0}{$\scriptstyle f_1$}
   \roundNbox{fill=white}{(1,0)}{.15}{0}{0}{$\scriptstyle f_2$}
   \roundNbox{fill=white}{(1.73,1)}{.15}{0}{0}{$\scriptstyle f_3$}
   \roundNbox{fill=white}{(1,2)}{.15}{0}{0}{$\scriptstyle f_4$}
   \roundNbox{fill=white}{(2.095,3.5)}{.15}{0}{0}{$\scriptstyle f_5$}
   \node at (.7,-.2) {$\scriptstyle x_1$};
   \node at (1.3,.2) {$\scriptstyle x_2$};
   \node at (1.44,.45) {$\scriptstyle c$};
   \node at (1.63,.7) {$\scriptstyle x_2$};
   \node at (1.43,1.25) {$\scriptstyle x_3$};
   \node at (1.3,1.45) {$\scriptstyle e$};
   \node at (1.1,1.7) {$\scriptstyle x_3$};
  }
 \]
 \caption{
  \label{fig:excitedStringExpansion}
  A string operator defined on excited states is resolved into local operators on individual vertex Hilbert spaces, as in Figure \ref{fig:stringExpansion}.
 }
\end{figure}

When $t_0=t_1=1$, we recover the string operators of the previous section.
Since the local operators $T$ also generate projections onto states with a particular type of anyon at each endpoint of $p$, namely
\[\sum_{\phi\in\ONB(s,x)}T_{\ell,v}^s(\phi,\phi)\text{,}\]
we can also define string operators for anyons of type $s$ which are applicable regardless of the type of anyon at each endpoint.
A general string operator of type $s$ is thus of the form
\[\sum_{t_0,t_1}\sigma^s_p[t_0,t_1](\psi_{t_0},\psi'_{t_1})\text{.}\]

Now that we have defined string operators on excited states, our string operators can fuse, according to the fusion of anyons in $Z(\cX)$.
That is, if $r,s,t\in\Irr(Z(\cX))$, $x_i\in\Irr(\cX)$, $\psi:rs\to t$, and $|\Omega\rangle$ is a ground state, then
\begin{equation}
 \label{eq:stringFusion}
 \sigma^r_p[s,s](\xi_0\circ\psi,\xi_1\circ\psi)\sigma^s_p(\eta_0,\eta_1)|\Omega\rangle=\sigma^t_p(\xi_0,\xi_1)|\Omega\rangle
\end{equation}
The strings fuse in the middle by associativity of the tensor product on $Z(\cX)$, analogous to the proof that $B_p$ is an idempotent, and it is straightforward to check the equalities at the endpoints of $p$.

\subsubsection{Hopping Operators}
\label{sssec:hopping}
Related to the string operators which create topological excitations are hopping operators \cite[\S\RN{5}.E]{PhysRevB.97.195154}, which move an existing topological excitation from one location to another.
In this section, for each $s\in\Irr(Z(\mathcal{X}))$ and each path $q$, we will define a hopping operator $h^s_q$ which sends states with an excitation of type $s$ at the initial location of $q$ to those with an excitation of type $s$ at the terminal location.
In other words, if $|\omega\rangle$ is a ground state, then $h^s_q\sigma^t_p(\psi',\psi)|\omega\rangle=0$ unless $s\cong t$ and the terminal link of $p$ is the initial link of $q$ (or $s\cong\overline{t}$ and $p$ and $q$ have the same terminal link, or either $s$ or $t$ is $1$), and $p$ and $q$ approach that link from the same vertex.
Moreover, our hopping operators will satisfy the relation
\begin{equation}
 \label{eq:hopping}
 h^s_q\sigma^s_p(\psi',\psi)=\sigma^s_{q\cdot p}(\psi',\psi)
\end{equation}
where $q\cdot p$ is the concatenation of the paths $q$ and $p$.

Hopping operators can be built up from the general string operators discussed in the previous section.
In order to do so, we make the following observation, which is an explicit description of the ``contraction of charges'' of \cite{PhysRevB.97.195154} in our setting.
\begin{lem}
 \label{lem:contractionOfCharges}
 Suppose $p$ and $q$ are two paths, such that $p$ ends where $q$ begins, adjacent to the plaquette $r$.
 Then
 \[\sum_{\phi\in B}B_r\sigma^s_p[1,\overline{s}](\psi,\ev)\sigma^s_q(\phi,\eta)=\sigma_{q\cdot p}(\psi,\eta)\text{,}\]
 where $B$ is an orthonormal basis of $\bigoplus_x\mathcal{X}(s\to x)$ and $\ev:\overline{s}s\to 1$ is part of the duality data of $\mathcal{X}$.
\end{lem}
\begin{proof}
 Let $\ell:v\to w$ be the link where $p$ ends and $q$ begins.
 Abbreviate $\pi^1_{\ell}:=\pi^1_{\ell,v}\pi^1_{\ell,w}$.
 From the definition of $\sigma$ and equation \eqref{eq:stringOpReverse}, we can compute that
 \[\pi^1_{\ell}\sigma^s_{q\cdot p}(\psi,\eta)=\pi^1_{\ell}\sum_{\phi\in B}\sigma^s_p[1,\overline{s}](\psi,\ev)\sigma^s_q(\phi,\eta)\text{.}\]

 Explicitly,
 \begin{align*}
\sigma^s_q(\phi,\eta)|\omega\rangle
&=
\tikzmath[scale=2]{
\coordinate (a) at (-50:.6cm);
\coordinate (b) at (-70:.6cm);
\coordinate (c) at (-110:.85cm);
\coordinate (d) at (-135:1.1cm);
\draw[dashed] (-60:.75) -- (0:.75);
\draw[red, thick] (a) .. controls ++(-150:.1cm) and ++(0:.1cm) .. (b) .. controls ++(180:.2cm) and ++(0:.2cm) .. (c) node[right, yshift=-.1cm]{$s$} node [sloped,allow upside down,pos=0.5] {\arrowIn[red]} -- node [sloped,allow upside down,pos=0.5] {\arrowIn[red]} (d);
\draw[orange, thick] (a) -- node[left, yshift=.15cm, xshift=.15cm]{$x$} ($ .5*(-60:.75) + .5*(0:.75) $);
\filldraw (a) node[left]{$\overline{\phi^\dag}$} circle (.025cm);
\fill[white] (-90:.65) circle (.04cm);
\fill[white] (-120:.9) circle (.04cm);
\levinHexOpen{0}{0}{.75}{black}{6}
\levinHexOpen{1.125}{-.6495}{.75}{black}{3}
\node at ($ .25*(-60:.75) + .75*(0:.75) + (.1,0)$) {$\ell$};
\node at (0:0) {$r$};
}
\\\displaybreak[1]
\pi^1_\ell\sigma_p^s[1,\overline{s}](\psi,\ev)\sigma_q^s(\phi,\eta)|\omega\rangle
&=
\tikzmath[scale=2]{
\coordinate (a) at (-50:.6cm);
\coordinate (b) at (-70:.6cm);
\coordinate (c) at (-110:.85cm);
\coordinate (d) at (-135:1.1cm);
\coordinate (x) at (.6, -.5);
\draw[dashed] (-60:.75) -- (0:.75);
\draw[red, thick] (b) .. controls ++(180:.2cm) and ++(0:.2cm) .. (c) node[right, yshift=-.1cm]{$s$} node [sloped,allow upside down,pos=0.5] {\arrowIn[red]} -- node [sloped,allow upside down,pos=0.5] {\arrowIn[red]} (d);
\draw[red, thick] (a) .. controls ++(30:.1cm) and ++(80:.2cm) .. (x);
\draw[orange, thick] (a) .. controls ++(-150:.1cm) and ++(0:.1cm) .. (b);
\filldraw (b) node[above]{$\overline{\phi^\dag}$} circle (.025cm);
\filldraw (a) node[above]{$\overline{\phi}$} circle (.025cm);
\fill[white] (-90:.65) circle (.04cm);
\fill[white] (-120:.9) circle (.04cm);
\coordinate (y) at (.55, -.8);
\coordinate (z) at (1.15, -1.2);
\draw[red, thick] (x) .. controls ++(-100:.1cm) and ++(120:.1cm) .. node [sloped,allow upside down,pos=0.5] {\arrowInR[red]} (y) .. controls ++(-60:.3cm) and ++(180:.5cm) .. (z) node[above]{$s$} -- node [sloped,allow upside down,pos=0.5] {\arrowInR[red]} ($ (z) + (1,0) $);
\fill[white] ($ (z) + (.4,0) $) circle (.04cm);
\levinHexOpen{0}{0}{.75}{black}{6}
\levinHexOpen{1.125}{-.6495}{.75}{black}{3}
\node at ($ .25*(-60:.75) + .75*(0:.75) + (.1,0)$) {$\ell$};
\node at (0:0) {$r$};
}
\\\displaybreak[1]
\pi^1_{\ell}\sum_{\phi\in B}\sigma_p^s[1,\overline{s}](\psi,\ev)\sigma_q^s(\phi,\eta)|\omega\rangle
&=
\tikzmath[scale=2]{
\coordinate (a) at (-50:.6cm);
\coordinate (b) at (-70:.6cm);
\coordinate (c) at (-110:.85cm);
\coordinate (d) at (-135:1.1cm);
\coordinate (x) at (.6, -.5);
\coordinate (y) at (.55, -.8);
\coordinate (z) at (1.15, -1.2);
\draw[dashed] (-60:.75) -- (0:.75);
\draw[red, thick] (b) .. controls ++(180:.2cm) and ++(0:.2cm) .. node [sloped,allow upside down,pos=0.5] {\arrowIn[red]} (c) node[right, yshift=-.1cm]{$s$} -- node [sloped,allow upside down,pos=0.5] {\arrowIn[red]} (d);
\draw[red, thick] (b) .. controls ++(0:.3cm) and ++(120:.3cm) .. (y) .. node [sloped,allow upside down,pos=0.5] {\arrowInR[red]} controls ++(-60:.3cm) and ++(180:.5cm) .. (z) node[above]{$s$} -- node [sloped,allow upside down,pos=0.5] {\arrowInR[red]} ($ (z) + (1,0) $);
\fill[white] (-90:.65) circle (.04cm);
\fill[white] (-120:.9) circle (.04cm);
\fill[white] ($ (z) + (.4,0) $) circle (.04cm);
\levinHexOpen{0}{0}{.75}{black}{6}
\levinHexOpen{1.125}{-.6495}{.75}{black}{3}
\node at ($ .25*(-60:.75) + .75*(0:.75) + (.1,0)$) {$\ell$};
\node at (0:0) {$r$};
}
\end{align*}
Since $\sigma^s_{q\cdot p}$ does not create an excitation at $\ell$, we need only apply Lemma \ref{lem:restoration}.
\end{proof}

With these ingredients in hand, we define
\[h^s_q=\frac{1}{|B|}B_r\sum_{\phi,\eta\in B}\sigma^s_q[s,1](\coev^\dag,\eta)T^s_{\ell,v}(\eta,\phi)\text{.}\]
Observe that
\begin{align*}
 h^s_q\sigma^s_p(\psi',\psi) &= \frac{1}{|B|}B_r\sum_{\phi,\eta}\sigma^s_q[s,1](\coev^\dag,\eta)T^s_{\ell,v}(\eta,\phi)\sigma^s_p(\psi',\psi)\\
 &= \frac{1}{|B|}\sum_{\phi,\eta}\langle\eta|\psi\rangle B_r\sigma^s_q(\coev^\dag,\eta)\sigma^s_p(\psi',\phi)\text{,}\\\intertext{and applying Lemma \ref{lem:contractionOfCharges},}
 h^s_q\sigma^s_p(\psi',\psi) &= \frac{1}{|B|}\sum_{\phi,\eta}\langle\eta|\psi\rangle \sigma^s_{q\cdot p}(\psi',\eta)\\
 h^s_q\sigma^s_p(\psi',\psi) &= \sum_{\eta}\sigma^s_{q\cdot p}\left(\psi',\langle\eta|\psi\rangle\eta\right)\\
 h^s_q\sigma^s_p(\psi',\psi) &= \sigma^s_{q\cdot p}(\psi',\psi)\text{,}
\end{align*}
verifying Equation \eqref{eq:hopping}.
Thus, $h^s_q$ also transports the local information $\psi$ at the end of the string.
This observation is formalized as Corollary \ref{cor:hoppingIntertwiner}.

\subsection{Tube Algebra Representations from Excitations}
\label{ssec:tubeImplementation}
We will now describe the correspondence between localized excitations in the Levin-Wen model and representations of the tube algebra $\Tube(\mathcal{X})$ (defined below), by implementing $\Tube(\mathcal{X})$ as an algebra of local operators at the location of an excitation in the lattice model.
As shown in \cite{MR1782145,MR1966525}, representations of $\Tube(\mathcal{X})$correspond to objects in $Z(\mathcal{X})$, so this gives a direct means of assigning an object in $Z(\mathcal{X})$ to an localized excitation.
The correspondence between localized excitations and $Z(\mathcal{X})$ is well known, and the type of an excitation can also be determined by other means, such as braiding experiments \cite[\S~III.C]{MR2443722}.
The use of tube algebras to classify excitations in various models of topological phases has been described in \cite{MR3614057,PhysRevB.90.115119,PhysRevLett.124.120601,2110.06079,1709.01941,MR4051062}.
However, we are interested in giving a concrete description of a local tube algebra action and the relationship to string operators in the case of Levin-Wen models, because in Section \ref{ssec:condensedPhase}, we will repeat the process with an appropriate variant of the tube algebra in order to prove Theorem \ref{thm:correctMTC}.

\begin{defn}[{\cite[\S3]{MR1782145},\cite[Def 3.3]{MR3447719}}]
 \label{def:tubeAlgebra}
 Let $\mathcal{X}$ be a unitary fusion category.
 The \emph{tube algebra} $\Tube(\mathcal{X})$ of $\mathcal{X}$ has the underlying Hilbert space $\bigoplus_{x,y,c\in\Irr(\mathcal{X})}\mathcal{X}(xc\to cy)$.
 If $\phi:xc\to cy$ and $\psi:zd\to dw$, the product $\psi\cdot\phi$ is defined to be the linear extension
$$
\tikzmath[xscale=-1]{
\draw (-.2,-.7) node[right, yshift=.2cm]{$\scriptstyle z$} -- (-.2,0);
\draw (.2,-.7) node[left, yshift=.2cm]{$\scriptstyle d$} -- (.2,0);
\draw (-.2,.7) node[right, yshift=-.2cm]{$\scriptstyle d$} -- (-.2,0);
\draw (.2,.7) node[left, yshift=-.2cm]{$\scriptstyle w$} -- (.2,0);
\roundNbox{fill=white}{(0,0)}{.3}{.2}{.2}{$\psi$}
}
\cdot
\tikzmath[xscale=-1]{
\draw (-.2,-.7) node[right, yshift=.2cm]{$\scriptstyle x$} -- (-.2,0);
\draw (.2,-.7) node[left, yshift=.2cm]{$\scriptstyle c$} -- (.2,0);
\draw (-.2,.7) node[right, yshift=-.2cm]{$\scriptstyle c$} -- (-.2,0);
\draw (.2,.7) node[left, yshift=-.2cm]{$\scriptstyle y$} -- (.2,0);
\roundNbox{fill=white}{(0,0)}{.3}{.2}{.2}{$\phi$}
}
:=
\delta_{y=z}
\sum_{f\in \Irr(\cX)}
\tikzmath[xscale=-1]{
\draw (-.17,-.2) to[rounded corners=5] (-.45,.1) to[sharp corners] (-.45,1.1) -- (-.55,1.4);
\draw (-.15,.8) to[bend right=25] (-.45,1.1);
\draw (.17,.2) to[rounded corners=5] (.45,-.1) to[sharp corners] (.45,-1.1) -- (.55,-1.4);
\draw (.15,-.8) to[bend right=25] (.45,-1.1);
\draw[fill=red] (-.45,1.1) circle (.05cm);
\draw[fill=red] (.45,-1.1) circle (.05cm);
	\draw (0,-1.4) -- node[left, xshift=.1cm]{$\scriptstyle y$} (0,1.4);
\roundNbox{fill=white}{(0,.5)}{.3}{0}{0}{$\psi$};
\roundNbox{fill=white}{(0,-.5)}{.3}{0}{0}{$\phi$};
	\node at (-.58,1.6) {$\scriptstyle f$};
	\node at (.58,-1.6) {$\scriptstyle f$};
	\node at (-.25,1.25) {$\scriptstyle d$};
	\node at (.25,-1.25) {$\scriptstyle c$};
	\node at (-.58,.5) {$\scriptstyle c$};
	\node at (.58,-.5) {$\scriptstyle d$};
	\node at (0,1.6) {$\scriptstyle w$};
	\node at (0,-1.6) {$\scriptstyle x$};
}\,\text{,}
$$
where $\tikzmath{\draw[fill=red] (0,0) circle (.05cm);}$ runs over an orthonormal basis of $\bigoplus_f\mathcal{X}(f\to dc)$.
We define a $*$-structure on $\Tube(\mathcal{X})$ by
$$
\left(
\tikzmath[xscale=-1]{
\draw (-.2,-.7) node[right, yshift=.2cm]{$\scriptstyle x$} -- (-.2,0);
\draw (.2,-.7) node[left, yshift=.2cm]{$\scriptstyle c$} -- (.2,0);
\draw (-.2,.7) node[right, yshift=-.2cm]{$\scriptstyle c$} -- (-.2,0);
\draw (.2,.7) node[left, yshift=-.2cm]{$\scriptstyle y$} -- (.2,0);
\roundNbox{fill=white}{(0,0)}{.3}{.2}{.2}{$\phi$}
}
\right)^*
:=
 \tikzmath{
\draw (-.2,-.7) node[left, yshift=.2cm]{$\scriptstyle y$} -- (-.2,0);
\draw (.2,-.3) arc (-180:0:.3cm) -- node[right]{$\scriptstyle \overline{c}$} (.8,.7);
\draw (-.2,.3) arc (0:180:.3cm) -- node[left]{$\scriptstyle \overline{c}$} (-.8,-.7);
\draw (.2,.7) node[right, yshift=-.2cm]{$\scriptstyle x$} -- (.2,0);
\roundNbox{fill=white}{(0,0)}{.3}{.2}{.2}{$\phi^\dag$}
}\,.
$$

One interpretation of the structure of the tube algebra is that elements of the tube algebra are morphisms on tubes:
$$
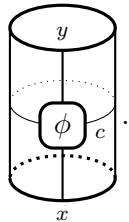
\tikzmath{
\draw[thick] (-.7,-1) -- (-.7,1);
\draw[thick] (.7,-1) -- (.7,1);
\draw[thick] (0,-.3) -- (0,-1.3) node[below]{$\scriptstyle x$};
\draw[thick] (0,0) -- (0,.7) node[above]{$\scriptstyle y$};
\draw[very thick] (0,1) ellipse (.7cm and .3cm);
\halfDottedEllipse[very thick]{(-.7,-1)}{.7}{.3}
\halfDottedEllipse[]{(-.7,0)}{.7}{.3}
\roundNbox{fill=white}{(0,-.3)}{.3}{0}{0}{$\phi$};
\node at (.5,-.4) {$\scriptstyle c$};
}\,.
$$
\end{defn}
The multiplication of the tube algebra is then accomplished by stacking tubes and applying identity \eqref{eq:fusionDecomp} to the strings running around the circumference of the tube, while $\dag$ reflects a tube vertically.

The algebraic correspondence between irreducible representations of the tube algebra and simple objects of $Z(\mathcal{X})$ is worked out in \cite[\S4]{MR1782145}.
There is a natural mathematical way to define a $\dag$-representation $\rho_H$ of $\Tube(\mathcal{X})$ from an object $H\in Z(\mathcal{X})$.
The Hilbert space for the representation will be $\bigoplus_{x\in\Irr(\mathcal{X})}\mathcal{X}(H\to x)$.
For $m:H\to x$ and $\phi\in\mathcal{X}(cy\to zc)$, the action is given by
\begin{equation}
 \label{eq:tubeRep}
\rho_H(\phi)m=\delta_{x,y}\cdot\,
\tikzmath{
\draw (.3,-2) node[below]{$\scriptstyle H$} -- (.3,-.7) -- node[right]{$\scriptstyle x$} (.3,-.3);
\draw (-.3,.3) -- (-.3,.7) node[above]{$\scriptstyle z$};
\draw[knot] (.3,.3) node[left, yshift=.15cm]{$\scriptstyle c$} arc (180:0:.3cm) -- node[right]{$\scriptstyle \overline{c}$} (.9,-1) arc (0:-180:.6cm) -- node[left]{$\scriptstyle c$} (-.3,-.3);
\roundNbox{fill=white}{(0,0)}{.3}{.3}{.3}{$\phi$}
\roundNbox{fill=white}{(.3,-1)}{.3}{0}{0}{$m$}
}
\qquad\qquad\qquad
\forall\,m\in \cX(H\to x).
\end{equation}
This action is the one that appears in \cite[Lemma~{4.7.{\romannumeral 3}}]{MR1782145}.
Since $\im(\rho_H(\id_x))=\mathcal{X}(H\to x)$, the object $H$ can be recovered from $\rho_H$.
Since $\phi\in\mathcal{X}(cy\to zc)$ is a sum of morphisms of the form $f\circ g$, where $f\in\mathcal{X}(cy\to w)$ and $g\in\mathcal{X}(w\to zc)$, the representation $\rho_H$ contains all the data needed to recover the half-braiding on $H$.
One computes the half-braiding by the following equation.
\begin{equation}
 \label{eq:tubeRepToZC}
 \tikzmath{
  \draw (0,-.6) node[below]{$\scriptstyle w$} -- (0,0);
  \draw (0,3) -- (0,3.9) node[above]{$\scriptstyle w$};
  \draw (.5,.3) -- node[right]{$\scriptstyle y$} (.5,.7) -- (.5,1.3) .. controls ++(90:.5cm) and ++(90:-.5cm) .. (-.5,2) --(-.5,2.6) -- node[left]{$\scriptstyle z$} (-.5,3);
  \draw[knot] (-.5,.3) -- node[left]{$\scriptstyle c$} (-.5,1.3) .. controls ++(90:.5cm) and ++(90:-.5cm) .. (.5,2) -- (.5,3);
  \node at (.7,1.5) {$\scriptstyle H$};
  \roundNbox{fill=white}{(0,0)}{.3}{.4}{.4}{$g$}
  \roundNbox{fill=white}{(.5,1)}{.3}{0}{0}{$n^\dag$}
  \roundNbox{fill=white}{(-.5,2.3)}{.3}{0}{0}{$m$}
  \roundNbox{fill=white}{(0,3.3)}{.3}{.4}{.43}{$f$}
 }
 :=\langle n|(g\circ f)\rhd m\rangle\id_w.
\end{equation}
In fact, $\rho_\bullet$ can be extended to a monoidal equivalence $Z(\mathcal{X})\to\Rep(\Tube(\mathcal{X}))$.

Suppose that, in the state $|\phi\rangle$, an excitation is located in a plaquette $p$ near a link $\ell$, as in
$$
|\phi\rangle=\tikzmath{
\coordinate (a) at (-30:.53cm);
\coordinate (b) at (-70:.6cm);
\coordinate (c) at (-110:.85cm);
\coordinate (d) at (-135:1.1cm);
\filldraw[red] (a) circle (.05cm);
\draw[red, thick] (a) .. controls ++(-135:.1cm) and ++(0:.15cm) .. (b)
                            .. controls ++(180:.2cm) and ++(0:.2cm) .. (c) -- (d);
\fill[white] (-90:.65) circle (.07cm);
\fill[white] (-120:.9) circle (.07cm);
\levinHexGrid[]{0}{0}{2}{1}{.75}{black}
\node at (0,0) {$\scriptstyle p$};
\node at (-30:.8) {$\scriptstyle \ell$};
\node at (-30:1.4) {$\scriptstyle q$};
\node at (-75:.9) {$\scriptstyle v$};
\node at (15:.9) {$\scriptstyle w$};
}\,.
$$
Here, the red string depicts a string operator which could be applied to the ground state to obtain $|\phi\rangle$.
Further suppose\footnote{These assumptions are not necessary, but greatly simplify the details.} that the Hamiltonian terms for other links of $p$ and plaquettes adjacent to $p$ other than $q$ are not excited in $|\phi\rangle$, meaning that $|\phi\rangle$ contains an isolated excitation at $(p,\ell)$.
(This excitation could be trivial, {i.e.} the vacuum; the key word here is ``isolated.'')

There is not a straightforward action of $\Tube(\mathcal{X})$ on the space of states with isolated excitations at $(p,\ell)$,
but we can construct an action on the image of $\pi^1_{\ell,w}$, where $w$ is either vertex of $\ell$.
As described in \S\ref{ssec:stringOperators}, the effect of applying $\pi^1_{\ell,w}$ is to replace the plaquettes $p$ and $q$ adjacent to $\ell$ with a decagonal plaquette $p\vee q$, with an link leading inside the plaquette where an excitation may be supported.
The action of $\Tube(\mathcal{X})$ on excitations inside $p\vee q$ is now as described in \cite[\S\RN{5}.A]{PhysRevB.97.195154}:
when acting by an element of $\cX(xc\to cy)\subseteq\Tube(\cX)$, the action is given by
\begin{equation}
 \label{eq:HGWaction}
\tikzmath{
\coordinate (a) at (-30:.53cm);
\coordinate (b) at (-70:.6cm);
\coordinate (c) at (-110:.85cm);
\coordinate (d) at (-135:1.1cm);
\filldraw[red] (a) circle (.05cm);
\draw[red, thick] (a) .. controls ++(-135:.1cm) and ++(0:.15cm) .. (b)
                            .. controls ++(180:.2cm) and ++(0:.2cm) .. (c) -- (d);
\fill[white] (-90:.65) circle (.07cm);
\fill[white] (-120:.9) circle (.07cm);
\levinHexOpen{0}{0}{.75}{black}{6}
\levinHexOpen{1.125}{-.6495}{.75}{black}{3}
\node at (210:1.299) {$\scriptstyle r$};
\node at (210:.8) {$\scriptstyle k$};
\node at (0:0) {$\scriptstyle p$};
\node at (-30:1.299) {$\scriptstyle q$};
}
\qquad
\overset{\phi \rhd -}{\longmapsto}
 \frac{d_y^{1/2}}{d_x^{1/2}}
\qquad
\tikzmath{
\coordinate (a) at (-30:.53cm);
\coordinate (b) at (-70:.6cm);
\coordinate (c) at (-110:.85cm);
\coordinate (d) at (-135:1.1cm);
\filldraw[red] (a) circle (.05cm);
\draw[red, thick] (a) .. controls ++(-135:.1cm) and ++(0:.15cm) .. (b)
                            .. controls ++(180:.2cm) and ++(0:.2cm) .. (c) -- (d);
\arHex[6]{0}{0}{.5}{blue}{3}{0}
\arHex[3]{1.125}{-.6495}{.5}{blue}{3}{0}
\fill[white] (-90:.65) circle (.07cm);
\fill[white] (-120:.9) circle (.07cm);
\draw[blue] (300:.5) -- (.625,-.6495);
\draw[blue] (0:.5) -- (0.875,-.2165);
\fill[black] (305:.6) circle (.07cm); 
\levinHexOpen{0}{0}{.75}{black}{6}
\levinHexOpen{1.125}{-.6495}{.75}{black}{3}
\node at (210:1.299) {$\scriptstyle r$};
\node at (210:.8) {$\scriptstyle k$};
\node at (0:0) {$\scriptstyle p$};
\node at (-30:1.299) {$\scriptstyle q$};
}
\end{equation}
where the vertex between the red and blue strands represents the morphism $\phi$.
The constant factor is required to make the $\Tube(\cX)$ representation a $\dag$-representation.
Because the excitations $Z(\cX)$ form a UMTC, and because an UMTC has a unique unitary structure \cite{MR4538281}, we do not need this fact here, and therefore leave the verification to a forthcoming article.

As we did when defining $B_p$, we use Equation \eqref{eq:fusionDecomp} to rewrite the above diagram as a sum of diagrams where all strings other than lattice links are local to a particular vertex, and then compute the effect on each vertex using the data of $\mathcal{X}$ and $Z(\mathcal{X})$.
One verifies that the action of $\Tube(\mathcal{X})$ is associative, essentially for the same reason that the plaquette operator $B_p$ is an idempotent; the action is manifestly unital.

Evidently, we have defined a representation of $\Tube(\mathcal{X})$ on $\pi^1_{\ell,w}\mathcal{H}$; we will now explain why this representation is faithful, and behaves as expected on excited states obtained via the application of string operators.
First, we establish a technical lemma about states in Levin-Wen models.
Essentially, this lemma shows that if a state $|\phi\rangle$ does \textit{not} contain an excitation along an edge $\ell$, then applying $\pi^1_\ell$ and removing $\ell$ is invertible on the space of states which agree with a ground state near $\ell$; {cf.} \cite{PhysRevB.85.075107}, where it is shown that mutations of trivalent lattices are unitary on ground states.

\begin{defn}
 \label{def:selectEdgeLabel}
For $\ell$ a link in our lattice from $v$ to $w$ and $x\in \Irr(\cX)$,
we define the orthogonal projection $\pi^x_{\ell,w}$ to project onto the subspace of $\cH$ of states which assign $X$ to $\ell$ at $w$, i.e., at the tensorand $\cH_w$, we have
$$
\pi_{\ell,w}^x:
\mathcal{H}_w
\cong
\bigoplus_{a,b,c\in\Irr(\mathcal{X})}\mathcal{X}(ab,c)\to\bigoplus_{a,c\in\Irr(\mathcal{X})}\mathcal{X}(ax,c).
$$
On the ground state of $A_\ell$, we denote the operator $\pi^x_{\ell,w}$ graphically by
$
\pi^x_{\ell,w}
=
\tikzmath{
  \levinHexOpen{0}{0}{.5}{black}{1}
  \draw[orange] (60:.5) -- node[right]{$\scriptstyle x$} (0:.5) node [sloped,allow upside down,pos=0.5] {\arrowIn[orange]};
}\,
$.
\end{defn}

\begin{lem}
\label{lem:restoration}
Suppose $\ell$ is a link between plaquettes $p$ and $q$, and $v$ is a vertex incident to $\ell$.
If $|\phi\rangle$ is a state which is not excited at $(q,\ell)$ {i.e.} $|\phi\rangle=A_\ell|\phi\rangle=B_q|\phi\rangle$,
 then $B_q\pi_{\ell,v}^1|\phi\rangle=\frac{1}{D}|\phi\rangle$.
\end{lem}
\begin{proof}
First, when $A_\ell|\phi\rangle=|\phi\rangle$, each $\pi_{\ell,v}^x|\phi\rangle$ is independent of $v$.
Observe in this case that
\[
|\phi\rangle
=
\sum_{\textcolor{orange}{x}\in \Irr(\cX)}
\tikzmath{
  \levinHexOpen{0}{0}{.5}{black}{1}
  \draw[orange] (60:.5) -- node[right]{$\scriptstyle x$} (0:.5) node [sloped,allow upside down,pos=0.5] {\arrowIn[orange]};
}\,
|\phi\rangle
\qquad\qquad
\Longrightarrow
\qquad\qquad
\pi^1_{\ell,v}B_q|\phi\rangle
=
\sum_{\textcolor{orange}{x}}
 \frac{1}{D}
\tikzmath{
\levinHexOpen{0}{0}{.5}{black}{1}
\draw[orange, mid>] (60:.5) .. controls ++(-60:.2cm) and ++(0:.2cm) .. (90:.3) arc (90:330:.3cm) .. controls ++(60:.2cm) and ++(120:.2cm) .. (0:.5);
\node[orange] at (-.08,-.08) {$\scriptstyle x$};
}\,
|\phi\rangle\text{,}
\]
 where the coefficient is $\frac{1}{D}$ rather than $\frac{d_x}{D}$ because the depicted bends in the $x$-string are rotationally invariant trivalent vertices, rather than the ones which are normal with respect to the isometry inner product that would appear on the right-hand side of \eqref{eq:fusionDecomp}.
 Since $B_q|\phi\rangle=|\phi\rangle$, it suffices to check that $B_q\pi^1_{\ell,v}B_q|\phi\rangle=\frac{1}{D}B_q|\phi\rangle$,
which may be expressed diagrammatically as
\[\sum_{\textcolor{orange}{x},\textcolor{blue}{y}}\frac{d_y}{D^2}
\begin{tikzpicture}[baseline]
\levinHexOpen{0}{0}{.75}{black}{1}
\draw[orange] (60:.75) .. controls ++(-60:.2cm) and ++(0:.2cm) .. (60:.5);
\arHex[1]{0}{0}{.5}{orange}{4}{0}
\draw[orange] (0:.5) .. controls ++(60:.2cm) and ++(120:.2cm) .. (0:.75);
\arHex{0}{0}{.25}{blue}{4}{0}
\end{tikzpicture}
\,|\phi\rangle
\overset{?}{=}
\sum_{\textcolor{orange}{x},\textcolor{DarkGreen}{z}}\frac{d_z}{D^2}
\begin{tikzpicture}[baseline]
\levinHexOpen{0}{0}{.75}{black}{1}
\draw[orange] (60:.75) -- (0:.75) node [sloped,allow upside down,pos=0.5] {\arrowIn[orange]};
\arHex{0}{0}{.5}{DarkGreen}{4}{0}
\end{tikzpicture}
\,|\phi\rangle.
\]
This equation holds, because the effect of applying $B_p$ is to permit the deformation of strings across the center of plaquettes; it first appeared in \cite[Appendix C]{PhysRevB.71.045110}.
\end{proof}

We point out some consequences of Lemma \ref{lem:restoration} which will be useful later.
Sometimes, we wish to write a state $|\phi\rangle$ as a sum of states which have a fixed label for a given edge, as in $|\phi\rangle=\sum_x\pi_{w,\ell}^x|\phi\rangle$.
It turns out that in the ground state, the summands $\pi_{w,\ell}^x$ can be determined in a straightforward manner from $\pi_{w,\ell}^1|\phi\rangle$.
\begin{cor}
\label{cor:edgeDecomp}
Suppose $|\phi\rangle$ is a state satisfying $B_p|\phi\rangle=|\phi\rangle$, and $\ell$ is an edge of $p$.
Then
\[\pi_{w,\ell}^x|\phi\rangle=d_xB_p^x\pi_{w,\ell}^1|\phi\rangle\text{.}\]
\end{cor}
\begin{proof}
By Lemma \ref{lem:restoration}, we have
$\pi_{w,\ell}^x|\phi\rangle=\pi_{w,\ell}^xDB_p\pi_{w,\ell}^1|\phi\rangle= 
d_xB_p^x\pi_{w,\ell}^1|\phi\rangle
$, as desired.
\end{proof}
Corollary \ref{cor:edgeDecomp} was not obvious \emph{a priori} when $X\neq 1$, since several fusion channels contribute to $\pi_{\ell,w}^xB_p|\phi\rangle$.

Lemma \ref{lem:restoration} also allows us to more precisely establish the sense in which applying $\pi_{\ell,w}^1$ replaces the plaquettes $p$ and $q$ with a decagonal plaquette $p\vee q$, analogous to the lattice mutations studied in \cite{PhysRevB.85.075107}.
One would expect to obtain a Levin-Wen model on the lattice obtained by removing the link $\ell$ between $p$ and $q$ from the modified Hamiltonian
\begin{equation}
 \label{eq:punchoutHamiltonian}
 H_{\setminus\ell}=-\pi_{\ell,w}^1-\sum_mA_m-\sum_{r\notin\{p,q\}}B_r-B_{p\vee q}
\end{equation}
where we still have $B_{p\vee q}=\frac{1}{D}\sum_{x\in\Irr(\cX)}d_xB_{p\vee q}^x$ and where $B_{p\vee q}^x$ is defined (on the ground state of $\pi_{\ell,w}^1$ and all $A_m$ terms, including $m=\ell$) by modifying \eqref{eq:Bps} to account for the new plaquette shape.
The new term $\pi_{\ell,w}$ evidently commutes with all $A_m$ terms and $B_r$ for $r\notin\{p,q\}$.
That $B_{p\vee q}$ behaves as expected follows from the following result.
\begin{cor}
 \label{cor:Bpvq}
 Suppose $p$ and $q$ are adjacent plaquettes, $\ell$ is the link where $p$ and $q$ meet, and $w$ is either vertex of $\ell$.
 On the space of states $\set{|\Omega\rangle}{B_p|\Omega\rangle=B_q|\Omega\rangle=|\Omega\rangle}$, 
 \[
 D\pi_{\ell,w}^1B_pB_q\pi_{\ell,w}^1 = B_{p\vee q}\pi_{\ell,w}^1 =\pi_{\ell,w}^1.
 \]
\end{cor}
\begin{proof}
 To show $B_{p\vee q}\pi_{w,\ell}^1|\Omega\rangle=D\pi_{w,\ell}^1B_pB_q\pi_{w,\ell}^1|\Omega\rangle$, we expand the left hand side.
 \begin{align*}
  D\pi_{w,\ell}^1B_pB_q\pi_{w,\ell}^1|\Omega\rangle &= \frac{1}{D}\sum_{x,y\in\Irr(\cX)}d_xd_y\pi_{w,\ell}^1B_p^xB_q^y\pi_{w,\ell}^1|\Omega\rangle\\
  &= \frac{1}{D}\sum_{x\in\Irr(\cX)}d_x^2\pi_{w,\ell}^1B_p^xB_q^x\pi_{w,\ell}^1|\Omega\rangle\\
  &= \frac{1}{D}\sum_{x\in\Irr(\cX)}d_xB_{p\vee q}^x\pi_{w,\ell}^1|\Omega\rangle\\
  &= B_{p\vee q}\pi_{w,\ell}^1|\Omega\rangle.
 \end{align*}
Above, the only way to get $1$ after applying $B_p^x$ and $B_q^y$ to $\pi^1_{\ell,w}$ is when $x=y$, and the scalar $d_x$ difference again arises from rotational invariance scaling of cups and caps.
Now by Lemma \ref{lem:restoration}, 
$D\pi_{\ell,w}^1B_pB_q\pi_{\ell,w}^1|\Omega\rangle = \pi_{\ell,w}^1B_p|\Omega\rangle = \pi_{\ell,w}^1|\Omega\rangle$.
\end{proof}
On the other hand, by Lemma \ref{lem:restoration}, all ground states of \eqref{eq:punchoutHamiltonian} are of the form $D\pi^1_{\ell,w}|\Omega\rangle$, where $|\Omega\rangle$ is a ground state of the original Hamiltonian \eqref{eq:LWHamiltonian}.
Consequently, $D\pi^1_{\ell,w}$ is a unitary map between the spaces of ground states of the two Hamiltonians.
We have now shown \eqref{eq:punchoutHamiltonian} to be a frustration-free commuting projector Hamiltonian which contains the terms $\pi_{\ell,w}^1$ and $A_\ell$, so it indeed has the same space of ground states as the natural Levin-Wen Hamiltonian defined on the lattice obtained by removing $\ell$.

Using the previous results, we can now demonstrate the compatibility of our definition of string operator with the tube algebra action \eqref{eq:HGWaction}.
\begin{prop}
 \label{prop:tubeActionCorrect}
 Suppose $H\in Z(\mathcal{X})$ and $\psi\in \cX(H\to x)$ where $x\in \Irr(\cX)$.
 Let $r$ be a path ending at $(p,\ell)$ and beginning far away,
 and suppose $|\Omega\rangle$ is locally a ground state near the path $r$.
 Then the tube algebra action on a state excited at the endpoint of $r$ generates a representation isomorphic to $\rho_H$, where $\rho_H$ is the representation defined in Equation \eqref{eq:tubeRep}.
 Explicitly, for $f\in\cX(xc\to cy)\subseteq\Tube(\mathcal{X})$ and $\psi\in\cX(H\to x)$, we have
 \begin{equation}
 \label{eq:tubeActionAtVertex}
 f\rhd\sigma^H_r(\phi,\psi)|\Omega\rangle=\sqrt{\frac{d_y}{d_x}}\sigma^H_r(\phi,\rho_H(f)\psi)|\Omega\rangle\text{.}
 \end{equation}
\end{prop}
\begin{proof}
Note that introducing or removing the scalar $\sqrt{\frac{d_y}{d_x}}$ is an automorphism of $\Tube(\cX)$, so the presence of this scalar is immaterial to whether the two representations are isomorphic; the scalar only affects unitarity.

In general, computing $f\rhd|\eta\rangle$ for $f\in\Tube(\mathcal{X})$ involves gluing a strand into the plaquette $p\vee q$, which means summing over many fusion channels and basis vectors in an expression of some state $|\eta\rangle$.
We will first exploit Lemma \ref{lem:restoration} to show that $f\rhd-$ actually reduces to an operator local to the final vertex of the path $r$, such that $f\rhd\sigma^H_r(\phi,\psi)|\Omega\rangle=\sigma^H_r(\phi,\psi')|\Omega\rangle$ for some $\psi'$.
Then, we will compute algebraically that this $\psi'$ is just $\rho_H(f)\psi$.

The situation of the Proposition can be depicted graphically by
\[
\sigma^H_r(\phi,\psi)|\Omega\rangle=
\tikzmath[scale=1.1]{
\coordinate (a) at (-40:.58cm);
\coordinate (b) at (-70:.6cm);
\coordinate (c) at (-110:.85cm);
\coordinate (d) at (-135:1.1cm);
\draw[thick, orange] (a) -- ($ (a) + (.15,.2) $) node[left]{$\scriptstyle x$};
\filldraw[red] (a) node[right]{$\scriptstyle \psi$} node[left]{$\scriptstyle H$} circle (.05cm);
\draw[red, thick] (a) .. controls ++(-135:.1cm) and ++(0:.15cm) .. (b)
                        .. controls ++(180:.2cm) and ++(0:.2cm) .. (c) -- (d);
\fill[white] (-90:.65) circle (.07cm);
\fill[white] (-120:.9) circle (.07cm);
\levinHexOpen{0}{0}{.75}{black}{6}
\levinHexOpen{1.125}{-.6495}{.75}{black}{3}
\node at (210:1.299) {$\scriptstyle r$};
\node at (210:.8) {$\scriptstyle k$};
\node at (0:0) {$\scriptstyle p$};
\node at (-30:1.299) {$\scriptstyle q$};
}
\,|\Omega\rangle
=
\tikzmath{
\draw[thick, red] (.4,-.5) -- (.4,-.3) to[out=90,in=180] (.7,-.1) --node[below]{$\scriptstyle H$} (.9,-.1) to[out=0,in=-90] (1,.1) -- (1,.4);
\filldraw[white] (.5,-.15) circle (.05cm);
\draw[knot] (-.5,0) -- (2.5,0);
\draw (-.5,1) -- (2.5,1);
\draw (0,0) -- (0,1);
\draw (2,0) -- (2,1);
\draw (.5,-.5) -- (.5,0);
\draw (.5,1) -- (.5,1.5);
\draw (1.5,-.5) -- (1.5,0);
\draw (1.5,1) -- (1.5,1.5);
\draw[dotted] (1,0) -- (1,1);
\draw[thick, orange] (1,.4) -- (1,.6) node[left]{$\scriptstyle x$};
\filldraw[red] (1,.4) node[right]{$\scriptstyle \psi$} circle (.05cm);
}
|\Omega\rangle
\]
We use the bricklayer lattice instead of the honeycomb lattice here for readability.
We also use the rotationally invariant version of the fusion relation \eqref{eq:fusionDecomp} using an ONB with respect to skein-module inner product rather than the isometry inner product:
$$
\sum_{z\in \Irr(\cX)}
\sqrt{d_z}
\begin{tikzpicture}[baseline=-.1cm]
	\draw (.2,-.6) -- (0,-.3) -- (-.2,-.6);
	\draw (.2,.6) -- (0,.3) -- (-.2,.6);
	\draw (0,-.3) -- (0,.3);
	\draw[fill=cyan] (0,-.3) circle (.05cm);
	\draw[fill=cyan] (0,.3) circle (.05cm);	
	\node at (-.2,-.8) {\scriptsize{$x$}};
	\node at (.2,-.8) {\scriptsize{$y$}};
	\node at (-.2,.8) {\scriptsize{$x$}};
	\node at (.2,.8) {\scriptsize{$y$}};
	\node at (.2,0) {\scriptsize{$z$}};
\end{tikzpicture}
\,=\,\sqrt{d_xd_y}\cdot
\begin{tikzpicture}[baseline=-.1cm]
	\draw (.2,-.6) -- (.2,.6);
	\draw (-.2,-.6) -- (-.2,.6);
	\node at (-.2,-.8) {\scriptsize{$x$}};
	\node at (.2,-.8) {\scriptsize{$y$}};
\end{tikzpicture}
$$

Let $f\in\mathcal{X}(by\to zb)\subseteq\Tube(\mathcal{X})$.
In case $x\neq y$, we have $f\rhd\sigma^H_r(\phi,\psi)|\Omega\rangle=0=\sigma^H_r(\phi,\rho_H(f)\psi)|\Omega\rangle$, by definition.
Now suppose $x=y$.
Since $\sigma^H_r(\phi,\psi)|\Omega\rangle = \sigma^H_r(\phi,\psi) \pi^1_{\ell,w}|\Omega\rangle = \sigma^H_r(\phi,\psi) B_{p\vee q}\pi^1_{\ell,w}|\Omega\rangle$,
we have
\begin{align*}
f\rhd
\tikzmath{
\draw[thick, red] (.4,-.5) -- (.4,-.3) to[out=90,in=180] (.7,-.1) --node[below]{$\scriptstyle H$} (.9,-.1) to[out=0,in=-90] (1,.1) -- (1,.4);
\filldraw[white] (.5,-.15) circle (.05cm);
\draw[knot] (-.5,0) -- (2.5,0);
\draw (-.5,1) -- (2.5,1);
\draw (0,0) -- (0,1);
\draw (2,0) -- (2,1);
\draw (.5,-.5) -- (.5,0);
\draw (.5,1) -- (.5,1.5);
\draw (1.5,-.5) -- (1.5,0);
\draw (1.5,1) -- (1.5,1.5);
\draw[dotted] (1,0) -- (1,1);
\draw[thick, orange] (1,.4) -- (1,.6) node[left]{$\scriptstyle x$};
\filldraw[red] (1,.4) node[right]{$\scriptstyle \psi$} circle (.05cm);
}
|\Omega\rangle
&=
\sqrt{\frac{d_y}{d_x}}
\frac{1}{D}\sum_{\textcolor{blue}{a}\in \Irr(\cX)} \textcolor{blue}{d_a}
\tikzmath[scale=2]{
\draw[thick, red] (.4,-.5) -- (.4,-.3) to[out=90,in=180] (.7,-.1) --node[below]{$\scriptstyle H$} (.9,-.1) to[out=0,in=-90] (1,.1) -- (1,.15);
\filldraw[white] (.5,-.15) circle (.05cm);
\draw[knot] (-.5,0) -- (2.5,0);
\draw (-.5,1) -- (2.5,1);
\draw (0,0) -- (0,1);
\draw (2,0) -- (2,1);
\draw (.5,-.5) -- (.5,0);
\draw (.5,1) -- (.5,1.5);
\draw (1.5,-.5) -- (1.5,0);
\draw (1.5,1) -- (1.5,1.5);
\draw[dotted] (1,0) -- (1,1);
\draw[thick, orange] (1,.15) -- node[left]{$\scriptstyle x$} (1,.3) ;
\draw[thick, orange] (1,.3) -- (1,.45) node[left]{$\scriptstyle z$};
\filldraw[orange] (1,.3) node[right, yshift=.15cm]{$\scriptstyle f$} circle (.025cm);
\filldraw[red] (1,.15) node[right]{$\scriptstyle \psi$} circle (.025cm);
\draw[knot, thick, blue, rounded corners=5pt] (.05,.05) rectangle (1.95,.95); 
\draw[thick, orange, rounded corners=5pt] (1,.3) to[out=180, in=0] (.5,.1) -- (.1,.1) -- (.1,.9) -- (1.9,.9) -- (1.9,.1) -- (1.5,.1) to[out=180,in=0] (1,.3);
\node[blue] at (.85,.1) {$\scriptstyle a$};
\node[orange] at (.3,.2) {$\scriptstyle b$};
}
|\Omega\rangle
\displaybreak[1]\\&=
\sqrt{\frac{d_y}{d_x}}
\frac{1}{D}\sum_{\textcolor{blue}{a},\textcolor{DarkGreen}{c}\in \Irr(\cX)} \frac{\sqrt{\textcolor{blue}{d_a}\textcolor{DarkGreen}{d_c}}}{\sqrt{\textcolor{orange}{d_b}}}
\tikzmath[scale=2]{
\draw[thick, red] (.4,-.5) -- (.4,-.3) to[out=90,in=180] (.7,-.1) --node[below]{$\scriptstyle H$} (.9,-.1) to[out=0,in=-90] (1,.1) -- (1,.15);
\filldraw[white] (.5,-.15) circle (.05cm);
\draw[knot] (-.5,0) -- (2.5,0);
\draw (-.5,1) -- (2.5,1);
\draw (0,0) -- (0,1);
\draw (2,0) -- (2,1);
\draw (.5,-.5) -- (.5,0);
\draw (.5,1) -- (.5,1.5);
\draw (1.5,-.5) -- (1.5,0);
\draw (1.5,1) -- (1.5,1.5);
\draw[dotted] (1,0) -- (1,1);
\draw[thick, orange] (1,.15) -- node[left]{$\scriptstyle x$} (1,.3) ;
\draw[thick, orange] (1,.3) -- (1,.45) node[left]{$\scriptstyle z$};
\filldraw[orange] (1,.3) node[right, yshift=.15cm]{$\scriptstyle f$} circle (.025cm);
\filldraw[red] (1,.15) node[right]{$\scriptstyle \psi$} circle (.025cm);
\draw[knot, thick, blue, rounded corners=5pt] (.8,.95) -- (.05,.95) -- (.05,.05) -- (1.95,.05) -- (1.95,.95) -- (1.2,.95); 
\draw[thick, orange, rounded corners=5pt] (1,.3) to[out=180, in=0] (.5,.1) -- (.1,.1) -- (.1,.9) -- (.5,.9) to[out=180, in=-135] (.8,.95);
\draw[thick, orange, rounded corners=5pt] (1.2,.95) to[out=-45, in=180] (1.5,.9) -- (1.9,.9) -- (1.9,.1) -- (1.5,.1) to[out=180,in=0] (1,.3);
\draw[thick, DarkGreen] (.8,.95) -- (1.2,.95);
\filldraw[fill=cyan] (.8,.95) circle (.025cm);
\filldraw[fill=cyan] (1.2,.95) circle (.025cm);
\node[blue] at (.85,.1) {$\scriptstyle a$};
\node[orange] at (.3,.2) {$\scriptstyle b$};
\node[DarkGreen] at (.9,.85) {$\scriptstyle \overline{c}$};
}
\displaybreak[1]\\&=
\sqrt{\frac{d_y}{d_x}}
\frac{1}{D}\sum_{\textcolor{blue}{a},\textcolor{DarkGreen}{c}\in \Irr(\cX)} \frac{\sqrt{\textcolor{blue}{d_a}\textcolor{DarkGreen}{d_c}}}{\sqrt{\textcolor{orange}{d_b}}}
\tikzmath[scale=2]{
\draw[thick, red] (.4,-.5) -- (.4,-.3) to[out=90,in=180] (.7,-.1) --node[below]{$\scriptstyle H$} (.9,-.1) to[out=0,in=-90] (1,.1) -- (1,.15);
\filldraw[white] (.5,-.15) circle (.05cm);
\draw[knot] (-.5,0) -- (2.5,0);
\draw (-.5,1) -- (2.5,1);
\draw (0,0) -- (0,1);
\draw (2,0) -- (2,1);
\draw (.5,-.5) -- (.5,0);
\draw (.5,1) -- (.5,1.5);
\draw (1.5,-.5) -- (1.5,0);
\draw (1.5,1) -- (1.5,1.5);
\draw[dotted] (1,0) -- (1,1);
\draw[thick, orange] (1,.15) -- node[left]{$\scriptstyle x$} (1,.3) ;
\draw[thick, orange] (1,.3) -- (1,.45) node[left]{$\scriptstyle z$};
\filldraw[orange] (1,.3) node[right, yshift=.15cm]{$\scriptstyle f$} circle (.025cm);
\filldraw[red] (1,.15) node[right]{$\scriptstyle \psi$} circle (.025cm);
\draw[knot, thick, blue, rounded corners=5pt] (.7,.05) -- (1.3,.05); 
\draw[thick, orange] (1.3,.05) to[out=0,in=0] (1.3,.3) -- (1,.3) -- (.7,.3) to[out=180, in=180] (.7,.05);
\draw[thick, DarkGreen, rounded corners=5pt] (.7,.05) -- (.05,.05) -- (.05,.95) -- (1.95,.95) -- (1.95,.05) -- (1.3,.05);
\filldraw[fill=cyan] (.7,.05) circle (.025cm);
\filldraw[fill=cyan] (1.3,.05) circle (.025cm);
\node[blue] at (.85,.1) {$\scriptstyle a$};
\node[orange] at (.55,.2) {$\scriptstyle \overline{b}$};
\node[DarkGreen] at (.3,.15) {$\scriptstyle c$};
}
|\Omega\rangle
\displaybreak[1]\\&=
\sqrt{\frac{d_y}{d_x}}
\frac{1}{D}\sum_{\textcolor{DarkGreen}{c}\in \Irr(\cX)} \textcolor{DarkGreen}{d_c}
\tikzmath[scale=2]{
\draw[thick, red] (.4,-.5) -- (.4,-.3) to[out=90,in=180] (.7,-.1) --node[below]{$\scriptstyle H$} (.9,-.1) to[out=0,in=-90] (1,.1) -- (1,.2);
\filldraw[white] (.5,-.15) circle (.05cm);
\draw[knot] (-.5,0) -- (2.5,0);
\draw (-.5,1) -- (2.5,1);
\draw (0,0) -- (0,1);
\draw (2,0) -- (2,1);
\draw (.5,-.5) -- (.5,0);
\draw (.5,1) -- (.5,1.5);
\draw (1.5,-.5) -- (1.5,0);
\draw (1.5,1) -- (1.5,1.5);
\draw[dotted] (1,0) -- (1,1);
\draw[thick, orange] (1,.2) -- node[left]{$\scriptstyle x$} (1,.4) ;
\draw[thick, orange] (1,.4) -- (1,.55) node[left]{$\scriptstyle z$};
\filldraw[red] (1,.2) node[right]{$\scriptstyle \psi$} circle (.025cm);
\draw[thick, orange, knot] (1,.4) -- (.8,.4) to[out=180, in=180] (.8,.1) -- (1.2,.1) to[out=0,in=0] (1.2,.4) -- (1,.4);
\draw[thick, DarkGreen, knot, rounded corners=5pt] (.05,.05) rectangle (1.95,.95);
\filldraw[orange] (1,.4) node[right, yshift=.15cm]{$\scriptstyle f$} circle (.025cm);
\node[orange] at (.55,.2) {$\scriptstyle \overline{b}$};
\node[DarkGreen] at (.3,.15) {$\scriptstyle c$};
}
|\Omega\rangle
\displaybreak[1]\\&=
\sqrt{\frac{d_y}{d_x}}\sigma^H_r(\phi,\rho_H(f)\psi)|\Omega\rangle
\end{align*}
as claimed.

So far, we have a well-defined representation of $\Tube(\cX)$ on the Hilbert space 
$$\spann\set{\sigma^H_r(\phi,\psi)|\Omega\rangle}{\psi\in\rho_H},$$ 
and an obvious surjective representation homomorphism given by $\psi\mapsto\sigma^H_r(\phi,\psi)|\Omega\rangle$.
There are therefore two possibilities: either this representation is isomorphic to $\rho_H$, or we simply have the $0$ representation.
The last thing we must do is rule out the latter possibility.

If the representation we have just defined were $0$ for all $\phi$ and $\Omega$, then all our string operators would just be $0$ on the space of ground states.
Hence, it suffices to check that there is at least one string operator $\sigma^H_R(\phi,\psi)$ for each $H\in\Irr(Z(\cX))$.
By Lemma \ref{lem:restoration}, string operators $\sigma^1_r(z,w)$ are all nonzero on ground states.
By Equation \eqref{eq:stringFusion}, we see that $\sigma^{\overline{H}}_r[1,1](1,1)\sigma^H_r(\phi,\psi)=\sigma^1_r(1,1)$, which has a nonzero action on the space of ground states.
Hence, the factor $\sigma^H_r(\phi,\psi)$ was itself nonzero on ground states, completing the proof.

One might worry that our construction of generalized string operators such as $\sigma^{\overline{H}}_r[1,1](1,1)$ relies on Corollary \ref{cor:tubeKetbra}, making this reasoning circular.
However, Corollary \ref{cor:tubeKetbra} does not rely on the fact that the representation $\set{\sigma^H_r(\phi,\psi)|\Omega\rangle}{\psi\in\rho_H}$ is nonzero, but only on the fact that it is a transitive representation of $\Tube(\cX)$, which we have already proven, so there is no issue.
\end{proof}

Because $\Tube(\mathcal{X})$ is a finite dimensional $\Cstar$-algebra, it is nothing more than a multimatrix algebra, i.e. $\bigoplus_kM_{n_k}(\mathbb{C})$, with one summand for each irreducible representation.
We know that the irreducible representations of $\Tube(\mathcal{X})$ are just $\rho_s$ for $s\in\Irr(Z(\mathcal{X}))$.
Actually computing an isomorphism $\Phi:\Tube(\mathcal{X})\to\bigoplus_s\End(\rho_s)$, where $\End(\rho_s)\cong M_{\dim(\rho_s)}(\mathbb{C})$ means endomorphisms as a Hilbert space, is another matter.
The existence of such a $\Phi$ has several consequences: for one thing, it implies that all our string operators $\sigma^s_p(\phi,\psi)$ are distinct.
However, one must explicitly compute $\Phi$ in order to obtain operators $T^s_{\ell,v}(\phi,\psi)$ used to define hopping operators and string operators on excited states, as we see in the following corollary.
\begin{cor}
 \label{cor:tubeKetbra}
 For any link $\ell$, vertex $v$ of $\ell$, anyon $s\in\Irr(Z(\mathcal{X}))$, and morphisms $\phi,\psi\in\bigoplus_x\mathcal{X}(s\to x)$, there is a local operator $T^s_{\ell,v}(\phi,\psi)$ such that, if $|\omega\rangle=A_\ell|\omega\rangle=B_r|\omega\rangle$ for plaquettes $r$ containing $\ell$, then
 \[
  T^s_{\ell,v}(\phi,\psi)\sigma^s_p(\eta',\eta)|\omega\rangle=\langle\phi|\eta\rangle\sigma^s_p(\eta',\psi)|\omega\rangle\text{,}
 \]
 and if $t\neq s$,
 \[T^s_{\ell,v}(\phi,\psi)\sigma^t_p(\eta',\eta)|\omega\rangle=0\text{.}\]
\end{cor}
\begin{proof}
 Since $\phi$ and $\psi$ are vectors in $\rho_s$, we simply set $T^s_{\ell,v}(\phi,\psi)=\Phi^{-1}(|\psi\rangle\langle\phi|)\rhd-$.
 The desired result is now just a case of Proposition \ref{prop:tubeActionCorrect}.
\end{proof}
\begin{cor}
 \label{cor:hoppingIntertwiner}
 The hopping operators $h^s_q$ are intertwiners for local $\Tube(\mathcal{X})$-actions on spaces of excitations at the endpoints of $q$.
\end{cor}
\begin{proof}
 This follows immediately from Equation \eqref{eq:hopping} and Proposition \ref{prop:tubeActionCorrect}.
\end{proof}

Proposition \ref{prop:tubeActionCorrect} makes explicit the correspondence between quasiparticle excitations and simple objects of $Z(\mathcal{X})$ which is described in \cite[Section~\RN{5}.A]{PhysRevB.97.195154}.
In the language of that article, when applying $\sigma_p^s(\phi,\psi)$ for different choices of $\psi$, we obtain different dyons belonging to the same dyon species $s$.
The possible excitations of type $s$ form an irreducible representation of the $\Tube(\mathcal{X})$, acting locally at the endpoint of $p$, so we can explicitly construct local operators that permute the dyons of a given species.
Moreover, the more general notion of string operator we have given here allows us to locally realize all dyons in a given dyon species via string operators.

\section{Lattice model for anyon condensation}
\label{sec:condensationModels}

We will now describe a class of string-net lattice model due to Corey Jones, parameterized by the choice of a unitary fusion category $\cX$ and a condensable algebra $A\in Z(\cX)$, which supports a phase transition between $Z(\cX)$ and $Z(\cX)_A^{\loc}$ topological order.
While \cite{PhysRevB.79.045316} discusses condensable algebras in the context of phase transitions, other works on anyons condensation, such as \cite{MR3246855,MR2942952}, have focused on describing a spatial boundary between regions where $A$ is and is not condensed.
In our construction, one can recover such a spatial boundary by performing the phase transition in only part of the lattice.

In \S~\ref{ssec:etaleMaths}, we review the mathematics of condensable algebras, which are the data necessary to perform anyon condensation, as described in \cite{PhysRevB.79.045316,MR3246855}.
In \S~\ref{ssec:condensationModel}, we present a string-net model where condensation of $A$ can be performed by tuning a parameter $t$ from $0$ to $1$.
In \S~\ref{ssec:t=0}, we see that when $t=0$, our model reduces to the one introduced in \S~\ref{ssec:lwModel}, and hence has $Z(\cX)$ topological order.
In \S~\ref{ssec:condensedPhase}, we see that when $t=1$, our model has $Z(\cX)_A^{\loc}$ topological order.
To show this, we investigate how string operators from the model of \S~\ref{ssec:lwModel} are modified to give string operators in the new model, and give an algebra $\PreTubeAC$, analogous to $\Tube(\cX)$, of local operators acting on states containing an isolated excitation.
We generalize the arguments of the previous section to show that $\PreTubeAC$ classifies localized excitations when $t=1$, some of which are not topologically mobile.
The excitations which are topological are representations of a quotient $\Tube_A(\cX)$ of $\PreTubeAC$, which we prove is Morita equivalent to $\Tube(\cX_A)$.
Since $Z(\cX_A)\cong Z(\cX)_A^{\loc}$, this verifies that topological excitations in our model at $t=1$ are indeed described by $Z(\cX)_A^{\loc}$.

\subsection{Background: Condensable Algebras}
\label{ssec:etaleMaths}
We begin by recalling the definition of a condensable algebra in a UMTC $\cC$.
We then give some basic facts about condensable algebras, which we will later use.
\begin{defn}
 An \textbf{algebra} in $\mathcal{C}$ is an object $A$ equipped with a unit morphism $u:1\to A$, depicted by a univalent vertex, and a multiplication morphism $m:A A\to A$, depicted by a trivalent vertex, with the following properties:
 \begin{itemize}
  \item (unitality)
$
\tikzmath{
\draw[thick,red] (0,-.4) -- (0,.4);
\draw[thick,red] (0,0) -- (-.2,-.2);
\filldraw[red] (-.2,-.2) circle (.05cm);
}
=
\tikzmath{
\draw[thick,red] (.2,-.4) -- (.2,.4);
}
=
\tikzmath[xscale=-1]{
\draw[thick,red] (0,-.4) -- (0,.4);
\draw[thick,red] (0,0) -- (-.2,-.2);
\filldraw[red] (-.2,-.2) circle (.05cm);
}
$
  \item (associativity)
$
\tikzmath{
\draw[thick,red] (-.2,-.2) arc (180:0:.2cm);
\draw[thick,red] (0,0) arc (180:0:.3cm) -- (.6,-.2);
\draw[thick,red] (.3,.3) -- (.3,.5);
}
=
\tikzmath[xscale=-1]{
\draw[thick,red] (-.2,-.2) arc (180:0:.2cm);
\draw[thick,red] (0,0) arc (180:0:.3cm) -- (.6,-.2);
\draw[thick,red] (.3,.3) -- (.3,.5);
}$
 \end{itemize}
 A \textbf{condensable algebra} is also \emph{commutative}, meaning $m\circ \beta_{A,A} = m$ where $\beta$ is the braiding in $\cC$, and \emph{unitarily separable}, meaning that $m^\dag$ splits $m$ as an $A-A$ bimodule map:
 \begin{equation}
  \label{eq:separable}
  \tikzmath{
\draw[thick,red] (-.2,.2) -- (-.2,.4);
\draw[thick,red] (.2,-.2) -- (.2,-.4);
\draw[thick,red] (-.4,-.4) -- (-.4,0) arc (180:0:.2cm) arc (-180:0:.2cm) -- (.4,.4);
}
=
\tikzmath{
\draw[thick,red] (-.2,.4) arc (-180:0:.2cm);
\draw[thick,red] (-.2,-.4) arc (180:0:.2cm);
\draw[thick,red] (0,-.2) -- (0,.2);
}
=
\tikzmath[xscale=-1]{
\draw[thick,red] (-.2,.2) -- (-.2,.4);
\draw[thick,red] (.2,-.2) -- (.2,-.4);
\draw[thick,red] (-.4,-.4) -- (-.4,0) arc (180:0:.2cm) arc (-180:0:.2cm) -- (.4,.4);
}
\qquad\qquad
\tikzmath{
\draw[thick,red] (0,.2) -- (0,.4);
\draw[thick,red] (0,-.2) -- (0,-.4);
\draw[thick,red] (0,0) circle (.2cm);
}
=
\tikzmath{
\draw[thick,red] (0,-.4) -- (0,.4);
}
\end{equation}
Here, the vertical reflection of $m$ denotes $m^\dag$.
Finally, a condensable algebra $A$ is \textit{connected}, meaning that $\mathcal{X}(1\to A)$ is $1$-dimensional.

 We call the projection $m^\dag m\in\End(AA)$ the \emph{condensation morphism}; by Condition \eqref{eq:separable}, the condensation morphism is a projection with image isomorphic to $A$.
\end{defn}

We will use the condensation morphism to define the term of our Hamiltonian which implements anyon condensation; see \eqref{eq:DTerm}.
By composing condensation morphisms for different pairs of copies of $A$, we can also obtain projections in $\End(A^{\otimes n})$ for $n>2$, as in Figure \ref{fig:condensationMorphisms}.
Intuitively, when a condensable algebra is condensed, copies of that algebra saturate the system, and any two nearby copies are entangled via the condensation morphism, so these projections in $\End(A^{\otimes n})$ will be of interest to us.
As seen in the figure, since the copies of $A$ fill a $2$-dimensional region, we must also consider conjugating condensation morphisms by the braiding.

\begin{figure}
 \centering
 \begin{tikzpicture}
  \draw[blue] (0,0) -- (1,0);
  \draw[blue] (0,0) -- (.3,.2);
  \draw[blue] (.3,.2) -- (1,0);
  \filldraw (0,0) circle (2pt) node[anchor=east] {$\scriptstyle A$};
  \filldraw (.3,.2) circle (2pt) node[anchor=north] {$\scriptstyle A$};
  \filldraw (1,0) circle (2pt) node[anchor=west] {$\scriptstyle A$};
  \draw[thick,red] (0,0) arc (213.69:33.69:.1803cm);
  \draw[thick,red] (.05,.25) -- (-.05,.4);
  \draw[thick,red] (-.05,.4) arc (-56.31:-146.31:.1803cm);
  \draw[thick,red] (-.05,.4) arc (-56.31:33.69:.1803cm);
  \draw[thick,red] (1,0) -- (1,.6);
 \end{tikzpicture}
 \begin{tikzpicture}
  \draw[thick,red] (.3,.2) -- (.3,1.2);
  \draw[thick,red,knot] (0,0) arc (180:0:.5cm);
  \draw[thick,red,knot] (.5,.5) -- (.5,.7);
  \draw[thick,red,knot] (0,1.2) arc (-180:0:.5cm);
  \draw[blue] (0,0) -- (1,0);
  \draw[blue] (0,0) -- (.3,.2);
  \draw[blue] (.3,.2) -- (1,0);
  \filldraw (0,0) circle (2pt) node[anchor=east] {$\scriptstyle A$};
  \filldraw (.3,.2) circle (2pt) node[anchor=north] {$\scriptstyle A$};
  \filldraw (1,0) circle (2pt) node[anchor=west] {$\scriptstyle A$};
 \end{tikzpicture}
 \begin{tikzpicture}
  \draw[blue] (0,0) -- (1,0);
  \draw[blue] (0,0) -- (.3,.2);
  \draw[blue] (.3,.2) -- (1,0);
  \filldraw (0,0) circle (2pt) node[anchor=east] {$\scriptstyle A$};
  \filldraw (.3,.2) circle (2pt) node[anchor=north] {$\scriptstyle A$};
  \filldraw (1,0) circle (2pt) node[anchor=west] {$\scriptstyle A$};
  \draw[thick,red] (.3,.2) arc (164.1:-15.9:.364cm);
  \draw[thick,red] (.75,.45) -- (.85,.8);
  \draw[thick,red] (.85,.8) arc (254.1:164.1:.364cm);
  \draw[thick,red] (.85,.8) arc (-105.9:-15.9:.364cm);
  \draw[thick,red] (0,0) -- (0,1.2);
 \end{tikzpicture}
  \begin{tikzpicture}
  \draw[blue] (0,0,0) -- (1,0,0);
  \draw[blue] (0,0,0) -- (.5,.866,0);
  \draw[blue] (.5,.866,0) -- (1,0,0);
  \filldraw (0,0,0) circle (2pt) node[anchor=east] {$\scriptstyle A$};
  \filldraw (.5,.866,0) circle (2pt) node[anchor=north] {$\scriptstyle A$};
  \filldraw (1,0,0) circle (2pt) node[anchor=west] {$\scriptstyle A$};
 \end{tikzpicture}
 \caption{Condensation morphisms between $3$ nearby copies of $A$, with time depicted in the vertical direction. Blue segments depict paths on the underlying lattice. Notice that one morphism involves braiding, depending on perspective.}
 \label{fig:condensationMorphisms}
\end{figure}
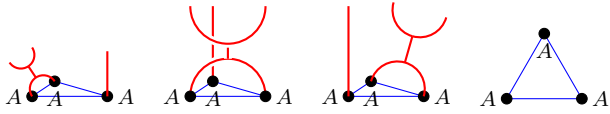

The following lemmas show that the order in which several condensation morphisms are composed and the choices of over- or under-braiding do not affect the resulting projection in $\End(A^{\otimes n})$.
In particular, condensation morphisms generate a unique projection $\End(A^{\otimes n})$ in which all $n$ copies of $A$ interact.
\begin{lem}
 \label{lem:DCommutesMath}
 Suppose $A$ is a condensable algebra in $\mathcal{C}$ with multiplication $m$.
 Then $( m^\dag m)\otimes 1$ and $1\otimes( m^\dag m)$ commute.
\end{lem}
\begin{proof}
First, observe that the separator is coassociative by applying $\dag$ to the associativity axiom.
We then have
 \[(1\otimes( m^\dag m))\circ(( m^\dag m)\otimes 1)=
  \left.
  \begin{tikzpicture}[baseline,shift={(0,-.5)}]
   \draw[thick,red] (0,0) arc (180:0:.2cm);
   \draw[thick,red] (.2,.2) -- (.2,.4);
   \draw[thick,red] (.8,0)-- (.8,.6) arc (0:180:.2cm) arc (0:-180:.2cm) -- (0,1.2);
   \draw[thick,red] (.6,.8) -- (.6,1);
   \draw[thick,red] (.4,1.2) arc (-180:0:.2cm);
  \end{tikzpicture}
  \right.
  =
  \begin{tikzpicture}[baseline,shift={(0,-.6)}]
   \draw[thick,red] (0,0) arc (180:0:.2cm);
   \draw[thick,red] (.2,.2) arc (180:0:.3cm) -- (.8,0);
   \draw[thick,red] (.5,.5) -- (.5,.8);
   \draw[thick,red] (.2,1.3) -- (.2,1.1) arc (-180:0:.3cm);
   \draw[thick,red] (.6,1.3) arc (-180:0:.2cm);
  \end{tikzpicture}
  =
  \begin{tikzpicture}[baseline,shift={(0,-.6)}]
   \draw[thick,red] (.4,0) arc (180:0:.2cm);
   \draw[thick,red] (0,0) -- (0,.2) arc (180:0:.3cm);
   \draw[thick,red] (.3,.5) -- (.3,.8);
   \draw[thick,red] (0,1.1) arc (-180:0:.3cm) -- (.6,1.3);
   \draw[thick,red] (-.2,1.3) arc (-180:0:.2cm);
  \end{tikzpicture}
  \hspace{-.1cm}=
  \begin{tikzpicture}[baseline,shift={(0,-.5)}]
   \draw[thick,red] (.4,0) arc (180:0:.2cm);
   \draw[thick,red] (.6,.2) -- (.6,.4);
   \draw[thick,red] (0,0)-- (0,.6) arc (180:0:.2cm) arc (-180:0:.2cm) -- (.8,1.2);
   \draw[thick,red] (.2,.8) -- (.2,1);
   \draw[thick,red] (0,1.2) arc (-180:0:.2cm);
\end{tikzpicture}
=(( m^\dag m)\otimes 1)\circ(1\otimes( m^\dag m))\text{,}\]
completing the proof.
\end{proof}

This lemma immediately implies the following corollary together with commutativity of $A$ and unitary separability.

\begin{cor}
 \label{cor:condensationFoamNE}
There is a unique morphism $A^{\otimes n}\to A^{\otimes m}$ generated by condensation morphisms and braidings for which all $n+m$ inputs and outputs are connected.

In particular, for any $n$, an endomorphism of $A^{\otimes n}$ generated by condensation morphisms between adjacent $A$-strands and conjugation by braiding depends only on which strands are connected.
\end{cor}

Given a condensable algebra $A\in Z(\cX)$, anyons in the condensed phase are described by objects in $\Irr(Z(\cX)_A^{\loc})$ \cite{MR3246855,PhysRevB.79.045316}.
This UMTC is also the center of a fusion category, namely the category $\cX_A$ of right $A$-modules in $\cX$ \cite[Thm.~3.20]{MR3039775}, a fact we will frequently use below.
\begin{defn}
 If $\cX$ is a fusion category and $A\in Z(\cX)$ is a condensable algebra, the fusion category of \textit{right $A$-modules in $\cX$} \cite{MR1976459}\cite[\S~A.3]{MR3246855} has as objects pairs $(M,m)$, where $M\in\cX$ and $m:MA\to M$ is an \textit{action} morphism.
 We often denote such a pair by $M_A$, where the subscript denotes the existence of the $A$ action.

 Similar to the multiplication of $A$, we denote the action morphism diagrammatically by a trivalent vertex, as follows.
 \[m=\tikzmath{
  \draw (0,0) node[below]{$\scriptstyle M$} -- (0,1);
  \draw[red,thick] (.5,0) node[below]{$\scriptstyle A$} -- (.5,.25) arc (0:90:.5);
 }\]
 Here, red strands denote the algebra $A$.
 The action must satisfy the following associativity, unitality, and separability conditions (associative and separable actually implies unital) \cite[\S3.2]{MR4419534}.
 \[
  \tikzmath{
  \draw (0,-.5) -- (0,.5);
  \node at (0,-.7) {$\scriptstyle M$};
  \draw[red,thick] (0,.2) arc (90:0:.7cm);
  \draw[red,thick] (0,-.2) arc (90:0:.3cm);
  }
  =
  \tikzmath{
  \draw (0,-.5) -- (0,.5);
  \node at (0,-.7) {$\scriptstyle M$};
  \draw[red,thick] (0,.3) arc (90:0:.5cm);
  \draw[red,thick] (.2,-.5) arc (180:0:.3cm);
  }\,
  \qquad\qquad
  \tikzmath[xscale=-1]{
  \draw (0,-.5) -- (0,.5);
  \node at (0,-.7) {$\scriptstyle M$};
  \draw[red,thick] (0,0) -- (-.2,-.2);
  \filldraw[red] (-.2,-.2) circle (.05cm);
  }
  =
  \tikzmath{
  \draw (.2,-.5) -- (.2,.5);
  \node at (.2,-.7) {$\scriptstyle M$};
  }
  =
  \tikzmath{
  \draw (0,-.5) -- (0,.5);
  \node at (0,-.7) {$\scriptstyle M$};
  \draw[thick, red] (0,-.3) arc (-90:90:.3cm);
  }\,.
 \]
 The tensor product $M \otimes_A N_A$ of $A$-modules $M_A$ and $N_A$ is a subobject of $MN$, defined to be the image of the projection
 \[
  p_{M,N}
  :=
  \tikzmath[scale=.666]{
  \draw[thick,red] (0.5,-.4) -- (0.5,-.1);
  \draw[thick,red] (0.75,.15) arc (0:-180:.25cm);
  \filldraw[red] (.5,-.4) circle (.075cm);
  \draw[thick,red] (0.25,.15) arc (0:90:.25cm);
  \draw[thick,red,knot] (.75,.15) .. controls ++(90:.25cm) and ++(270:.25cm) .. (1.25,.65) arc (0:90:.25cm);
  \draw (0,-.6) -- (0,1.2);
  \draw[knot] (1,-.6) -- (1,1.2);
  }
  \in\End_\cC(M N),
 \]
 in $\cX(MN\to MN)$.
 Thus, the tensor unit of $\cX_A$ is $A_A$, where the $A$-module action on $A_A$ is defined to be the multiplication of $A$, and the associator and unitors of $\cX_A$ come from those of $\cX$.

\begin{rem}
The separability condition for $A$-modules ensures that $\cX_A$ is again unitary \cite[\S3.2]{MR4419534}.
Recall that $\cX_A$ can also be defined as the idempotent completion of the category of \emph{free} right $\cA$-modules of the form $x\otimes A$ for $x\in \cX$ where $A$ acts on the right using the multiplication on $A$.
Every free module is unitarily separable since $A$ is, and so the category of free right $A$-modules is a $\rm C^*$-category with finite dimensional hom spaces.
This means the unitary idempotent completion is equivalent to the ordinary idempotent completion.
\end{rem}

\end{defn}
Similarly, one can define the fusion category ${}_A\cX$ of left $A$-modules in $\cX$, and the UMTC $Z({}_A\cX)\cong{}_A^{\loc}Z(\cX)$.
However, there are canonical equivalences of UFCs ${}_A\cX\cong\cX_A$ and BFCs ${}_A^{\loc}Z(\cX)\cong Z(\cX)_A^{\loc}$.
Therefore, in the following, we will speak only of $\cX_A$ and $Z(\cX)_A^{\loc}\cong Z(\cX_A)$, even when $A$ acts on the left.

\subsection{General Lattice Model}
\label{ssec:condensationModel}
The essential idea behind this lattice model is to modify the original string-net models of \cite{PhysRevB.71.045110,PhysRevB.103.195155} to support a copy of the condensable algebra $A$ inside each plaquette, by adding appropriate local Hilbert spaces and modifying the Hamiltonian to account for the excitation, similar to the extended Levin-Wen models of \cite{PhysRevB.97.195154}.
The commuting projector Hamiltonian is then augmented with additional families of terms $C$ and $D$ depending on whether $Z(\cX)$ or $Z(\cX)_A^{\loc}$ topological order is desired.
Now that we have models of the condensed and uncondensed phases living in the same Hilbert space, one can smoothly pass from one Hamiltonian to the other through convex combinations.

Rather than give a construction for arbitrary lattice geometries, we will work out the details explicitly for a regular hexagonal lattice, and then describe how our construction must be adapted for other cases.
We realize the additional Hilbert space on each plaquette by adding an additional vertex and edge, as shown.
$$
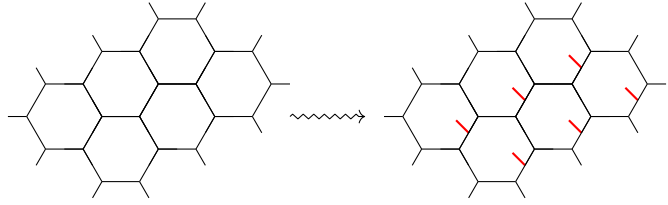
\tikzmath{
\levinHexGrid[]{0}{0}{2}{3}{.5}{black}
\draw [->, line join=round,
decorate, decoration={
    zigzag,
    segment length=4,
    amplitude=.9,post=lineto,
    post length=2pt
}]  (3,0) -- (4,0);
\levinHexGrid[thick, red]{5}{0}{2}{3}{.5}{black}
}
$$
The usual vertices of the plaquette are assigned the usual hom spaces of the Levin-Wen model,
and the new trivalent vertices correspond to the Hilbert space
\begin{align*}
\tikzmath{
\draw[thick, red] (0:0cm) node[below, xshift=.1cm] {$\scriptstyle v$} -- (-.3,.3) node[above]{$\scriptstyle A$} ;
\draw (-120:.5cm)node[below]{$\scriptstyle x$}  -- (60:.5cm) node[above]{$\scriptstyle y$};
}
\qquad
&\longleftrightarrow
\qquad
 \cH_v := \bigoplus_{x,y\in \Irr(\cX)}\cX(U(A)x\to y)\text{,}
\end{align*}
where $U: Z(\cX)\to \cX$ is the forgetful functor, and the inner product $\langle\cdot|\cdot\rangle_v$ on the orthogonal summand $\cX(U(A)x\to y)$ of $\cH_v$ is given by
\begin{equation}
 \label{eq:AVertexIP}
 \langle g|f\rangle_v\id_y=\sqrt{\frac{d_y}{d_x}}fg^\dag
\end{equation}
This construction resembles the extended Levin-Wen models of \cite{PhysRevB.97.195154}, but it is slightly different, because $A$ may contain more than one copy of a given simple object $x\in\Irr(\cX)$.
As with \eqref{eq:vertexIP}, the choice of normalization in \eqref{eq:AVertexIP} is necessary so that the forthcoming plaquette operator $B_p$ will be self-adjoint.

One should think of the red edges as leading to an $A$-defect at the center of each plaquette, i.e.~a puncture in the surface labelled by the object $A\in Z(\cX)$, as described in \cite{PhysRevB.97.195154,1106.6033}.
While before, we interpreted states of the lattice model as living in the diagrammatic calculus of $\cX$, the red edges should be thought of as living in the diagrammatic calculus of $Z(\cX)$, descending into the page (c.f.~Figure~\ref{fig:condensationMorphisms}).
This interpretation will determine how we extend the existing terms of the Hamiltonian to our new Hilbert space, as well as how we define the additional terms necessary to select either $Z(\cX)$ or $Z(\cX)_A^{\loc}$ topological order.

The Hamiltonian for our model will consist of the original $A_v$ and $B_p$ terms, modified to account for the new red links, as well as two new terms needed to implement anyon condensation.
The required modifications of the $A_v$ and $B_p$ terms to account for the red links are fairly minor.
The $A_\ell$ terms associated to each black link $\ell$ are defined as before, and no new $A$-term is associated to red links.
Recall that the $A$-term ensure the morphisms labelling the two vertices at either end of a link are composable, so that ground states for all $A$-terms can be locally interpreted as living in the diagrammatic calculus of $\cX$.
Since the red edge is only incident to a single vertex, no additional constraint is needed.

The $B_p$ term is defined as previously: it is $0$ outside of the ground state of nearby $A_\ell$ terms, and on their ground states, $B_p=\frac{1}{D_\cX}\sum_{x\in\Irr(\cX)}d_xB_p^x$.
However, the operators $B_p^x$ which insert a loop labelled by $x$ into the plaquette must now take into account the half-braiding on $A$.
That is,
\begin{equation}
B_p^x=
 \label{eq:newBps}
 \tikzmath{
 \coordinate (center) at ($ .7*(-.333,.333) +.3*(-.5,0)$);
 \coordinate (pointD) at (canvas polar cs:angle=-30,radius=.866*.5);
 \coordinate (pointE) at (canvas polar cs:angle=-30,radius=.288*.5);
  \draw[thick, red] ($(center)+(pointD)$) -- +($.5*(-.333, .333)$);
  \levinHex[red, thick]{0}{0}{.5}{black}
  \arHexKnot{0}{0}{.35}{blue}{4}{1}
  \node[blue] at (210:.18) {$\scriptstyle{x}$};
 }
\end{equation}

We omit the proof of the following lemma which is a straightforward adaptation of \cite[Thm.~5.0.1]{0907.2204}.

\begin{lem}
Given two simple tensors of orthonormal basis elements
$$
\xi 
=
\tikzmath{
\coordinate (center) at ($ .7*1.6666*(-.333,.333) +.3*(-.5,0)$);
\coordinate (pointD) at (canvas polar cs:angle=-30,radius=.866*.75);
\coordinate (pointE) at (canvas polar cs:angle=-30,radius=.288*.75);
\draw[thick, red] ($(center)+(pointD)$) -- +($.75*(-.333, .333)$);
\levinHex[red, thick]{0}{0}{.75}{black}
\node at (-1.3,0) {$\scriptstyle x_1$};
\node at (-.5,0) {$\scriptstyle \xi_1$};
\node at (-.65,-1.1) {$\scriptstyle x_2$};
\node at (-.25,-.45) {$\scriptstyle \xi_2$};
\node at (.65,-1.1) {$\scriptstyle x_3$};
\node at (.25,-.45) {$\scriptstyle \xi_3$};
\node[red] at (.15,0) {$\scriptstyle c_2$};
\node[red] at (.75,-.35) {$\scriptstyle \gamma_2$};
\node[red] at (-.8,.8) {$\scriptstyle c_1$};
\node[red] at (-.35,.35) {$\scriptstyle \gamma_1$};
\node at (1.3,0) {$\scriptstyle x_4$};
\node at (.5,0) {$\scriptstyle \xi_4$};
\node at (.65,1.1) {$\scriptstyle x_5$};
\node at (.25,.45) {$\scriptstyle \xi_5$};
\node at (-.65,1.1) {$\scriptstyle x_6$};
\node at (-.25,.85) {$\scriptstyle \xi_6$};
}
\qquad
\text{ and }
\qquad
\xi'
=
\tikzmath{
\coordinate (center) at ($ .7*1.6666*(-.333,.333) +.3*(-.5,0)$);
\coordinate (pointD) at (canvas polar cs:angle=-30,radius=.866*.75);
\coordinate (pointE) at (canvas polar cs:angle=-30,radius=.288*.75);
\draw[thick, red] ($(center)+(pointD)$) -- +($.75*(-.333, .333)$);
\levinHex[red, thick]{0}{0}{.75}{black}
\node at (-1.3,0) {$\scriptstyle x_1'$};
\node at (-.5,0) {$\scriptstyle \xi_1'$};
\node at (-.65,-1.15) {$\scriptstyle x_2'$};
\node at (-.25,-.45) {$\scriptstyle \xi_2'$};
\node at (.65,-1.1) {$\scriptstyle x_3'$};
\node at (.25,-.45) {$\scriptstyle \xi_3'$};
\node[red] at (.15,0) {$\scriptstyle c_2'$};
\node[red] at (.75,-.35) {$\scriptstyle \gamma_2'$};
\node[red] at (-.8,.8) {$\scriptstyle c_1'$};
\node[red] at (-.35,.35) {$\scriptstyle \gamma_1'$};
\node at (1.3,0) {$\scriptstyle x_4'$};
\node at (.5,0) {$\scriptstyle \xi_4'$};
\node at (.65,1.1) {$\scriptstyle x_5'$};
\node at (.25,.45) {$\scriptstyle \xi_5'$};
\node at (-.65,1.15) {$\scriptstyle x_6'$};
\node at (-.25,.85) {$\scriptstyle \xi_6'$};
}
$$
whose internal edge labels are consistent, we have
$$
\langle \xi | B_p \xi' \rangle
=
\prod_i\delta_{c_i=c_i'}\prod_j\delta_{x_j=x_j'}
\frac{1}{D_\cX\sqrt{d_{x_1}d_{x_2}d_{x_3}d_{x_4}d_{x_5}d_{x_6}}}
\tikzmath{
\begin{scope}[xshift=-2.25cm]
\draw[thick, red] (-.5383,-.3883) .. controls ++(60:1.1cm) and ++(120:1.1cm) .. ($ (.5383,-.3883) + (2.25,0) $);
\draw[thick, red] (.5383,.3883) .. controls ++(60:.5cm) and ++(120:.5cm) .. ($ (-.5383,.3883) + (2.25,0) $);
\draw[] ($ (.564,-.974) $) .. controls ++(-60:.5cm) and ++(-120:.5cm) .. ($ (-.564,-.974) + (2.25,0) $);
\draw[] ($ (.564,.974) $) .. controls ++(60:.5cm) and ++(120:.5cm) .. ($ (-.564,.974) + (2.25,0) $);
\draw[] ($ (-.564,-.974) $) .. controls ++(-120:1cm) and ++(-60:1cm) .. ($ (.564,-.974) + (2.25,0) $);
\draw[] ($ (-.564,.974) $) .. controls ++(120:1cm) and ++(60:1cm) .. ($ (.564,.974) + (2.25,0) $);
\draw[] (-1.1,0) .. controls ++(180:.5cm) and ++(180:.5cm) .. (-.5,-2) -- ($ (.5,-2) + (2.25,0) $) .. controls ++(0:.5cm) and ++(0:.5cm) .. ($ (1.1,0) + (2.25,0) $);
\node[red] at ($ (-1.1,.85) + (2.25,0) $) {$\scriptstyle c_1$};
\node[red] at ($ (-1.1,.45) + (2.25,0) $) {$\scriptstyle c_2$};
\levinHex[]{0}{0}{.75}{black,knot}
\node at (-.5,0) {$\scriptstyle \xi_4^\dag$};
\node at (-.25,-.45) {$\scriptstyle \xi_3^\dag$};
\node at (.25,-.45) {$\scriptstyle \xi_2^\dag$};
\node[red] at (-.75,-.35) {$\scriptstyle \gamma_2^\dag$};
\node[red] at (.3,.45) {$\scriptstyle \gamma_1^\dag$};
\node at (.5,0) {$\scriptstyle \xi_1^\dag$};
\node at (.25,.85) {$\scriptstyle \xi_6^\dag$};
\node at (-.25,.45) {$\scriptstyle \xi_5^\dag$};
\end{scope}
\levinHex[]{0}{0}{.75}{black,knot}
\node at (-1.1,.15) {$\scriptstyle x_1$};
\node at (-.5,0) {$\scriptstyle \xi_1'$};
\node at (-1.1,-1.1) {$\scriptstyle x_2$};
\node at (-.25,-.45) {$\scriptstyle \xi_2'$};
\node at (-1.1,-1.45) {$\scriptstyle \overline{x_3}$};
\node at (.25,-.45) {$\scriptstyle \xi_3'$};
\node[red] at (.75,-.35) {$\scriptstyle \gamma_2'$};
\node[red] at (-.3,.4) {$\scriptstyle \gamma_1'$};
\node at (-1.1,-1.8) {$\scriptstyle \overline{x_4}$};
\node at (.5,0) {$\scriptstyle \xi_4'$};
\node at (-1.1,1.8) {$\scriptstyle \overline{x_5}$};
\node at (.25,.45) {$\scriptstyle \xi_5'$};
\node at (-1.1,1.45) {$\scriptstyle x_6$};
\node at (-.25,.85) {$\scriptstyle \xi_6'$};
}\,.
$$
Thus the operator $B_p$ is a self-adjoint projector on the ground state space of $-\sum_\ell A_\ell$.
\end{lem}

Now we turn to the additional Hamiltonian terms, which pick out either desired topological order.
To obtain $Z(\cX)$ topological order, we introduce the term $\Cl$ for each red edge $\redl$.
If $v$ is the vertex incident to $\redl$, then $\Cl$ is local to $\cH_v$, and projects on the subspace generated by morphisms of the form $x\xrightarrow{\lambda_x^{-1}}1x\xrightarrow{u\otimes\id_x}Ax$,
where $u:1\to A$ is the unit of $A$.

We observe that $\Cl$ commutes with the $A_v$ and $B_p$ terms.
First, all pure tensors of morphisms are eigenstates for both $A_v$ and $\Cl$ operators, so they are diagonalized with respect to the same bases, and hence commute.
Second, for any $c\in\Irr(Z(\cX))$ and any morphism $f:c\to A$, the $B_p$ term preserves states where the morphism labeling the vertex incident to the red edge $\redl$ factors through $f$.
This is because morphisms in $Z(\cX)$ are, by definition, those that pass under crossings obtained from the half-braiding, including the crossing appearing in the definition \eqref{eq:newBps} of $B_p^x$.
In particular, we can take $f=u$, the unit of $A$, and observe that since $1_{Z(\cX)}\cong 1_\cX$ has a simple underlying object, this means that $B_p$ preserves the eigenspaces of $\Cl$.
Finally, distinct $\Cl$ terms are all disjointly supported, and hence commute.
Thus, the Hamiltonian
\[H=-\sum_{\ell}A_{\ell}-\sum_pB_p-\sum_{\redl}\Cl\] is a local commuting projector Hamiltonian.

To obtain $Z(\cX)_A^{\loc}$ topological order, we instead add the following term $D_{p,q}$ for every pair of adjacent plaquettes $p$ and $q$.
\begin{equation}
 \label{eq:DTerm}
 D_{p,q}=
 \tikzmath{
 \coordinate (a) at (.93,.15);
 \coordinate (b) at (1.8,.15);
 \coordinate (c) at (2.3,.5);
 \coordinate (m) at (2.1,.7);
 \coordinate (md) at (1.9,.9);
 \coordinate (d) at (1.9,1.2);
 \coordinate (e) at (1.7,.9);
 \coordinate (f) at (1.9,.35);
 \coordinate (g) at (.87,.25);
 \coordinate (h) at (.5,.45);
 \draw[thick, red] (.7,-.5) .. controls +(90:.3) and +(180:.2) .. (a) -- (b) .. controls +(0:.15) and +(-90:.5) .. (m);
 \draw[thick, red] (c) .. controls +(135:.2) and +(-45:.2) .. (m);
 \draw[thick, red] (m) -- (md);
 \draw[thick, red] (md) -- (d);
 \draw[thick, red] (md) .. controls +(135:.15) and +(90:.15) .. (e) .. controls +(-90:.2) and +(90:.2) .. (f) arc (0:-90:.1cm) -- (g) .. controls +(180:.2) and +(-90:.2) .. (h);
 \fill[white] (a) circle (.07cm);
 \fill[white] (g) circle (.07cm);
 \fill[red] (m) circle (.05cm);
 \fill[red] (md) circle (.05cm);
 \levinHexGrid[]{0}{0}{1}{2}{1}{black}
 \node at (0,0) {$\scriptstyle p$};
 \node at (1.5,.866) {$\scriptstyle q$};
 }
 :=
 \sum_{\textcolor{blue}{H},\textcolor{orange}{K}}
 \sum_{\textcolor{blue}{i},\textcolor{orange}{j}}
 \tikzmath{
 \coordinate (z) at (.72,-.2);
 \coordinate (a) at (.93,.15);
 \coordinate (b) at (1.8,.15);
 \coordinate (c) at (2.3,.5);
 \coordinate (bandhalf) at (2.05,.45);
 \coordinate (m) at (2.1,.7);
 \coordinate (md) at (1.9,.9);
 \coordinate (d) at (1.9,1.3);
 \coordinate (beforee) at (1.68,1.15);
 \coordinate (e) at (1.45,.9);
 \coordinate (f) at (1.8,.35);
 \coordinate (g) at (.87,.25);
 \coordinate (gandhalf) at (.55,.4);
 \coordinate (h) at (.5,.65);
 \draw[thick, red] (.7,-.5) -- (z);
 \draw[thick, blue] (z) .. controls +(90:.3) and +(180:.2) .. (a);
 \draw[thick, blue] (a) -- (b) .. controls +(0:.15) and +(-90:.35) .. (bandhalf);
 \draw[thick, red] (c) .. controls +(135:.2) and +(-45:.2) .. (m);
 \draw[thick, red] (bandhalf) -- (m) -- (md);
 \draw[thick, red] (md) -- (d);
 \draw[thick, red] (md) -- (beforee);
 \draw[thick, orange] (beforee) .. controls +(135:.45) and +(90:.15) .. (e) .. controls +(-90:.2) and +(90:.2) .. (f) arc (0:-90:.1cm) -- (g) .. controls +(180:.2) and +(-90:.2) .. (gandhalf);
 \draw[thick, red] (gandhalf) -- (h);
 \fill[white] (a) circle (.07cm);
 \fill[white] (g) circle (.07cm);
 \fill[red] (m) circle (.05cm);
 \fill[red] (md) circle (.05cm);
 \fill[blue] (z) circle (.05cm);
 \fill[blue] (bandhalf) circle (.05cm);
 \fill[orange] (beforee) circle (.05cm);
 \fill[orange] (gandhalf) circle (.05cm);
 \levinHexGrid[]{0}{0}{1}{2}{1}{black}
 \node at (0,0) {$\scriptstyle p$};
 \node at (1.5,.866) {$\scriptstyle q$};
 }\text{.}
\end{equation}
Here, the blue and orange vertices run over bases and dual bases of $\bigoplus_{c\in\Irr(Z(\mathcal{X}))}\cX(c\to A)$; hence, \eqref{eq:DTerm} can be resolved into local operators at each vertex in the same manner as string operators, given the additional data necessary to resolve the condensation morphism.
The second diagrammatic term in the definition is an appropriate interpretation of the first one by a variant of \eqref{eq:fusionDecomp}.

The $D_{p,q}$ terms commute with $A_\ell$ and $B_p$ terms by construction, i.e.~for the same reasons that string operators commute with $A_\ell$ and $B_p$ terms whenever $\ell$ and $p$ are far from the endpoints of the string.
Finally, terms $D_{p,q}$ and $D_{r,s}$ commute, by Lemma \ref{lem:DCommutesMath}.
Thus, the Hamiltonian
\[H=-\sum_vA_v-\sum_pB_p-\sum_{p,q}D_{p,q}\]
is a local commuting projector Hamiltonian.

The overall Hamiltonian for our lattice model which supports a phase transition is then given by
\begin{equation}
 \label{eq:newHamiltonian}
 H_t=-\sum_{v}A_v-\sum_{p}B_p-K\left((1-t)\sum_{v}C_v+t\sum_{p,q}D_{p,q}\right)
\end{equation}
where $K\gg1$.
In sections \S~\ref{ssec:t=0} and \S~\ref{ssec:condensedPhase} below, we will see that this Hamiltonian realizes $Z(\cX)$ topological order at $t=0$ and $Z(\cX)_A^{\loc}$ topological order at $t=1$.
The intended effect of choosing a large value for the constant $K$ is that at $t=0$ and $t=1$, the low-energy physics of the lattice model consists only of ground states for the $\Cl$ and $D_{p,q}$ terms respectively.
Therefore, in our analyses, we only consider such states.
In other words, one could analyze the models at $t=0$ and $t=1$ by projecting onto the space of ground states for each $\Cl$ and $D_{p,q}$, at which point the Hamiltonian would consist of only $A_\ell$ and $B_p$ terms.

\subsection{Topological order when \texorpdfstring{$t=0$}{t=0}}
\label{ssec:t=0}
In this section, we will argue that when $t=0$, the Hamiltonian \eqref{eq:newHamiltonian} gives rise to $Z(\cX)$-topological order.

Since $A$ is unital and the half-braiding on $1_{Z(\cX)}$ is trivial, projecting onto the ground state of $\Cl$ terms is equivalent to removing the red edges and the corresponding $A$-punctures from the lattice, leaving behind the original string-net model associated to $\cX$.
In the hexagonal model discussed above, if $v$ is the vertex incident to $\redl$, the image of $\Cl$ in $\cH_v$ is isomorphic to $\bigoplus_{a,b\in\Irr(\cX)}\mathcal{\cX}(a\to b)$.
Thus, the Hamiltonian on the ground state of $\Cl$ terms describes the Levin-Wen model on a heptagonal lattice, obtained by adding the additional vertex $v$ to each plaquette of the regular hexagonal lattice.
One can remove this vertex and return to the original hexagonal model in a straightforward manner.

\subsection{Topological Order when \texorpdfstring{$t=1$}{t=1}}
\label{ssec:condensedPhase}
In this section, we will show that the topological order when $t=1$ is described by $Z(\mathcal{X})_A^{\loc}$.
To do so, we will again introduce an algebra of local observables which acts on low-energy topological excitations of our lattice model, which we will call $\PreTubeAC$.
This algebra will be generated by $\Tube(\cX)$, which acts as before, as well as a new algebra $\Absorb(A)$, which consists of operators which fuse the localized excitation with an excitation pulled from the condensate $A$.

We will then describe the fates of string operators from our original model at $t=1$.
The definition of the operators operators $\sigma^s_p(\phi,\psi)$ in the original Levin-Wen model given in \S~\ref{ssec:stringOperators} still makes sense, with slight modifications to account for when the string crosses over a vertical red link, similar to how the definition of $B_p^s$ was modified.
However, there are two complications.
The first is that, at $t=1$, there is more freedom in choosing the local data $\phi$ and $\psi$ at the endpoints of the string, since operators in $\Absorb(A)$ can be applied to the excitation.
The second is that not all types of localized excitations obtainable from the operators $\sigma^s_p$ will be topological, i.e. anyons in the condensed phase, because $\sigma^s_p$ operators associated to anyon types $s$ which are not transparent to the condensate $A$ will not commute with the $D_{p,q}$ terms of the Hamiltonian.
Thus, at $t=1$, the operators $\sigma^s_p$ should be thought of as ``defect operators,'' which produce topological line defects along $p$, terminated by point defects at the endpoints of $p$.
Those defect operators which commute with $D_{p,q}$ terms will produce anyons at the end of trivial line defects, and thus become the string operators for anyons in the condensed phase.

Finally, we will show that, similar to the story in the original Levin-Wen model, states containing an isolated excitations at a fixed location form a representation of $\PreTubeAC$, and that topological excitations are representations of a quotient algebra $\Tube_A(\cX)$, which satisfy the additional relation \eqref{eq:tubeARel}.
It is expected that excitations in the condensed phase correspond to simple objects in $Z(\cX)_A^{\loc}$ \cite{MR3246855}.
We will show that this is the case, by checking that $\Tube_A(\cX)$ has the same representation theory as (i.e., is Morita equivalent to) $\Tube(\cX_A)$, and recalling that $Z(\cX)_A^{\loc}\cong Z(\cX_A)$ \cite[Thm.~3.20]{MR3039775}.
Representations of $\PreTubeAC$ will instead correspond to simple objects of $Z(\cX)_A$.
The details will be analogous to those of \S\ref{ssec:tubeImplementation}.

We will not characterize the line defects and point defects created by defect operators which do not correspond to anyons in the condensed phase.
However, the interpretation of simple objects of $Z(\cX)_A$ as point defects at the ends of nontrivial line defects in $Z(\cX)_A^{\loc}$ topological order is explored algebraically in \cite[\S~IV.D]{2208.14018}.

\subsubsection{Tube algebras of local operators when \texorpdfstring{$t=1$}{t=1}}
\label{sssec:tubeA}
We begin by defining the algebra $\PreTubeAC$ of local observables near a localized excitation in the condensed phase, as well as the quotient $\Tube_A(\cX)$ of which topological representations are excitations.
\begin{defn}
  \label{def:tubeA}
  The C* algebra $\PreTubeAC$ has the underlying vector space
  \[\PreTubeAC=\bigoplus_{x,y,c\in\Irr(\mathcal{X})}(cx\to U(A)yc)\text{,}\]
  In the definition that follows we abuse notation and denote $U(A)$ by $A$. The multiplication on $\PreTubeAC$ is defined as
  \[
\tikzmath[xscale=-1]{
\draw (-.2,-.7) node[right, yshift=.2cm]{$\scriptstyle z$} -- (-.2,0);
\draw (.2,-.7) node[left, yshift=.2cm]{$\scriptstyle d$} -- (.2,0);
\draw (-.2,.7) node[right, yshift=-.2cm]{$\scriptstyle d$} -- (-.2,0);
\draw (.2,.7) node[left, yshift=-.2cm]{$\scriptstyle w$} -- (.2,0);
\draw[thick, red] (.4,.2) -- node[left]{$\scriptstyle A$} (.9,.7);
\roundNbox{fill=white}{(0,0)}{.3}{.2}{.2}{$\psi$}
}
\cdot
\tikzmath[xscale=-1]{
\draw (-.2,-.7) node[right, yshift=.2cm]{$\scriptstyle x$} -- (-.2,0);
\draw (.2,-.7) node[left, yshift=.2cm]{$\scriptstyle c$} -- (.2,0);
\draw (-.2,.7) node[right, yshift=-.2cm]{$\scriptstyle c$} -- (-.2,0);
\draw (.2,.7) node[left, yshift=-.2cm]{$\scriptstyle y$} -- (.2,0);
\draw[thick, red] (.4,.2) -- node[left]{$\scriptstyle A$} (.9,.7);
\roundNbox{fill=white}{(0,0)}{.3}{.2}{.2}{$\phi$}
}
:=
\delta_{y=z}
\sum_{f\in \Irr(\cX)}
\tikzmath[xscale=-1]{
\draw[thick, red] (.2,.7) -- (.45,1.1);
\draw[thick, red] (.2,-.3) to[rounded corners=5] (.45,.1) to[sharp corners] (.45,1.1) -- (.55,1.4) ;
\filldraw[red] (.45,1.1) circle (.05cm);
\draw (-.17,-.2) to[rounded corners=5] (-.45,.1) to[sharp corners] (-.45,1.1) -- (-.55,1.4);
\draw (-.15,.8) to[bend right=25] (-.45,1.1);
\draw[knot] (.17,.2) to[rounded corners=5] (.45,-.1) to[sharp corners] (.45,-1.1) -- (.55,-1.4);
\draw (.15,-.8) to[bend right=25] (.45,-1.1);
\draw[fill=blue] (-.45,1.1) circle (.05cm);
\draw[fill=blue] (.45,-1.1) circle (.05cm);
	\draw (0,-1.4) -- node[left, xshift=.1cm]{$\scriptstyle y$} (0,1.4);
\roundNbox{fill=white}{(0,.5)}{.3}{0}{0}{$\psi$};
\roundNbox{fill=white}{(0,-.5)}{.3}{0}{0}{$\phi$};
	\node at (-.58,1.6) {$\scriptstyle f$};
	\node at (.58,-1.6) {$\scriptstyle f$};
	\node at (-.25,1.25) {$\scriptstyle d$};
	\node at (.25,-1.25) {$\scriptstyle c$};
	\node at (-.58,.5) {$\scriptstyle c$};
	\node at (.58,-.5) {$\scriptstyle d$};
	\node at (0,1.6) {$\scriptstyle w$};
	\node at (0,-1.6) {$\scriptstyle x$};
	\node[red] at (.58,1.6) {$\scriptstyle A$};
}
  \]
  and the involution is given by
  \[
\left(
\tikzmath[xscale=-1]{
\draw (-.2,-.7) node[right, yshift=.2cm]{$\scriptstyle x$} -- (-.2,0);
\draw (.2,-.7) node[left, yshift=.2cm]{$\scriptstyle c$} -- (.2,0);
\draw (-.2,.7) node[right, yshift=-.2cm]{$\scriptstyle c$} -- (-.2,0);
\draw (.2,.7) node[left, yshift=-.2cm]{$\scriptstyle y$} -- (.2,0);
\draw[thick, red] (.4,.2) -- node[left]{$\scriptstyle A$} (.9,.7);
\roundNbox{fill=white}{(0,0)}{.3}{.2}{.2}{$\phi$}
}
\right)^*
:=
 \tikzmath{
\draw[thick, red] (-.4,-.2) -- (-.5,-.3) arc (-45:-180:.3cm) -- node[left]{$\scriptstyle A$} (-1.012,.7);
\draw (-.2,-.7) node[right, yshift=.2cm]{$\scriptstyle y$} -- (-.2,0);
\draw (.2,-.3) arc (-180:0:.3cm) -- node[right]{$\scriptstyle \overline{c}$} (.8,.7);
\draw[knot] (-.2,.3) arc (0:180:.3cm) -- node[left, yshift=-.25cm]{$\scriptstyle \overline{c}$} (-.8,-.7);
\draw (.2,.7) node[right, yshift=-.2cm]{$\scriptstyle x$} -- (.2,0);
\roundNbox{fill=white}{(0,0)}{.3}{.2}{.2}{$\phi^\dag$}
}\,.
  \]
  The algebra $\Tube_A(\mathcal{X})$ is the quotient of $\PreTubeAC$ by the relation
   \begin{equation}
   \label{eq:tubeARel}
\tikzmath[xscale=-1]{
\draw (-.2,-1.7) -- node[right]{$\scriptstyle z$} (-.2,-.3);
\draw (.2,-.7) -- node[left]{$\scriptstyle d$} (.2,-.3);
\draw (-.2,.8) -- node[right]{$\scriptstyle d$} (-.2,.3);
\draw (.2,.8) -- node[right]{$\scriptstyle w$} (.2,.3);
\draw (.4,-1.3) -- node[left]{$\scriptstyle d$} (.4,-1.7);
\draw[thick, red] (.4,.2) -- (.6,.4) -- node[left]{$\scriptstyle A$} (.6,.8);
\draw[thick, red] (.6,-.7) -- (.6,.4);
\filldraw[red] (.6,.4) circle (.05cm);
\roundNbox{fill=white}{(0,0)}{.3}{.2}{.2}{$\psi$}
\roundNbox{fill=white}{(.4,-1)}{.3}{.1}{.1}{$h$}
}
\approx
\tikzmath[xscale=-1]{
\draw[thick, red] (.5,.2) to[rounded corners=5] (.7,.7) to[sharp corners] (.6,1.7) -- node[left]{$\scriptstyle A$} (.6,2) ;
\draw[thick, red] (-.1,1.2) -- (.6,1.7);
\filldraw[red] (.6,1.7) circle (.05cm);
\draw (-.3,-.7) -- node[right]{$\scriptstyle z$} (-.3,-.3);
\draw (.3,-.7) -- node[left]{$\scriptstyle d$} (.3,-.3);
\draw (-.3,2) -- node[right]{$\scriptstyle d$} (-.3,1.3) -- (-.3,.8) -- node[right]{$\scriptstyle d$} (-.3,.3);
\draw[knot] (.3,2) node[right, yshift=-.2cm]{$\scriptstyle w$} -- (.3,.3);
\roundNbox{fill=white}{(0,0)}{.3}{.3}{.3}{$\psi$}
\roundNbox{fill=white}{(-.3,1)}{.3}{0}{0}{$h$}
}
  \end{equation}
 Is is straightforward to verify that this relation is $\ast$-closed.
\end{defn}

As promised, we now check that representations of $\Tube_A(\mathcal{X})$ correspond to simple objects of $Z(\cX)_A^{\loc}\cong Z(\cX_A)$, which is the UMTC which should describe topological excitations at $t=1$.
Because we have a monoidal equivalence $\Rep(\Tube(\cX))\cong Z(\cX)$, there is also a monoidal equivalence $\Rep(\Tube(\cX_A))\cong Z(\cX_A)$.
However, $\PreTubeAC$, rather than $\Tube(\cX_A)$, is the algebra which acts on local excitations of our model at $t=1$.
We will therefore need the following lemma, which follows from straightforward application of techniques from \cite{MR3447719}.
\begin{lem}
 \label{lem:tubeAMoritaClass}
 The algebras $\Tube(\mathcal{X}_A)$ and $\Tube_A(\mathcal{X})$ are Morita equivalent.
\end{lem}
\begin{proof}
 An algebra $T_1$ which is Morita equivalent to $\Tube(\mathcal{X}_A)$ can be obtained by replacing the underlying vector space with $\bigoplus_{x,y\in\Irr(\mathcal{X}),{}_AM\in\Irr(\mathcal{X}_A)}\mathcal{X}_A({}_AM\otimes_AAx\to {}_AAy\otimes_AM)$, while keeping the same diagrammatic multiplication from \ref{def:tubeA}.
 That $T_1$ is Morita equivalent to $\Tube(\mathcal{X}_A)$ follows from \cite[Theorem 4.2]{MR3447719}.

 As described at the end of \cite[\S3]{MR3447719}, the annular algebra $T_1$ can be obtained as the quotient of a much larger algebra $\widetilde{T_1}$.
 The underlying vectorspace of $\widetilde{T_1}$ is
 \[\bigoplus_{x,y\in\Irr(\mathcal{X}), {}_AM\in\mathcal{X}_A}\mathcal{X}_A({}_AM\otimes_AAx\to {}_AAy\otimes_AM)\text{,}\]
 with the multiplication
 \[\mathcal{X}_A({}_AN\otimes_AAy,{}_AAz\otimes_AN)\otimes\mathcal{X}_A({}_AM\otimes_AAx,{}_AAy\otimes_AM)\to\mathcal{X}_A({}_AN\otimes_AM\otimes_AAx,{}_AAz\otimes_AN\otimes_AM)\]
 given by joining the ${}_AAy$ strings and tensoring the others.
 To obtain $T_1$ from $\widetilde{T_1}$, we impose relation \eqref{eq:aroundBack}, which is reproduced from \cite[p.10]{MR3447719}, allowing morphisms to be pulled around the back of the tube.
 \begin{equation}
  \label{eq:aroundBack}
  \begin{tikzpicture}[baseline=-.1cm]
   \draw(-.35,-1.27) -- node[right]{$\scriptstyle xA$} (-.35,-.6);
   \draw(-.35,.73) -- node[right]{$\scriptstyle yA$} (-.35,-.05);
   \node at (0,-.2) {$\scriptstyle M$};
   \node at (.75,-.45) {$\scriptstyle N$};
   \draw[thick] (-.9,-1) -- (-.9,1);
   \draw[thick] (.9,-1) -- (.9,1);
   \draw[very thick] (0,1) ellipse (.9cm and .3cm);
   \halfDottedEllipse[very thick]{(-.9,-1)}{.9}{.3}
   \halfDottedEllipse[]{(-.9,-.1)}{.9}{.3}
   \roundNbox{unshaded}{(-.4,-.35)}{.3}{-.1}{-.1}{$f$}
   \roundNbox{unshaded}{(.4,-.35)}{.3}{-.1}{-.1}{$g$}
  \end{tikzpicture}
  =
  \begin{tikzpicture}[xscale=-1, baseline=-.1cm]
   \draw(-.35,-1.27) -- node[left]{$\scriptstyle xA$} (-.35,-.6);
   \draw(-.35,.73) -- node[left]{$\scriptstyle yA$} (-.35,-.05);
   \node at (0,-.2) {$\scriptstyle N$};
   \node at (.75,-.45) {$\scriptstyle M$};
   \draw[thick] (-.9,-1) -- (-.9,1);
   \draw[thick] (.9,-1) -- (.9,1);
   \draw[very thick] (0,1) ellipse (.9cm and .3cm);
   \halfDottedEllipse[very thick]{(-.9,-1)}{.9}{.3}
   \halfDottedEllipse[]{(-.9,-.1)}{.9}{.3}
   \roundNbox{unshaded}{(-.4,-.35)}{.3}{-.1}{-.1}{$f$}
   \roundNbox{unshaded}{(.4,-.35)}{.3}{-.1}{-.1}{$g$}
  \end{tikzpicture}
 \end{equation}
 Every ${}_AM\in\mathcal{X}_A$ is the direct sum of irreducible objects, so we may rewrite $\id_{{}_AM}=\sum_i\pi_i^\dag\circ\pi_i$, where each $\pi_i$ is a projection onto a simple object.
 Pulling each $\pi_i$ around the back via \eqref{eq:aroundBack}
 lets us identify the elements of $\widetilde{T_1}$ with those of $T_1$, and recovers the multiplication for $T_1$.

 Instead of going directly from $\widetilde{T_1}$ to $T_1$, we can also consider an in-between algebra $T_2$, which we will later identify with an algebra containing $\Tube_A(\mathcal{X})$.
 A subalgebra of $\widetilde{T_1}$ is
 \[\widetilde{T_2}=\bigoplus_{x,y\in\Irr(\mathcal{X}), C\in\mathcal{X}}\mathcal{X}_A({}_AAC\otimes_AAx\to {}_AAy\otimes_AAC)\text{.}\]
 A special case of \eqref{eq:aroundBack} occurs when the morphism $g$ is in the image of the free module functor; imposing relation \eqref{eq:aroundBack} for only such $g$ on $\widetilde{T_2}$ produces an algebra
 \[T_2\cong\bigoplus_{x,y,c\in\Irr(\mathcal{X})}\mathcal{X}_A({}_AAc\otimes_AAx\to {}_AAy\otimes_AAc)\text{,}\]
 where the multiplication now involves using fusion channels in $\mathcal{X}$ to decompose the strands wrapping around the tube.
 Fully imposing \eqref{eq:aroundBack} on $T_2$ gives all of $T_1$ as a quotient, since free modules span all of $\mathcal{X}$ under direct sum.

 We can translate hom spaces between free modules in $\mathcal{X}_A$ to hom spaces in $\mathcal{X}$, because the free module functor $x\mapsto {}_AAx$ is monoidal and adjoint to the forgetful functor:
 \[\mathcal{X}_A({}_AAc\otimes_AAx\to {}_AAy\otimes_AAc)\cong\mathcal{X}_A({}_AAcx\to {}_AAyc)\cong\mathcal{X}(cx\to Ayc)\text{.}\]
 This gives an isomorphism of vector spaces from algebra $T_2$ to $\PreTubeAC$:
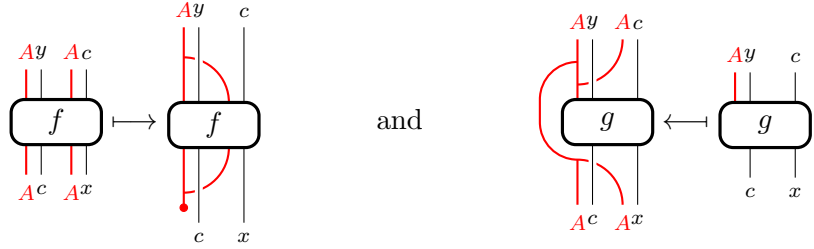
\[
\tikzmath{
\draw[thick, red] (-.4,-.7) node[below]{$\scriptstyle A$} -- (-.4,.7) node[above]{$\scriptstyle A$};
\draw (-.2,-.7) node[below]{$\scriptstyle c$} -- (-.2,.7) node[above]{$\scriptstyle y$};
\draw[thick, red] (.2,-.7) node[below]{$\scriptstyle A$} -- (.2,.7) node[above]{$\scriptstyle A$};
\draw (.4,-.7) node[below]{$\scriptstyle x$} -- (.4,.7) node[above]{$\scriptstyle c$};
\roundNbox{fill=white}{(0,0)}{.3}{.3}{.3}{$f$}
}
\longmapsto
\tikzmath{
\draw[thick, red] (-.4,-1.1) -- (-.4,1.3) node[above]{$\scriptstyle A$};
\draw[thick, red] (-.4,.9) arc (90:0:.6cm) -- (.2,-.3) arc (0:-90:.6cm);
\filldraw[red] (-.4,-1.1) circle (.05cm);
\draw[knot] (-.2,-1.3) node[below]{$\scriptstyle c$} -- (-.2,1.3) node[above]{$\scriptstyle y$};
\draw (.4,-1.3) node[below]{$\scriptstyle x$} -- (.4,1.3) node[above]{$\scriptstyle c$};
\roundNbox{fill=white}{(0,0)}{.3}{.3}{.3}{$f$}
}
\qquad\qquad
\text{and}
\qquad\qquad
\tikzmath{
\draw[thick, red] (-.4,.3) -- (-.4,1.1) node[above]{$\scriptstyle A$};
\draw[thick, red] (-.4,.5) arc (-90:0:.6cm) node[above]{$\scriptstyle A$};
\draw[thick, red] (-.4,.8) arc (90:180:.5cm) (-.9,.3) -- (-.9,0) arc (180:270:.5cm) (-.4,-.5) arc (90:0:.6cm) node[below]{$\scriptstyle A$};
\draw[thick, red] (-.4,-.5) -- (-.4,-1.1) node[below]{$\scriptstyle A$};
\draw[knot] (-.2,-1.1) node[below]{$\scriptstyle c$} -- (-.2,1.1) node[above]{$\scriptstyle y$};
\draw (.4,-1.1) node[below]{$\scriptstyle x$} -- (.4,1.1) node[above]{$\scriptstyle c$};
\roundNbox{fill=white}{(0,0)}{.3}{.3}{.3}{$g$}
}
\longmapsfrom
\tikzmath{
\draw[thick, red] (-.4,0) -- (-.4,.7) node[above]{$\scriptstyle A$};
\draw (-.2,-.7) node[below]{$\scriptstyle c$} -- (-.2,.7) node[above]{$\scriptstyle y$};
\draw (.4,-.7) node[below]{$\scriptstyle x$} -- (.4,.7) node[above]{$\scriptstyle c$};
\roundNbox{fill=white}{(0,0)}{.3}{.3}{.3}{$g$}
}
\]
 Carrying over the multiplication from $T_2$ gives the algebra structure on $\PreTubeAC$ described in Definition \ref{def:tubeA}, and imposing \eqref{eq:aroundBack} on both algebras gives the quotient $T_1\cong\Tube_A(\mathcal{X})$.
 Since $\Tube(\cX_A)$ was Morita equivalent to $T_1$, this shows that $\Tube(\cX_A)$ is Morita equivalent to $\Tube_A(\cX)$.
\end{proof}

We make one more preliminary observation: that, as previously described, $\PreTubeAC$ is generated by two subalgebras.
One is $\Tube(\mathcal{X})$, which lives inside $\PreTubeAC$ as a subalgebra, along the unit map $1\to A$.
The other (nonunital) subalgebra of $\widetilde{\Tube_A(\mathcal{X})}$ is
\[\Absorb(A)=\bigoplus_{x,y}\mathcal{X}(1x\to U(A)y1)\text{,}\]
the subalgebra where no string runs around the circumference of the tube.
Since elements of $\Absorb(A)$ do not have a string running around the back of the tube, it is straightforward to compute $f\cdot\phi$ where $f\in\Absorb(A)$ and $\phi\in\Tube(\cX)$, and therefore to see that indeed, $\PreTubeAC\cong\Absorb(A)\Tube(\cX)$.

\subsubsection{String Operators when \texorpdfstring{$t=1$}{t=1}}
\label{sssec:newStrings}
The string operators described in \S~\ref{ssec:stringOperators} can be extended to the models of \S~\ref{ssec:condensationModel} in a straightforward way.
Namely, when a string crosses a vertical link inside a plaquette, we must apply the half-braiding of the algebra $A$, as we did when defining $B_p$.
Thus, the string operators $\sigma^s_p(\phi,\psi)$ previously defined still make sense as operators which commute with $A_\ell$ and $B_q$ terms as long as $\ell$ and $q$ are far from the endpoints of $p$.

However, in the condensed phase, the data needed to terminate a string operator is different.
Because each plaquette now contains a vertical link supporting a copy of the condensate $A$, rather than pick elements $\phi,\psi\in\bigoplus_{x\in\Irr(\cX)}\cX(s\to x)$, we define a string operator $\sigma^s_p(\phi,\psi)$ where $\phi,\psi\in\bigoplus_{x\in\Irr(\cX)}\cX(s\to Ax)$.
The $x$ factor remains at the endpoint of $p$, as before, while the $A$-factor is multiplied into a nearby vertical link.
Of course, a morphism $\psi:s\to Ax$, can be factorized as $(1_A\otimes eta)\circ a$, where $a:s\to At$ for some $t\in\Irr(Z(\cX))$ and $\eta:t\to x$.
Thus, the end of a string operator now takes the form
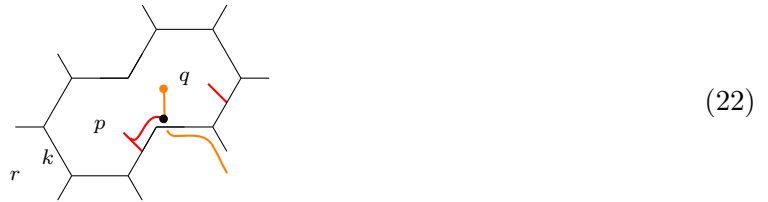
\begin{equation}
 \label{eq:condensedStringEnding}
 \tikzmath{
 \coordinate (a) at (1.5,-.25);
 \coordinate (b) at (1,-.12);
 \coordinate (c) at ($(20:.9cm)+(0,.2)$);
 \coordinate (g) at ($(c)+(0,-.4)$);
 \coordinate (z) at (-20:1.8cm);
 \filldraw[orange] (c) circle (.05cm);
 \draw[orange,thick] (z) .. controls ++(120:.2cm) and ++(-60:.2cm) .. (a) .. controls ++(120:.2cm) and ++(0:.2cm) .. (b) .. controls ++(180:.2cm) and ++(-80:.2cm)   .. (c);
 \draw[red,thick] (g) .. controls ++(150:.3cm) and ++(-45:.2cm) .. ($(-30:.866)+(-.345,.285)$); 
 \filldraw (g) circle (.05cm);
 \fill[white] (0:.87) circle (.05cm); 
 \levinHexOpen[red, thick]{0}{0}{.75}{black}{1}
 \levinHexOpen[red, thick]{1.125}{.6495}{.75}{black}{4}
 \node at (210:1.299) {$\scriptstyle r$};
 \node at (210:.8) {$\scriptstyle k$};
 \node at (0:0) {$\scriptstyle p$};
 \node at (1.125,.6495) {$\scriptstyle q$};
 }
\end{equation}
where the black vertex is a morphism $s\to At$, and the orange vertex is the morphism $\eta:t\to x$, as previously.

An important consequence is that the string operators $\sigma^s$ should no longer be viewed as a string operator for the anyon $s\in\Irr(Z(\cX))$, but instead as a string operator for the free module $sA_A\in Z(\cX)_A$.
This is because, in the ground state of all $D_{p,q}$ terms, the trivalent vertex $a\in Z(\cX)(s\to At)$ is able to slide topologically along the string.
Suppose $\sigma^t_r(\phi,\psi)$ is the following string operator.
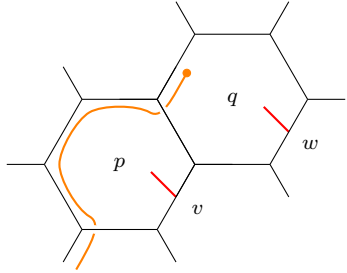
\begin{equation*}
 \tikzmath{ 
  \coordinate(c2) at (1.125*1.333,.6495*1.333);
  \coordinate (e) at ($(c2)+(150:.7)$);
  \draw[orange,thick] ($(-120:1.5)+(-30:.2)$) .. controls ++(60:.4) and ++(-30:.2) .. (-120:.8) .. controls ++(150:.2) and ++(90:-.2) .. (180:.8) .. controls ++(90:.2) and ++(30:-.2) .. (120:.8) .. controls ++(30:.2) and ++(-30:-.2) .. (60:.8) .. controls ++(-30:.2) and ++(60:-.4) .. (e);
  \filldraw[orange] (e) circle (.05cm);
  \filldraw[white] ($(-120:1)+(0:.17)$) circle (.05cm); 
  \filldraw[white] ($(60:1)+(-60:.17)$) circle (.05cm); 
  \levinHex[red, thick]{0}{0}{1}{black}
  \levinHex[red, thick]{1.125*1.333}{.6495*1.333}{1}{black}
  \node at (0,0) {$\scriptstyle p$};
  \node at (c2) {$\scriptstyle q$};
  \node at (-30:1.2) {$\scriptstyle v$};
  \node at ($(c2)+(-30:1.2)$) {$\scriptstyle w$};
 }
\end{equation*}
Now suppose that $f\in Z(\cX)(s\to At)$.
We can modify $\sigma^t_r(\phi,\psi)$ to obtain a different operator by applying $f$ at various possible locations along $p$, and then multiplying the resulting $A$ string into the vertical link in the plaquette where $f$ was applied.
By construction of our string operators, on local ground states, we can clearly slide $f$ topologically along the string within each plaquette.
However, moving from one plaquette to another will change which vertical link $f$ interacts with.
The key point is that on ground states of relevant $D_{p,q}$ terms, this makes no difference, essentially by Corollary~\ref{cor:condensationFoamNE}.
For example, observe that
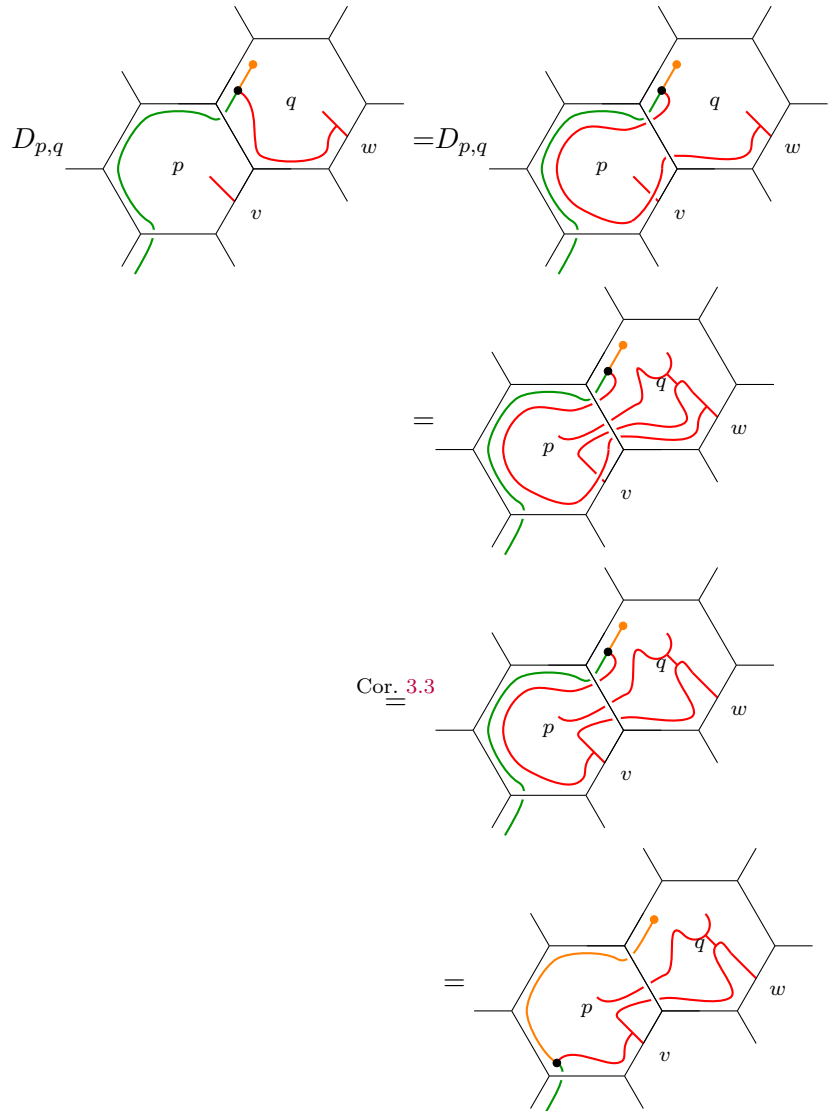
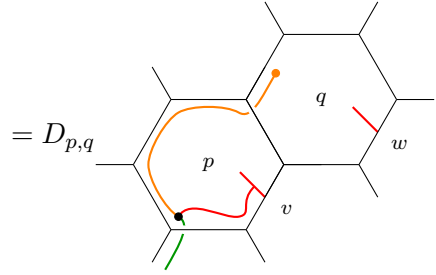
\begin{align*}
 D_{p,q}
 \tikzmath{
  \coordinate(c2) at (1.125*1.333,.6495*1.333);
  \coordinate (e) at ($(c2)+(150:.7)+(60:.2)$);
  \coordinate (g2) at ($(e)+(60:-.4)$);
  \coordinate (m) at ($(c2)+(-30:1*.866)+(-45:-.2)$);
  \draw[DarkGreen,thick] ($(-120:1.5)+(-30:.2)$) .. controls ++(60:.4) and ++(-30:.2) .. (-120:.8) .. controls ++(150:.2) and ++(90:-.2) .. (180:.8) .. controls ++(90:.2) and ++(30:-.2) .. (120:.8) .. controls ++(30:.2) and ++(-30:-.2) .. (60:.8) .. controls ++(-30:.2) and ++(60:-.4) .. (g2);
  \draw[orange,thick] (g2) -- (e);
  \draw[red,thick] (g2) .. controls ++(-30:.4) and ++(-30:-.2) .. ($(c2)+(-120:.8)$) .. controls ++(-30:.2) and ++(30:-.2) .. ($(c2)+(-60:.8)$) .. controls ++(30:.2) and ++(-135:.3) .. (m);
  \filldraw[black] (g2) circle (.05cm);
  \filldraw[orange] (e) circle (.05cm);
  \filldraw[white] ($(-120:1)+(0:.17)$) circle (.05cm); 
  \filldraw[white] ($(60:1)+(-60:.17)$) circle (.05cm); 
  \levinHex[red, thick]{0}{0}{1}{black}
  \levinHex[red, thick]{1.125*1.333}{.6495*1.333}{1}{black}
  \node at (0,0) {$\scriptstyle p$};
  \node at (c2) {$\scriptstyle q$};
  \node at (-30:1.2) {$\scriptstyle v$};
  \node at ($(c2)+(-30:1.2)$) {$\scriptstyle w$};
 }
 =&
 D_{p,q}
 \tikzmath{
  \coordinate(c2) at (1.125*1.333,.6495*1.333);
  \coordinate (e) at ($(c2)+(150:.7)+(60:.2)$);
  \coordinate (g2) at ($(e)+(60:-.4)$);
  \coordinate (m) at ($(c2)+(-30:1*.866)+(-45:-.2)$);
  \draw[DarkGreen,thick] ($(-120:1.5)+(-30:.2)$) .. controls ++(60:.4) and ++(-30:.2) .. (-120:.8) .. controls ++(150:.2) and ++(90:-.2) .. (180:.8) .. controls ++(90:.2) and ++(30:-.2) .. (120:.8) .. controls ++(30:.2) and ++(-30:-.2) .. (60:.8) .. controls ++(-30:.2) and ++(60:-.4) .. (g2);
  \draw[orange,thick] (g2) -- (e);
  \draw[red,thick] (g2) .. controls ++(-30:.4) and ++(0:.2) .. (60:.6) .. controls ++(0:-.2) and ++(30:.2) .. (120:.6) .. controls ++(30:-.2) and ++(90:.2) .. (180:.6) .. controls ++(90:-.2) and ++(-30:-.2) .. (-120:.6) .. controls ++(-30:.2) and ++(30:-.2) .. (-60:.8);
  \draw[red,thick] (0:.8) .. controls ++(90:.2) and ++(0:-.2) .. ($(c2)+(-120:.8)$) .. controls ++(0:.2) and ++(30:-.2) .. ($(c2)+(-60:.8)$) .. controls ++(30:.2) and ++(-135:.3) .. (m);
  \filldraw[black] (g2) circle (.05cm);
  \filldraw[orange] (e) circle (.05cm);
  \filldraw[white] ($(-120:1)+(0:.17)$) circle (.05cm); 
  \filldraw[white] ($(60:1)+(-60:.17)$) circle (.05cm); 
  \filldraw[white] ($(60:1)+(-60:.27)$) circle (.05cm); 
  \filldraw[white] ($(0:1)+(120:.19)$) circle (.05cm); 
  \levinHex[red, thick]{0}{0}{1}{black}
  \levinHex[red, thick]{1.125*1.333}{.6495*1.333}{1}{black}
  \draw[red,thick,knot] (-60:.8) .. controls ++(30:.2) and ++(-90:.2) .. (0:.8); 
  \node at (0,0) {$\scriptstyle p$};
  \node at (c2) {$\scriptstyle q$};
  \node at (-30:1.2) {$\scriptstyle v$};
  \node at ($(c2)+(-30:1.2)$) {$\scriptstyle w$};
 }
 \displaybreak[1]\\=&
 \tikzmath{
  \coordinate(c2) at (1.125*1.333,.6495*1.333);
  \coordinate (e) at ($(c2)+(150:.7)+(60:.2)$);
  \coordinate (g2) at ($(e)+(60:-.4)$);
  \coordinate (m) at ($(c2)+(-30:1*.866)+(-45:-.2)$);
  \coordinate (loffs) at ($(-30:1*.866)+(-.333,.333)$);
  \coordinate (moffs) at (-.2,.1);
  \coordinate (mout) at ($(c2)+(loffs)+(moffs)$);
  \draw[DarkGreen,thick] ($(-120:1.5)+(-30:.2)$) .. controls ++(60:.4) and ++(-30:.2) .. (-120:.8) .. controls ++(150:.2) and ++(90:-.2) .. (180:.8) .. controls ++(90:.2) and ++(30:-.2) .. (120:.8) .. controls ++(30:.2) and ++(-30:-.2) .. (60:.8) .. controls ++(-30:.2) and ++(60:-.4) .. (g2);
  \draw[orange,thick] (g2) -- (e);
  \draw[red,thick] (g2) .. controls ++(-30:.4) and ++(0:.2) .. (60:.6) .. controls ++(0:-.2) and ++(30:.2) .. (120:.6) .. controls ++(30:-.2) and ++(90:.2) .. (180:.6) .. controls ++(90:-.2) and ++(-30:-.2) .. (-120:.6) .. controls ++(-30:.2) and ++(30:-.2) .. (-60:.8);
  \draw[red,thick] (0:.8) .. controls ++(90:.2) and ++(0:-.2) .. ($(c2)+(-120:.8)$) .. controls ++(0:.2) and ++(30:-.2) .. ($(c2)+(-60:.8)$) .. controls ++(30:.2) and ++(-135:.3) .. (m);
  \draw[red,thick] (loffs) .. controls ++(135:.2) and ++(0:-.2) .. ($(c2)+(-120:.6)$) .. controls ++(0:.2) and ++(30:-.2) .. ($(c2)+(-60:.6)$) .. controls ++(30:.2) and ++(45:-.2) .. (mout) .. controls ++(45:.1) .. ($(c2)+(loffs)$);
  \draw[red,thick] (mout) -- ++(-45:-.2);
  \draw[red,thick] ($(mout)+(-45:-.2)$) arc (-45:45:.2);
  \draw[red,thick] ($(mout)+(-45:-.2)$) arc (-45:-135:.2) .. controls ++(135:.4) and ++(20:.2) .. ($(c2)+(-140:.4)$) .. controls ++(20:-.2) and ++(135:-.3) .. ($(loffs)+(-45:-.4)$);
  \filldraw[black] (g2) circle (.05cm);
  \filldraw[orange] (e) circle (.05cm);
  \filldraw[white] ($(-120:1)+(0:.17)$) circle (.05cm); 
  \filldraw[white] ($(60:1)+(-60:.17)$) circle (.05cm); 
  \filldraw[white] ($(60:1)+(-60:.27)$) circle (.05cm); 
  \filldraw[white] ($(0:1)+(120:.19)$) circle (.05cm); 
  \filldraw[white] ($(0:1)+(120:.29)$) circle (.05cm); 
  \filldraw[white] ($(0:1)+(120:.42)$) circle (.05cm); 
  \levinHex[red, thick]{0}{0}{1}{black}
  \levinHex[red, thick]{1.125*1.333}{.6495*1.333}{1}{black}
  \draw[red,thick,knot] (-60:.8) .. controls ++(30:.2) and ++(-90:.2) .. (0:.8); 
  \node at (0,0) {$\scriptstyle p$};
  \node at (c2) {$\scriptstyle q$};
  \node at (-30:1.2) {$\scriptstyle v$};
  \node at ($(c2)+(-30:1.2)$) {$\scriptstyle w$};
 }
 \displaybreak[1]\\\overset{\text{Cor.~\ref{cor:condensationFoamNE}}}{=}&
 \tikzmath{
  \coordinate (c2) at (1.125*1.333,.6495*1.333);
  \coordinate (e) at ($(c2)+(150:.7)+(60:.2)$);
  \coordinate (g2) at ($(e)+(60:-.4)$);
  \coordinate (m1) at ($(-30:1*.866)+(-45:-.2)$);
  \coordinate (loffs) at ($(-30:1*.866)+(-.333,.333)$);
  \coordinate (moffs) at (-.2,.1);
  \coordinate (mout) at ($(c2)+(loffs)+(moffs)$);
  \draw[DarkGreen,thick] ($(-120:1.5)+(-30:.2)$) .. controls ++(60:.4) and ++(-30:.2) .. (-120:.8) .. controls ++(150:.2) and ++(90:-.2) .. (180:.8) .. controls ++(90:.2) and ++(30:-.2) .. (120:.8) .. controls ++(30:.2) and ++(-30:-.2) .. (60:.8) .. controls ++(-30:.2) and ++(60:-.4) .. (g2);
  \draw[orange,thick] (g2) -- (e);
  \draw[red,thick] (g2) .. controls ++(-30:.4) and ++(0:.2) .. (60:.6) .. controls ++(0:-.2) and ++(30:.2) .. (120:.6) .. controls ++(30:-.2) and ++(90:.2) .. (180:.6) .. controls ++(90:-.2) and ++(-30:-.2) .. (-120:.6) .. controls ++(-30:.2) and ++(30:-.2) .. (-60:.8) .. controls ++(30:.2) and ++(-135:.3) .. (m1);
  \draw[red,thick] (loffs) .. controls ++(-45:-.2) and ++(0:-.2) .. ($(c2)+(-120:.8)$) .. controls ++(0:.2) and ++(30:-.2) .. ($(c2)+(-60:.8)$) .. controls ++(30:.2) and ++(45:-.2) .. (mout) .. controls ++(45:.1) .. ($(c2)+(loffs)$);
  \draw[red,thick] (mout) -- ++(-45:-.2);
  \draw[red,thick] ($(mout)+(-45:-.2)$) arc (-45:45:.2);
  \draw[red,thick] ($(mout)+(-45:-.2)$) arc (-45:-135:.2) .. controls ++(135:.4) and ++(20:.2) .. ($(c2)+(-140:.6)$) .. controls ++(20:-.2) and ++(135:-.3) .. ($(loffs)+(-45:-.4)$);
  \filldraw[black] (g2) circle (.05cm);
  \filldraw[orange] (e) circle (.05cm);
  \filldraw[white] ($(-120:1)+(0:.17)$) circle (.05cm); 
  \filldraw[white] ($(60:1)+(-60:.17)$) circle (.05cm); 
  \filldraw[white] ($(60:1)+(-60:.27)$) circle (.05cm); 
  \filldraw[white] ($(0:1)+(120:.185)$) circle (.05cm); 
  \filldraw[white] ($(0:1)+(120:.4)$) circle (.05cm); 
  \levinHex[red, thick]{0}{0}{1}{black}
  \levinHex[red, thick]{1.125*1.333}{.6495*1.333}{1}{black}
  \node at (0,0) {$\scriptstyle p$};
  \node at (c2) {$\scriptstyle q$};
  \node at (-30:1.2) {$\scriptstyle v$};
  \node at ($(c2)+(-30:1.2)$) {$\scriptstyle w$};
 }
 \displaybreak[1]\\&=
 \tikzmath{
  \coordinate (c2) at (1.125*1.333,.6495*1.333);
  \coordinate (e) at ($(c2)+(150:.7)$);
  \coordinate (g1) at (-120:.8);
  \coordinate (m1) at ($(-30:1*.866)+(-45:-.2)$);
  \coordinate (loffs) at ($(-30:1*.866)+(-.333,.333)$);
  \coordinate (moffs) at (-.2,.1);
  \coordinate (mout) at ($(c2)+(loffs)+(moffs)$);
  \draw[DarkGreen,thick] ($(-120:1.5)+(-30:.2)$) .. controls ++(60:.4) and ++(-30:.2) .. (g1);
  \draw[orange,thick] (g1) .. controls ++(150:.2) and ++(90:-.2) .. (180:.8) .. controls ++(90:.2) and ++(30:-.2) .. (120:.8) .. controls ++(30:.2) and ++(-30:-.2) .. (60:.8) .. controls ++(-30:.2) and ++(60:-.4) .. (e);
  \draw[red,thick] (g1) .. controls ++(60:.3) and ++(10:-.2) .. (-60:.8) .. controls ++(10:.2) and ++(-135:.3) .. (m1);
  \draw[red,thick] (loffs) .. controls ++(-45:-.2) and ++(0:-.2) .. ($(c2)+(-120:.8)$) .. controls ++(0:.2) and ++(30:-.2) .. ($(c2)+(-60:.8)$) .. controls ++(30:.2) and ++(45:-.2) .. (mout) .. controls ++(45:.1) .. ($(c2)+(loffs)$);
  \draw[red,thick] (mout) -- ++(-45:-.2);
  \draw[red,thick] ($(mout)+(-45:-.2)$) arc (-45:45:.2);
  \draw[red,thick] ($(mout)+(-45:-.2)$) arc (-45:-135:.2) .. controls ++(135:.4) and ++(20:.2) .. ($(c2)+(-140:.6)$) .. controls ++(20:-.2) and ++(135:-.3) .. ($(loffs)+(-45:-.4)$);
  \filldraw[black] (g1) circle (.05cm);
  \filldraw[orange] (e) circle (.05cm);
  \filldraw[white] ($(-120:1)+(0:.17)$) circle (.05cm); 
  \filldraw[white] ($(60:1)+(-60:.17)$) circle (.05cm); 
  \filldraw[white] ($(60:1)+(-60:.27)$) circle (.05cm); 
  \filldraw[white] ($(0:1)+(120:.185)$) circle (.05cm); 
  \filldraw[white] ($(0:1)+(120:.4)$) circle (.05cm); 
  \levinHex[red, thick]{0}{0}{1}{black}
  \levinHex[red, thick]{1.125*1.333}{.6495*1.333}{1}{black}
  \node at (0,0) {$\scriptstyle p$};
  \node at (c2) {$\scriptstyle q$};
  \node at (-30:1.2) {$\scriptstyle v$};
  \node at ($(c2)+(-30:1.2)$) {$\scriptstyle w$};
 }
 \displaybreak[1]\\&=D_{p,q}
 \tikzmath{
  \coordinate (c2) at (1.125*1.333,.6495*1.333);
  \coordinate (e) at ($(c2)+(150:.7)$);
  \coordinate (g1) at (-120:.8);
  \coordinate (m1) at ($(-30:1*.866)+(-45:-.2)$);
  \coordinate (loffs) at ($(-30:1*.866)+(-.333,.333)$);
  \coordinate (moffs) at (-.2,.1);
  \coordinate (mout) at ($(c2)+(loffs)+(moffs)$);
  \draw[DarkGreen,thick] ($(-120:1.5)+(-30:.2)$) .. controls ++(60:.4) and ++(-30:.2) .. (g1);
  \draw[orange,thick] (g1) .. controls ++(150:.2) and ++(90:-.2) .. (180:.8) .. controls ++(90:.2) and ++(30:-.2) .. (120:.8) .. controls ++(30:.2) and ++(-30:-.2) .. (60:.8) .. controls ++(-30:.2) and ++(60:-.4) .. (e);
  \draw[red,thick] (g1) .. controls ++(60:.3) and ++(10:-.2) .. (-60:.8) .. controls ++(10:.2) and ++(-135:.3) .. (m1);
  \filldraw[black] (g1) circle (.05cm);
  \filldraw[orange] (e) circle (.05cm);
  \filldraw[white] ($(-120:1)+(0:.17)$) circle (.05cm); 
  \filldraw[white] ($(60:1)+(-60:.17)$) circle (.05cm); 
  \filldraw[white] ($(60:1)+(-60:.27)$) circle (.05cm); 
  \levinHex[red, thick]{0}{0}{1}{black}
  \levinHex[red, thick]{1.125*1.333}{.6495*1.333}{1}{black}
  \node at (0,0) {$\scriptstyle p$};
  \node at (c2) {$\scriptstyle q$};
  \node at (-30:1.2) {$\scriptstyle v$};
  \node at ($(c2)+(-30:1.2)$) {$\scriptstyle w$};
 }
\end{align*}

Recall that $Z(\cX)(s\to At)\cong Z(\cX)_A(sA_A\to tA_A)$.
Consequently, one example of such a trivalent vertex is the projection onto a simple summand of $sA_A$, i.e., onto a single anyon in the condensed phase.
In other words, to obtain a string operator for an anyon $M_A\in\Irr(Z(\mathcal{X})_A)$, we pick a simple object $s\le M$, which determines an inclusion $M_A\to sA_A$, and apply a linear combination of string operators $\sigma^s$ which absorb the projection $\pi_{M_A}\in\Tube_A(\mathcal{X})$ onto a summand of type $M_A$ at the endpoint of the string.
Different choices of inclusion $M_A\to sA_A$ produce the same superselection sector because they are related by $A$-module endomorphisms of $sA_A$, which can also slide along the string operator.
For the same reason, the projections in $Z(\cX)(sA_A\to sA_A)$ onto each isotypic component also slide topologically along the string operator.
Thus, the operator $\sigma^s$ can be written as a direct sum of operators, one for each type of simple object in $Z(\cX)_A$ which appears as a summand of $sA_A$.

When we condense the algebra $A$, we expect that objects which are not transparent to $A$ become confined.
To see this, we will investigate when string operators $\sigma^s_r$ commute with $D_{p,q}$ terms.
Suppose $s\in\Irr(Z(\cX))$ is an anyon in the uncondensed phase, and let $|\Omega\rangle$ be a local ground state near $r$ and on the support of $D_{p,q}$.
If $r$ does not cross the path from $v$ to $w$ chosen when defining $D_{p,q}$, then $\sigma^s_p$ and $D_{p,q}$ obviously commute.
If the paths do cross, then we have, for example,
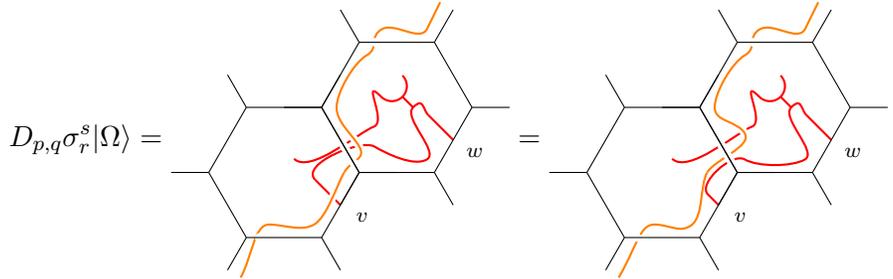
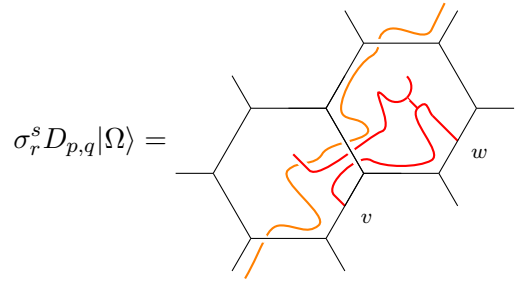
\begin{align*}
 D_{p,q}\sigma^s_r|\Omega\rangle &=
 \tikzmath{
  \coordinate (c2) at (1.125*1.333,.6495*1.333);
  \coordinate (c3) at ($(c2)+(0,1.732)$); 
  \coordinate (g1) at (-120:.8);
  \coordinate (m1) at ($(-30:1*.866)+(-45:-.2)$);
  \coordinate (loffs) at ($(-30:1*.866)+(-.333,.333)$);
  \coordinate (moffs) at (-.2,.1);
  \coordinate (mout) at ($(c2)+(loffs)+(moffs)$);
  \draw[red,thick] (loffs) .. controls ++(135:.3) and ++(-30:-.2) .. ($(c2)+(-120:.6)$) .. controls ++(-30:.2) and ++(30:-.2) .. ($(c2)+(-60:.8)$) .. controls ++(30:.2) and ++(-135:.2) .. (mout) .. controls ++(-135:-.2) and ++(-45:-.1) .. ($(c2)+(loffs)$);
  \draw[red,thick] (mout) -- ++(-45:-.2);
  \draw[red,thick] ($(mout)+(-45:-.2)$) arc (-45:45:.2);
  \draw[red,thick] ($(mout)+(-45:-.2)$) arc (-45:-135:.2) .. controls ++(135:.4) and ++(20:.2) .. ($(c2)+(-140:.4)$) .. controls ++(20:-.2) and ++(135:-.3) .. ($(loffs)+(-45:-.4)$);
  \filldraw[white] ($(0:1)+(120:.35)$) circle (.05cm); 
  \filldraw[white] ($(0:1)+(120:.41)$) circle (.05cm); 
  \draw[orange,thick] ($(-120:1.5)+(-30:.2)$) .. controls ++(60:.3) and ++(0:-.2) .. ($(-120:.8)+(.2,0)$);
  \draw[orange,thick,knot] (0:.8) .. controls ++(60:.2) and ++(120:-.2) .. ($(c2)+(-120:.8)+(120:.2)$) .. controls ++(120:.2) and ++(90:-.2) .. ($(c2)+(180:.8)$) .. controls ++(90:.2) and ++(30:-.2) .. ($(c2)+(120:.8)$) .. controls ++(30:.2) and ++(30:-.2) .. ($(c3)+(-120:.8)+(.2,0)$) .. controls ++(30:.2) and ++(60:-.2) .. ($(c3)+(-60:.8)$) -- ($(c2)+(60:1.5)+(150:.2)$); 
  \filldraw[white] ($(-120:1)+(0:.17)$) circle (.05cm); 
  \filldraw[white] ($(0:1)+(120:.15)$) circle (.05cm); 
  \filldraw[white] ($(c2)+(120:1)+(0:.21)$) circle (.05cm); 
  \levinHex[red,thick]{0}{0}{1}{black}
  \levinHex[red, thick]{1.125*1.333}{.6495*1.333}{1}{black}
  \draw[orange,thick,knot] ($(-120:.8)+(.2,0)$) .. controls ++(0:.2) and ++(30:-.2) .. (-60:.8) .. controls ++(30:.2) and ++(60:-.2) .. (0:.8); 
  \node at (-30:1.2) {$\scriptstyle v$};
  \node at ($(c2)+(-30:1.2)$) {$\scriptstyle w$};
 }
 =
 \tikzmath{
  \coordinate (c2) at (1.125*1.333,.6495*1.333);
  \coordinate (c3) at ($(c2)+(0,1.732)$); 
  \coordinate (g1) at (-120:.8);
  \coordinate (m1) at ($(-30:1*.866)+(-45:-.2)$);
  \coordinate (loffs) at ($(-30:1*.866)+(-.333,.333)$);
  \coordinate (moffs) at (-.2,.1);
  \coordinate (mout) at ($(c2)+(loffs)+(moffs)$);
  \draw[red,thick] (-30:.8660) -- ++(135:.2) .. controls ++(135:.2) and ++(-30:-.2) .. ($(c2)+(-120:.8)$) .. controls ++(-30:.2) and ++(30:-.2) .. ($(c2)+(-60:.8)$) .. controls ++(30:.2) and ++(-135:.2) .. (mout) .. controls ++(-135:-.2) and ++(-45:-.1) .. ($(c2)+(loffs)$);
  \draw[red,thick] (mout) -- ++(-45:-.2);
  \draw[red,thick] ($(mout)+(-45:-.2)$) arc (-45:45:.2);
  \draw[red,thick] ($(mout)+(-45:-.2)$) arc (-45:-135:.2) .. controls ++(135:.4) and ++(20:.2) .. ($(c2)+(-140:.4)$) .. controls ++(20:-.2) and ++(135:-.3) .. ($(loffs)+(-45:-.4)$);
  \filldraw[white] ($(0:1)+(120:.18)$) circle (.05cm); 
  \filldraw[white] ($(0:1)+(120:.41)$) circle (.05cm); 
  \draw[orange,thick] ($(-120:1.5)+(-30:.2)$) .. controls ++(60:.5) and ++(0:-.2) .. ($(-120:.8)+(.3,0)$) .. controls ++(0:.2) and ++(40:-.2) .. (-60:.8) .. controls ++(40:.2) and ++(60:-.2) .. (0:.6); 
  \draw[orange,thick,knot] (0:.6) .. controls ++(60:.2) and ++(120:-.2) .. ($(c2)+(-120:.6)+(120:.2)$) .. controls ++(120:.2) and ++(90:-.2) .. ($(c2)+(180:.8)$) .. controls ++(90:.2) and ++(30:-.2) .. ($(c2)+(120:.8)$) .. controls ++(30:.2) and ++(30:-.2) .. ($(c3)+(-120:.8)+(.2,0)$) .. controls ++(30:.2) and ++(60:-.2) .. ($(c3)+(-60:.8)$) -- ($(c2)+(60:1.5)+(150:.2)$); 
  \filldraw[white] ($(-120:1)+(0:.17)$) circle (.05cm); 
  \filldraw[white] ($(0:1)+(120:.26)$) circle (.05cm); 
  \filldraw[white] ($(c2)+(120:1)+(0:.21)$) circle (.05cm); 
  \levinHex{0}{0}{1}{black} 
  \levinHex[red, thick]{1.125*1.333}{.6495*1.333}{1}{black}
  \node at (-30:1.2) {$\scriptstyle v$};
  \node at ($(c2)+(-30:1.2)$) {$\scriptstyle w$};
 }
 \,.
\end{align*}
Meanwhile,
\[
 \sigma^s_rD_{p,q}|\Omega\rangle =
 \tikzmath{
  \coordinate (c2) at (1.125*1.333,.6495*1.333);
  \coordinate (c3) at ($(c2)+(0,1.732)$); 
  \coordinate (g1) at (-120:.8);
  \coordinate (m1) at ($(-30:1*.866)+(-45:-.2)$);
  \coordinate (loffs) at ($(-30:1*.866)+(-.333,.333)$);
  \coordinate (moffs) at (-.2,.1);
  \coordinate (mout) at ($(c2)+(loffs)+(moffs)$);
  \draw[orange,thick] ($(-120:1.5)+(-30:.2)$) .. controls ++(60:.5) and ++(0:-.2) .. ($(-120:.8)+(.3,0)$) .. controls ++(0:.2) and ++(40:-.2) .. (-60:.8) .. controls ++(40:.2) and ++(45:-.2) .. ($(loffs)+(45:-.2)+(-45:-.4)$) .. controls ++(45:.2) and ++(60:-.2) .. ($(0:.6)+(-60:-.3)$) .. controls ++(60:.2) and ++(120:-.2) .. ($(c2)+(-120:.6)+(120:.3)$) .. controls ++(120:.2) and ++(90:-.2) .. ($(c2)+(180:.8)$) .. controls ++(90:.2) and ++(30:-.2) .. ($(c2)+(120:.8)$) .. controls ++(30:.2) and ++(30:-.2) .. ($(c3)+(-120:.8)+(.2,0)$) .. controls ++(30:.2) and ++(60:-.2) .. ($(c3)+(-60:.8)$) -- ($(c2)+(60:1.5)+(150:.2)$); 
  \draw[red,thick] (-30:.8660) -- ++(135:.2) .. controls ++(135:.2) and ++(-30:-.2) .. ($(c2)+(-120:.8)$) .. controls ++(-30:.2) and ++(30:-.2) .. ($(c2)+(-60:.8)$) .. controls ++(30:.2) and ++(-135:.2) .. (mout) .. controls ++(-135:-.2) and ++(-45:-.1) .. ($(c2)+(loffs)$);
  \draw[red,thick] (mout) -- ++(-45:-.2);
  \draw[red,thick] ($(mout)+(-45:-.2)$) arc (-45:45:.2);
  \draw[red,thick,knot] ($(mout)+(-45:-.2)$) arc (-45:-135:.2) .. controls ++(135:.4) and ++(20:.2) .. ($(c2)+(-130:.4)$) .. controls ++(20:-.2) and ++(135:-.2) .. ($(loffs)+(-45:-.2)$) -- ++(-45:-.3);
  \filldraw[white] ($(0:1)+(120:.18)$) circle (.05cm); 
  \filldraw[white] ($(0:1)+(120:.34)$) circle (.05cm); 
  \filldraw[white] ($(-120:1)+(0:.17)$) circle (.05cm); 
  \filldraw[white] ($(0:1)+(120:.48)$) circle (.05cm); 
  \filldraw[white] ($(c2)+(120:1)+(0:.21)$) circle (.05cm); 
  \levinHex{0}{0}{1}{black} 
  \levinHex[red, thick]{1.125*1.333}{.6495*1.333}{1}{black}
  \node at (-30:1.2) {$\scriptstyle v$};
  \node at ($(c2)+(-30:1.2)$) {$\scriptstyle w$};
 }
\]
Thus, $D_{p,q}$ commutes with $\sigma^s(\phi,\psi)$ up to the double braiding between $A$ and $s$.
On the other hand, as we saw above, the string operator $\sigma^s$ now corresponds to the object $sA_A\in Z(\cX)_A$, so in general $\sigma^s$ splits as a direct sum of string operators associated to the summands of $sA_A$.
As we will see in Section~\ref{sssec:tubeARep} below, this means that $D_{p,q}$ commutes with string operators for precisely those excitations which are representations of $\Tube_A(\cX)$, i.e.~objects in $\Irr(Z(\cX)_A^{\loc})$.
Thus, summands of string operators $\sigma^s$ corresponding to nonlocal summands of $sA_A$ pay an energy cost proportional to the length of the string, the excitations at the end of these strings are confined.

\subsubsection{Tube algebra representations from excitations at \texorpdfstring{$t=1$}{t=1}}
\label{sssec:tubeARep}
We will now explain how states containing an isolated low-energy (i.e.~only the $A_\ell$ and $B_p$ terms of the Hamiltonian are violated) excitation at $t=1$ give representations of $\Tube_A(\mathcal{X})$.
As before, we begin by defining an abstract representation $\rho_{M_A}$ of $\PreTubeAC$ for each $M_A\in Z(\cX)_A$.
We will then define an action of $\PreTubeAC$ as local operators in our lattice model, such that a variant of Proposition~\ref{prop:tubeActionCorrect} holds.

Our definition of $\rho_{M_A}$ is a small modification of the previous definition \eqref{eq:tubeRep}.
First, if $H\in\Irr(Z(\cX))$, then we define a representation $\rho_{HA_A}$ of $\PreTubeAC$ on the Hilbert space $\bigoplus_{x\in\Irr(\cX)}\cX(H\to U(A)x)$ by
\begin{equation}
 \label{eq:preTubeARep}
 \rho_H(\phi)m=\delta_{x,y}\cdot\,
 \tikzmath{
 \draw (.3,-2.2) node[below]{$\scriptstyle H$} -- (.3,-.7) -- node[right]{$\scriptstyle x$} (.3,-.3);
 \draw (-.2,.3) -- (-.2,.7) node[above]{$\scriptstyle z$};
 \draw[red,thick] (.1,-1) .. controls ++(90:.4cm) and ++(90:-.4cm) .. (-.75,-.3) -- (-.75,.3) arc (180:0:.2);
 \draw[red,thick] (-.55,.475) -- (-.55,.7) node[above]{$\scriptstyle A$};
 \draw[knot] (.3,.3) node[left, yshift=.15cm]{$\scriptstyle c$} arc (180:0:.3cm) -- node[right]{$\scriptstyle \overline{c}$} (.9,-1.3) arc (0:-180:.6cm) -- node[left]{$\scriptstyle c$} (-.3,-.3);
 \roundNbox{fill=white}{(0,0)}{.3}{.3}{.3}{$\phi$}
 \roundNbox{fill=white}{(.3,-1.3)}{.3}{0}{0}{$m$}
 }
 \qquad\qquad\qquad
 \forall\,m\in \cX(H\to U(A)x).
\end{equation}

Just as we computed an equivalence between $\Rep(\Tube(\cX))$ and $Z(\cX)$, there is also an equivalence $\Rep(\PreTubeAC))\cong Z(\cX)_A$.
One direction is provided by $sA_A\to \rho_{sA_A}$.
As for the other, suppose $\rho$ is a representation of $\PreTubeAC$.
Since $\Tube(\cX)\subseteq\Rep(\PreTubeAC)$ as a subalgebra, we can still recover an object $M\in Z(\cX)$ from the $\Tube(\cX)$ action, obtaining the half-braiding by equation \eqref{eq:tubeRepToZC}.
We can define an action $AM\to M$ using the data $\rho(\Absorb(A))$ by the following equation.
\[
 \tikzmath{
  \draw[red,thick] (-.3,.3) -- (-.3,1) arc (180:90:.6cm);
  \draw (0,-.6) node[below]{$\scriptstyle x$} --  (0,0);
  \draw (.3,.3) -- node[right]{$\scriptstyle y$} (.3,.7) -- (.3,1.3) -- node[right]{$\scriptstyle M$} (.3,2) -- (.3,2.9) node[above]{$\scriptstyle x$};
  \roundNbox{fill=white}{(0,0)}{.3}{.3}{.3}{$a$}
  \roundNbox{fill=white}{(.3,1)}{.3}{0}{0}{$n^\dag$}
  \roundNbox{fill=white}{(.3,2.2)}{.3}{0}{0}{$m$}
 }
 :=\langle n|a\lhd m\rangle\id_x.
\]
This action morphism lives in $\cX$, but the compatibility between the action of the subalgebras $\Tube(\cX)$ and $\Absorb(A)$ ensures that the action $AM\to M$ is lives in $Z(\cX)$, i.e. $M_A\in Z(\cX)_A$.
In particular, the representation $\rho_{sA_A}$ decomposes as a direct sum just as $sA_A\in Z(\cX)_A$ does.
An arbitrary choice of inclusion $M_A\to sA_A$ therefore lets us define $\rho_{M_A}$.
Because of the equivalence of categories $\Rep(\PreTubeAC)\cong Z(\cX)_A$ we have obtained, this definition is independent off the choice of inclusion, up to unique isomorphism.

Because $\widetilde{\Tube_A(\mathcal{X})}=\Absorb(A)\Tube(\mathcal{X})$, we can define the action of $\widetilde{\Tube_A(\mathcal{X})}$ on each subalgebra and then check the relations defining $\Tube_A(\mathcal{X})$.
As before, suppose that $|\phi\rangle$ is a state with an isolated excitation at the plaquette $p$ and link $\ell$.
Much as before, we define the action of $f\in\Tube(\mathcal{X})$ by
\[
 \tikzmath{
 \coordinate (a) at (1.5,-.25);
 \coordinate (b) at (1,-.12);
 \coordinate (c) at (20:.9cm);
 \coordinate (z) at (-20:1.8cm);
 \filldraw[orange] (c) circle (.05cm);
 \draw[orange,thick] (z) .. controls ++(120:.2cm) and ++(-60:.2cm) .. (a) .. controls ++(120:.2cm) and ++(0:.2cm) .. (b) .. controls ++(180:.2cm) and ++(-80:.2cm)   .. (c);
 \fill[white] (0:.87) circle (.07cm);
 \levinHexOpen[red, thick]{0}{0}{.75}{black}{1}
 \levinHexOpen[red, thick]{1.125}{.6495}{.75}{black}{4}
 \node at (210:1.299) {$\scriptstyle r$};
 \node at (210:.8) {$\scriptstyle k$};
 \node at (0:0) {$\scriptstyle p$};
 \node at (1.125,.6495) {$\scriptstyle q$};
 }
 |\Omega\rangle
 \qquad
 \overset{\phi \rhd -}{\longmapsto}
 \qquad
 \tikzmath{
 \coordinate (a) at (1.5,-.25);
 \coordinate (b) at (1.05,-.12);
 \coordinate (cp) at (30:.9cm);
 \coordinate (x) at (13.5:.925cm);
 \coordinate (z) at (-20:1.8cm);
 \filldraw[orange] (cp) circle (.05cm);
 \draw[orange,thick] (z) .. controls ++(120:.2cm) and ++(-60:.2cm) .. (a) .. controls ++(120:.2cm) and ++(0:.2cm) .. (b) .. controls ++(180:.2cm) and ++(-110:.2cm)   .. (x) -- (cp);
 \fill[white] (0:.87) circle (.07cm);
 \arHexKnot[1]{0}{0}{.5}{blue}{3}{0}
 \arHexKnot[4]{1.125}{.6495}{.5}{blue}{3}{0}
 \draw[blue] (60:.5) -- (.625,.6495);
 \draw[blue] (0:.5) -- (0.875,.2165);
 \filldraw[black] (x) circle (.05cm);
 \levinHexOpen[red, thick]{0}{0}{.75}{black}{1}
 \levinHexOpen[red, thick]{1.125}{.6495}{.75}{black}{4}
 \node at (210:1.299) {$\scriptstyle r$};
 \node at (210:.8) {$\scriptstyle k$};
 \node at (0:0) {$\scriptstyle p$};
 \node at (1.125,.6495) {$\scriptstyle q$};
}
|\Omega\rangle\text{,}
\]
where $|\Omega\rangle$ is a ground state.
However, this definition can be significantly simplified.
The proofs and results of Section \ref{ssec:tubeImplementation}, and in particular Proposition \ref{prop:tubeActionCorrect}, apply equally well in our new lattice model, as long as we use the half-braiding on $A$ whenever we cross the vertical $A$-strand inside a plaquette.
Therefore, we can apply Proposition \ref{prop:tubeActionCorrect} to rewrite the action as
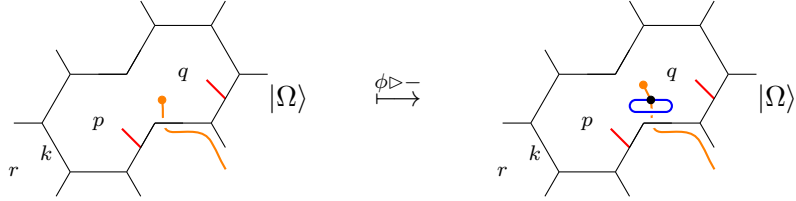
\[
 \tikzmath{
 \coordinate (a) at (1.5,-.25);
 \coordinate (b) at (1,-.12);
 \coordinate (c) at (20:.9cm);
 \coordinate (z) at (-20:1.8cm);
 \filldraw[orange] (c) circle (.05cm);
 \draw[orange,thick] (z) .. controls ++(120:.2cm) and ++(-60:.2cm) .. (a) .. controls ++(120:.2cm) and ++(0:.2cm) .. (b) .. controls ++(180:.2cm) and ++(-80:.2cm)   .. (c);
 \fill[white] (0:.87) circle (.07cm);
 \levinHexOpen[red, thick]{0}{0}{.75}{black}{1}
 \levinHexOpen[red, thick]{1.125}{.6495}{.75}{black}{4}
 \node at (210:1.299) {$\scriptstyle r$};
 \node at (210:.8) {$\scriptstyle k$};
 \node at (0:0) {$\scriptstyle p$};
 \node at (1.125,.6495) {$\scriptstyle q$};
 }
 |\Omega\rangle
 \qquad
 \overset{\phi \rhd -}{\longmapsto}
 \qquad
 \tikzmath{
 \coordinate (a) at (1.5,-.25);
 \coordinate (b) at (1,-.12);
 \coordinate (c) at (20:.9cm);
 \coordinate (d) at ($(c)+(-.1,.2)$);
 \coordinate (z) at (-20:1.8cm);
 \filldraw[orange] (d) circle (.05cm);
 \draw[orange,thick] (z) .. controls ++(120:.2cm) and ++(-60:.2cm) .. (a) .. controls ++(120:.2cm) and ++(0:.2cm) .. (b) .. controls ++(180:.2cm) and ++(-80:.2cm)  .. (c) -- (d);
 \filldraw[white] ($(c)+(0,-.16)$) circle (.04cm); 
 \draw[blue,thick] (c) -- ++(.2,0) arc (90:-90:.08cm) -- ++(-.4,0) arc (270:90:.08cm) -- (c);
 \filldraw[black] (c) circle (.05cm);
 \fill[white] (0:.87) circle (.07cm);
 \levinHexOpen[red, thick]{0}{0}{.75}{black}{1}
 \levinHexOpen[red, thick]{1.125}{.6495}{.75}{black}{4}
 \node at (210:1.299) {$\scriptstyle r$};
 \node at (210:.8) {$\scriptstyle k$};
 \node at (0:0) {$\scriptstyle p$};
 \node at (1.125,.6495) {$\scriptstyle q$};
}
|\Omega\rangle
\]
as in equation~\eqref{eq:tubeActionAtVertex}.

The action of $g\in\Absorb(A)$ is given by
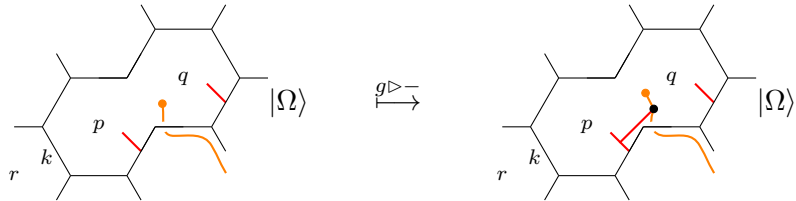
\[
 \tikzmath{
 \coordinate (a) at (1.5,-.25);
 \coordinate (b) at (1,-.12);
 \coordinate (c) at (20:.9cm);
 \coordinate (z) at (-20:1.8cm);
 \filldraw[orange] (c) circle (.05cm);
 \draw[orange,thick] (z) .. controls ++(120:.2cm) and ++(-60:.2cm) .. (a) .. controls ++(120:.2cm) and ++(0:.2cm) .. (b) .. controls ++(180:.2cm) and ++(-80:.2cm)   .. (c);
 \fill[white] (0:.87) circle (.07cm);
 \levinHexOpen[red, thick]{0}{0}{.75}{black}{1}
 \levinHexOpen[red, thick]{1.125}{.6495}{.75}{black}{4}
 \node at (210:1.299) {$\scriptstyle r$};
 \node at (210:.8) {$\scriptstyle k$};
 \node at (0:0) {$\scriptstyle p$};
 \node at (1.125,.6495) {$\scriptstyle q$};
 }
 |\Omega\rangle
 \qquad
 \overset{g\rhd-}{\longmapsto}
 \qquad
 \tikzmath{
 \coordinate (a) at (1.5,-.25);
 \coordinate (b) at (1.05,-.12);
 \coordinate (c) at (20:.9cm);
 \coordinate (x) at (15:.92cm);
 \coordinate (z) at (-20:1.8cm);
 \coordinate (w) at (.433,-.2145);
 \coordinate (cp) at (30:.9cm);
 \filldraw[orange] (cp) circle (.05cm);
 \draw[orange,thick] (z) .. controls ++(120:.2cm) and ++(-60:.2cm) .. (a) .. controls ++(120:.2cm) and ++(0:.2cm) .. (b) .. controls ++(180:.2cm) and ++(-110:.2cm)   .. (x) -- (cp);
 \fill[white] (0:.87) circle (.07cm);
 \draw[red,thick] (x) -- (w);
 \filldraw[black] (x) circle (.05cm);
 \levinHexOpen[red, thick]{0}{0}{.75}{black}{1}
 \levinHexOpen[red, thick]{1.125}{.6495}{.75}{black}{4}
 \node at (210:1.299) {$\scriptstyle r$};
 \node at (210:.8) {$\scriptstyle k$};
 \node at (0:0) {$\scriptstyle p$};
 \node at (1.125,.6495) {$\scriptstyle q$};
}
|\Omega\rangle
\]
Note that the action of $g\in\Absorb(A)$ commutes with all $A_\ell$ and $B_p$ terms away from $p\vee q$, as well as all $D_{p,q}$ terms,
including the case where $v\in p$ and $w\in q$.
Thus, $\Absorb(A)$ is also an algebra of local operators acting on the space of localized excitations which do not violate any $D_{p,q}$ terms, i.e. low energy excitations.

These actions assemble to an action of $\widetilde{\Tube_A(\cX)}=\Absorb(A)\Tube(\cX)$: if $\psi\in\Absorb(A)$ and $\phi\in\Tube(\cX)$, then we define $(\psi\phi)\rhd:=(\psi\rhd)\circ(\phi\rhd)$.
To check that this action is well-defined, we also need to check that $(\phi\psi)\rhd=(\phi\rhd)\circ(\psi\rhd)$.
Suppose $\phi\in\cX(cy\to zc)$ and $\psi\in\cX(x\to Ay)$.
Then we have
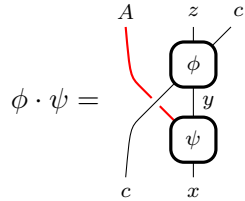
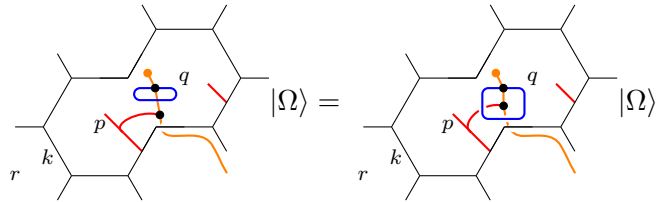
\[\phi\cdot\psi=\tikzmath{
 \draw (0,0) -- (0,2);
 \draw[thick,red] (0,.5) -- (-.7,1.2) .. controls ++(135:.2cm) .. (-.9,2);
 \filldraw[white] (-.5,1) circle (.07cm); 
 \draw (0,1.5) -- (-.7,.8) .. controls ++(225:.2cm) .. (-.9,0);
 \draw (0,1.5) -- (.5,2);
 \roundNbox{fill=white}{(0,.5)}{.3}{0}{0}{$\scriptstyle\psi$}
 \roundNbox{fill=white}{(0,1.5)}{.3}{0}{0}{$\scriptstyle\phi$}
 \node at (0,-.2) {$\scriptstyle x$};
 \node at (0,2.2) {$\scriptstyle z$};
 \node at (.2,1) {$\scriptstyle y$};
 \node at (-.9,-.2) {$\scriptstyle c$};
 \node at (-.9,2.2) {$\scriptstyle A$};
 \node at (.6,2.2) {$\scriptstyle c$};
}\]
Meanwhile,
\[
 \tikzmath{
  \coordinate (a) at (1.5,-.25);
  \coordinate (b) at (1,-.12);
  \coordinate (x) at (.8,.16);
  \coordinate (c) at (35:.9cm);
  \coordinate (d) at ($(c)+(-.1,.2)$);
  \coordinate (z) at (-20:1.8cm);
  \filldraw[orange] (d) circle (.05cm);
  \draw[orange,thick] (z) .. controls ++(120:.2cm) and ++(-60:.2cm) .. (a) .. controls ++(120:.2cm) and ++(0:.2cm) .. (b) .. controls ++(180:.2cm) and ++(-80:.2cm) .. (x) -- (c) -- (d);
  \draw[red,thick] (x) .. controls ++(165:.3cm) and ++(-30:-.2cm) .. ($(-30:.866*.75)+(-.25,.25)$);
  \draw[red,thick] ($(-30:.866*.75)+(-.25,.25)$) -- ++(-45:-.35cm); 
  \draw[blue,thick] (c) -- ++(.2,0) arc (90:-90:.08cm) -- ++(-.4,0) arc (270:90:.08cm) -- (c);
  \filldraw[black] (c) circle (.05cm);
  \filldraw[black] (x) circle (.05cm);
  \fill[white] (0:.87) circle (.07cm); 
  \levinHexOpen[red, thick]{0}{0}{.75}{black}{1}
  \levinHexOpen[red, thick]{1.125}{.6495}{.75}{black}{4}
  \node at (210:1.299) {$\scriptstyle r$};
  \node at (210:.8) {$\scriptstyle k$};
  \node at (0:0) {$\scriptstyle p$};
  \node at (1.125,.6495) {$\scriptstyle q$};
 }
 |\Omega\rangle
 =
 \tikzmath{
  \coordinate (a) at (1.5,-.25);
  \coordinate (b) at (1,-.12);
  \coordinate (x) at (.74,.28);
  \coordinate (c) at (35:.9cm);
  \coordinate (d) at ($(c)+(-.1,.2)$);
  \coordinate (z) at (-20:1.8cm);
  \filldraw[orange] (d) circle (.05cm);
  \draw[orange,thick] (z) .. controls ++(120:.2cm) and ++(-60:.2cm) .. (a) .. controls ++(120:.2cm) and ++(0:.2cm) .. (b) .. controls ++(180:.2cm) and ++(-80:.2cm) .. (x) -- (c) -- (d);
  \draw[red,thick] (x) .. controls ++(165:.3cm) and ++(-30:-.2cm) .. ($(-30:.866*.75)+(-.25,.25)$);
  \draw[red,thick] ($(-30:.866*.75)+(-.25,.25)$) -- ++(-45:-.35cm); 
  \filldraw[white] ($(x)+(.04,-.17)$) circle (.04cm); 
  \filldraw[white] ($(x)+(-.28,-.04)$) circle (.05cm); 
  \draw[blue,thick] (c) -- ++(.2,0) arc (90:0:.08cm) -- ++(0,-.25) arc (0:-90:.08cm) -- ++(-.4,0) arc (270:180:.08cm) -- ++(0,.25) arc (180:90:.08cm) -- (c);
  \filldraw[black] (c) circle (.05cm);
  \filldraw[black] (x) circle (.05cm);
  \fill[white] (0:.87) circle (.07cm); 
  \levinHexOpen[red, thick]{0}{0}{.75}{black}{1}
  \levinHexOpen[red, thick]{1.125}{.6495}{.75}{black}{4}
  \node at (210:1.299) {$\scriptstyle r$};
  \node at (210:.8) {$\scriptstyle k$};
  \node at (0:0) {$\scriptstyle p$};
  \node at (1.125,.6495) {$\scriptstyle q$};
 }
 |\Omega\rangle
\]
because the morphisms being applied locally at the vertex are equal,
showing that $(\phi\lhd)\circ(\psi\lhd)=(\phi\psi)\lhd$, as desired.

Now that we have an action of $\PreTubeAC$, we will check that this action satisfies \eqref{eq:tubeARel} on states which correspond to local $A$-modules.
Suppose $\phi\in\cX(cx\to Ayd)$ and $\psi\in\cX(d\to Ac)$.
Condition \eqref{eq:tubeARel} then becomes
\begin{equation}
 \label{eq:tubeARelActs}
 \tikzmath{
  \coordinate (a) at (1.75,-.25);
  \coordinate (b) at (1.25,-.12);
  \coordinate (c) at ($(20:.9cm)+(.25,.1)$);
  \coordinate (d) at ($(c)+(0,.4)$);
  \coordinate (z) at (-20:2cm);
  \coordinate (g) at ($(c)+(-.3,0)$);
  \filldraw[orange] (d) circle (.05cm);
  \draw[orange,thick] (z) .. controls ++(120:.2cm) and ++(-60:.2cm) .. (a) .. controls ++(120:.2cm) and ++(0:.2cm) .. (b) .. controls ++(180:.2cm) and ++(-80:.2cm)  .. (c) -- (d);
  \filldraw[white] ($(c)+(0,-.3)$) circle (.04cm); 
  \draw[DarkGreen,thick] (c) -- ++(.3,0) arc (90:-90:.15cm) -- ++(-.6,0) arc (270:90:.15cm); 
  \draw[blue,thick] (g) -- (c);
  \draw[red,thick] (g) arc (180:0:.15cm);
  \draw[red,thick] ($(g)+(.15,.15)$) .. controls ++(120:.5cm) and ++(135:.2cm) .. ($(-30:.866)+(-.25,.25)$) -- ++(-45:-.35cm);
  \filldraw (c) circle (.05cm);
  \filldraw (g) circle (.05cm);
  \fill[white] (0:1.12) circle (.05cm); 
  \levinHexOpen[red, thick]{0}{0}{1}{black}{1}
  \levinHexOpen[red, thick]{1.125*1.333}{.6495*1.333}{1}{black}{4}
  \node at (210:1.299*1.333) {$\scriptstyle r$};
  \node at (210:.8*1.333) {$\scriptstyle k$};
  \node at (0:0) {$\scriptstyle p$};
  \node at (1.125*1.333,.6495*1.333) {$\scriptstyle q$};
 }
 =
 \tikzmath{
  \coordinate (a) at (1.75,-.25);
  \coordinate (b) at (1.25,-.12);
  \coordinate (c) at ($(20:.9cm)+(.25,.1)$);
  \coordinate (d) at ($(c)+(0,.4)$);
  \coordinate (z) at (-20:2cm);
  \coordinate (g) at ($(c)+(.3,0)$);
  \coordinate (m) at ($(c)+(-.15,.15)$);
  \draw[red,thick] (g) .. controls ++(120:.2cm) and ++(45:.2cm) .. (m) .. controls ++(45:-.1cm) and ++(150:.2cm) .. (c);
  \draw[red,thick] (m) .. controls ++(120:.4cm) and ++(135:.2cm) .. ($(-30:.866)+(-.25,.25)$) -- ++(-45:-.35cm);
  \filldraw[white] ($(c)+(0,.21)$) circle (.05cm); 
  \filldraw[orange] (d) circle (.05cm);
  \draw[orange,thick] (z) .. controls ++(120:.2cm) and ++(-60:.2cm) .. (a) .. controls ++(120:.2cm) and ++(0:.2cm) .. (b) .. controls ++(180:.2cm) and ++(-80:.2cm)  .. (c) -- (d);
  \filldraw[white] ($(c)+(0,-.3)$) circle (.04cm); 
  \draw[blue,thick] (g) arc (90:-90:.15cm) -- ++(-.6,0) arc (270:90:.15cm) -- (c);
  \draw[DarkGreen,thick] (g) -- (c);
  \filldraw (c) circle (.05cm);
  \filldraw (g) circle (.05cm);
  \fill[white] (0:1.12) circle (.05cm); 
  \levinHexOpen[red, thick]{0}{0}{1}{black}{1}
  \levinHexOpen[red, thick]{1.125*1.333}{.6495*1.333}{1}{black}{4}
  \node at (210:1.299*1.333) {$\scriptstyle r$};
  \node at (210:.8*1.333) {$\scriptstyle k$};
  \node at (0:0) {$\scriptstyle p$};
  \node at (1.125*1.333,.6495*1.333) {$\scriptstyle q$};
 }
\end{equation}
where the black dots are $\phi$ and $\psi$.
We can manipulate the morphism appearing on the right-hand side as follows.
\[
 \tikzmath{
  \draw (.3,-2) node[below]{$\scriptstyle s$} -- (.3,-.7) -- node[right]{$\scriptstyle x$} (.3,-.3);
  \draw[thick,red] (-.6,0) -- (-.6,1.9) node[above]{$\scriptstyle A$};
  \draw[thick,red] (.3,1) .. controls ++(90:.5cm) and ++(110:-.2cm) .. (-.6,1.5);
  \draw[knot] (-.3,.3) -- (-.3,1.9) node[above]{$\scriptstyle y$};
  \draw[knot] (.3,.3) -- node[right]{$\scriptstyle d$} (.3,.7) -- (.3,1.3) node[left, yshift=.15cm]{$\scriptstyle c$} arc (180:0:.3cm) -- node[right]{$\scriptstyle \overline{c}$} (.9,-1) arc (0:-180:.6cm) -- node[left]{$\scriptstyle c$} (-.3,-.3);
  \roundNbox{fill=white}{(0,0)}{.3}{.5}{.3}{$\phi$}
  \roundNbox{fill=white}{(.3,-1)}{.3}{0}{0}{$m$}
  \roundNbox{fill=white}{(.3,1)}{.3}{0}{0}{$\psi$}
 }
 =
 \tikzmath{
  \draw (.3,-2) node[below]{$\scriptstyle s$} -- (.3,-.7) -- node[right]{$\scriptstyle x$} (.3,-.3);
  \draw[thick,red] (.3,1) .. controls ++(90:.5cm) and ++(90:-.2cm) .. (-.8,1.5) .. controls ++(90:.2cm) and ++(90:.2cm) .. (-.6,1.7) ;
  \filldraw[white] (-.6,1.41) circle (.05cm);
  \draw[thick,red] (-.6,0) -- (-.6,1.9) node[above]{$\scriptstyle A$};
  \draw[knot] (-.3,.3) -- (-.3,1.9) node[above]{$\scriptstyle y$};
  \draw[knot] (.3,.3) -- node[right]{$\scriptstyle d$} (.3,.7) -- (.3,1.3) node[left, yshift=.15cm]{$\scriptstyle c$} arc (180:0:.3cm) -- node[right]{$\scriptstyle \overline{c}$} (.9,-1) arc (0:-180:.6cm) -- node[left]{$\scriptstyle c$} (-.3,-.3);
  \roundNbox{fill=white}{(0,0)}{.3}{.5}{.3}{$\phi$}
  \roundNbox{fill=white}{(.3,-1)}{.3}{0}{0}{$m$}
  \roundNbox{fill=white}{(.3,1)}{.3}{0}{0}{$\psi$}
 }
 =
 \tikzmath{
  \draw (.3,-2.8) node[below]{$\scriptstyle s$} -- (.3,-.7) -- node[right]{$\scriptstyle x$} (.3,-.3);
  \draw[thick,red,knot] (-.6,-.8) arc (0:180:.15cm) -- (-.9,-1.3) arc (-180:0:1.15cm) -- (1.4,1) .. controls ++(90:.2cm) and ++(90:-.2cm) .. (-.8,1.5) .. controls ++(90:.2cm) and ++(90:.2cm) .. (-.6,1.7); 
  \filldraw[white] (-.6,1.41) circle (.05cm);
  \draw[thick,red] (-.6,0) -- (-.6,1.9) node[above]{$\scriptstyle A$};
  hdraw[knot] (-.3,.3) -- (-.3,1.9) node[above]{$\scriptstyle y$};
  \draw[knot] (.3,.3) node[left,yshift=.15cm]{$\scriptstyle d$} arc (180:0:.3cm) -- node[right]{$\scriptstyle\overline{d}$} (.9,-1.5) arc (0:-180:.7cm) -- (-.5,-.7) -- node[right]{$\scriptstyle c$} (-.5,-.3);
  \roundNbox{fill=white}{(0,0)}{.3}{.5}{.3}{$\phi$}
  \roundNbox{fill=white}{(.3,-1)}{.3}{0}{0}{$m$}
  \roundNbox{fill=white}{(-.5,-1)}{.3}{0}{0}{$\psi$}
 }
 =
 \tikzmath{
  \draw[thick,red] (.6,-2.7) .. controls ++(90:.4cm) and ++(-45:.2cm) .. (0,-2.4) .. controls ++(-45:.2cm) and ++(90:-.2cm) .. (-1.3,-1.5) -- (-1.3,-.4) .. controls ++(90:.2cm) and ++(90:-.2cm) .. (-.9,-.1) -- (-.9,1.2) arc (180:90:.3cm); 
  \draw[knot] (.3,-3.2) node[below]{$\scriptstyle s$} -- (.3,-.7) -- node[right]{$\scriptstyle x$} (.3,-.3);
  \draw[thick,red,knot] (-.7,-.7) arc (0:180:.15cm) -- (-1,-3.1) arc (-180:0:.3cm) .. controls ++(90:.4cm) and ++(90:-.4cm) .. (.6,-2.7);
  \draw[thick,red] (-.6,0) -- (-.6,1.9) node[above]{$\scriptstyle A$};
  \draw[knot] (-.3,.3) -- (-.3,1.9) node[above]{$\scriptstyle y$};
  \draw[knot] (.3,.3) node[left,yshift=.15cm]{$\scriptstyle d$} arc (180:0:.3cm) -- node[right]{$\scriptstyle\overline{d}$} (.9,-1.5) arc (0:-180:.7cm) -- (-.5,-.7) -- node[right]{$\scriptstyle c$} (-.5,-.3);
  \roundNbox{fill=white}{(0,0)}{.3}{.5}{.3}{$\phi$}
  \roundNbox{fill=white}{(.3,-1)}{.3}{0}{0}{$m$}
  \roundNbox{fill=white}{(-.5,-1)}{.3}{0}{0}{$\psi$}
 }
\]
In case the excitation lives in a local summand of $sA_A$,
the double braid is equal to the identity, and we have
\[
 \tikzmath{
  \draw (.3,-2) node[below]{$\scriptstyle s$} -- (.3,-.7) -- node[right]{$\scriptstyle x$} (.3,-.3);
  \draw[thick,red] (-.6,0) -- (-.6,1.9) node[above]{$\scriptstyle A$};
  \draw[thick,red] (.3,1) .. controls ++(90:.5cm) and ++(110:-.2cm) .. (-.6,1.5);
  \draw[knot] (-.3,.3) -- (-.3,1.9) node[above]{$\scriptstyle y$};
  \draw[knot] (.3,.3) -- node[right]{$\scriptstyle d$} (.3,.7) -- (.3,1.3) node[left, yshift=.15cm]{$\scriptstyle c$} arc (180:0:.3cm) -- node[right]{$\scriptstyle \overline{c}$} (.9,-1) arc (0:-180:.6cm) -- node[left]{$\scriptstyle c$} (-.3,-.3);
  \roundNbox{fill=white}{(0,0)}{.3}{.5}{.3}{$\phi$}
  \roundNbox{fill=white}{(.3,-1)}{.3}{0}{0}{$m$}
  \roundNbox{fill=white}{(.3,1)}{.3}{0}{0}{$\psi$}
 }
 =
 \tikzmath{
  \draw (.3,-2.5) node[below]{$\scriptstyle s$} -- (.3,-.7) -- node[right]{$\scriptstyle x$} (.3,-.3);
  \draw[thick,red] (-1.5,.6) .. controls ++(0:0) and ++(-45:-.6cm) .. (-1.1,.4) .. controls ++(-45:.2cm) and ++(90:.2cm) .. (-.9,0) -- (-.9,-.3) .. controls ++(90:-.2cm) and ++(90:.2cm) .. (-.7,-.7);
  \draw[thick,red,knot] (-.6,1.5) arc (90:180:.3cm) .. controls ++(90:-.2cm) and ++(45:.2cm) .. (-1.1,.8) .. controls ++(45:-.6cm) and ++(0:0) .. (-1.5,.6);
  \draw[thick,red] (-.6,0) -- (-.6,1.9) node[above]{$\scriptstyle A$};
  \draw[knot] (-.3,.3) -- (-.3,1.9) node[above]{$\scriptstyle y$};
  \draw[knot] (.3,.3) node[left,yshift=.15cm]{$\scriptstyle d$} arc (180:0:.3cm) -- node[right]{$\scriptstyle\overline{d}$} (.9,-1.5) arc (0:-180:.7cm) -- (-.5,-.7) -- node[right]{$\scriptstyle c$} (-.5,-.3);
  \roundNbox{fill=white}{(0,0)}{.3}{.5}{.3}{$\phi$}
  \roundNbox{fill=white}{(.3,-1)}{.3}{0}{0}{$m$}
  \roundNbox{fill=white}{(-.5,-1)}{.3}{0}{0}{$\psi$}
 }
 =
 \tikzmath{
  \draw (.3,-2.5) node[below]{$\scriptstyle s$} -- (.3,-.7) -- node[right]{$\scriptstyle x$} (.3,-.3);
  \draw[thick,red,knot] (-.6,1.5) arc (90:180:.3cm) -- (-.9,-.3) .. controls ++(90:-.2cm) and ++(90:.2cm) .. (-.7,-.7);
  \draw[thick,red] (-.6,0) -- (-.6,1.9) node[above]{$\scriptstyle A$};
  \draw[knot] (-.3,.3) -- (-.3,1.9) node[above]{$\scriptstyle y$};
  \draw[knot] (.3,.3) node[left,yshift=.15cm]{$\scriptstyle d$} arc (180:0:.3cm) -- node[right]{$\scriptstyle\overline{d}$} (.9,-1.5) arc (0:-180:.7cm) -- (-.5,-.7) -- node[right]{$\scriptstyle c$} (-.5,-.3);
  \roundNbox{fill=white}{(0,0)}{.3}{.5}{.3}{$\phi$}
  \roundNbox{fill=white}{(.3,-1)}{.3}{0}{0}{$m$}
  \roundNbox{fill=white}{(-.5,-1)}{.3}{0}{0}{$\psi$}
 }
\]
The last equality leaves us with the morphism on the left-hand side of \eqref{eq:tubeARelActs}, and follows from the fact that the twist of a condensable algebra is always equal to the identity \cite[Prop.~2.25]{MR2187404}.
Thus, under the equivalence $\Rep(\PreTubeAC)\cong Z(\cX)_A$ we defined in and after equation \eqref{eq:preTubeARep}, imposing relation \eqref{eq:tubeARel} selects exactly those objects in $Z(\mathcal{X})_A^{\loc}\cong Z(\mathcal{X}_A)$.

Taken together, the results we have accumulated in this section are sufficient to prove the following result.
\begin{thm}
 \label{thm:correctMTC}
 When $t=1$, the low-energy excitations of the lattice model of \S\ref{ssec:condensationModel} are classified by $Z(\mathcal{X})_A^{\loc}$.
\end{thm}

\subsection{Ground State Degeneracy}
\label{ssec:GSD}
Apart from characterizing anyons at times $t=0$ and $t=1$, one might also wish to understand the phase transition in terms of the effect on the space of ground states of the $A$ and $B$ terms.
We restrict our attention to the case where $\mathcal{X}$ admits monoidal fiber functor to $\Hilb$, so that a condensable algebra $A$ in $Z(\mathcal{X})$ becomes an actual $\Cstar$-algebra.
We will also not give a full analysis of the phase transition, but instead sketch a relationship to a more well-understood phase transition.

Since the space of ground states of a Levin-Wen model is locally (on an open disk) $1$-dimensional \cite{PhysRevB.85.075107,1106.6033}, the ground-states of the $A$ and $B$ terms of our Hamiltonian are parameterized by the labels of vertical edges and the flux between nearby pairs of vertical edges.
Therefore, at low energy, our model of anyon condensation is equivalent to another simpler lattice model, which generalizes the 2D-transverse field Ising model, which we call the \textbf{transverse field nearest-neighbor} model associated to the algebra $A$.

The first step in defining our simplified lattice model is describing the lattice.
Vertices of our new lattice correspond to vertical edges in the model of anyon condensation, with an edge between any pair two vertices which support a common $D$-term.
Our models for anyon condensation add one vertical edge extending above each plaquette, and $D_{p,q}$ terms for pairs of vertical edges above neighboring plaquettes, so we obtain the dual lattice of the square or hexagonal lattices we began with.
Therefore, from the honeycomb lattice described in Section \ref{ssec:condensationModel}, we obtain a regular triangular lattice, and from the square lattice, we obtain another square lattice.

To each vertex of the new model, we associate the same Hilbert space used in the old model, which is isomorphic to $A$ as an $A$-representation.
To each edge, we associate the term $\widetilde{D_{p,q}}=m^\dag m$, where $m$ is the multiplication of $A$, viewed as an operator on $A$.
To each vertex, we associate the term $C_v=uu^\dag$, where $u:1\to A$ is the unit map of $A$, viewed as an operator $\mathbb{C}\to A$.
Finally, on a closed lattice, we add a single nonlocal term $F$, projecting onto the states where $\prod_v\phi_v=1_A$, where $\phi_v$ is the label of the vertex $v$.
One should think of this term as picking out the ground states of the condensed $Z(\cX)_A^{\loc}$ theory as a superselection sector.

Notice that, in the case $\mathcal{X}=\Hilb[\Z/2]$ and $A=\C[\Z/2]$, i.e.~condensing $1\oplus e$ in toric code, we recover the Ising model from the $C$ and $D$ terms, because the $C_v$ term is just $Z_v$, while $\widetilde{D_{p,q}}$ simplifies to $X_vX_w$.

The general case is morally similar, but a full analysis would amount to deriving the TQFT structure from the lattice model, as in \cite{1106.6033}.
Such an argument is beyond the scope of this work.

\section{Examples}
\label{sec:examples}

In this section, we give worked examples including $\bbZ/2$-Toric Code, $\bbZ/n$-Toric Code, doubled semion, and doubled $\Fib$.
These examples can be substantially simplified as they are \emph{multiplicity free}, i.e., $\dim(\cX(xy\to z)) \in \{0,1\}$ for all $x,y,z\in \Irr(\cX)$.
In this setting, on a trivalent lattice, we can push the degrees of freedom from the vertices onto the links, as in the original Levin-Wen string-net model \cite{PhysRevB.71.045110}.
Above, we made frequent use of the fact that the morphism labeling a vertex determines an object labeling each adjacent link; the multiplicity free case is simply the situation where the labels on adjacent links determine (up to scalar) the morphism labeling the vertex.

To each ordinary link of our lattice, we assign the Hilbert space
\[
\mathcal{H}_\ell=\mathbb{C}^{\Irr(\mathcal{X})}
=
\bigoplus_{c\in \Irr(\cX)} \bbC |c\rangle.
\]
The Hamiltonian includes vertex terms and plaquette terms.
If all edges at a vertex $v$ are oriented away from $v$, the vertex term $A_v$ penalizes states where the labels on trivalent vertices are not admissible (the clockwise tensor product of the labels does not contain $1_\cX$).
If some edges at $v$ are oriented towards $v$, we take the dual of the simple object labeling those edges when defining $A_v$.

The plaquette term $B_p$ for a plaquette $p$ is
\[
 \frac{1}{D}\sum_{x\in\Irr(\cX)}d_xB_p^x\text{,}
\]
where $B_p^x$ has the effect of inserting a counterclockwise loop labeled by $x$ inside the plaquette.
We use the conventions of \cite{MR2726654}.
From the associator of $\mathcal{X}$, we may calculate an $F$-symbol, defined by
\[
\tikzmath{
\draw[->, semithick, >=stealth] (-.3,0) -- (0,.6);
\draw[->, semithick, >=stealth] (0.9,0)-- (.6,.6);
\draw[->, semithick, >=stealth] (-.3,1.2) -- (0,.6);
\draw[->, semithick, >=stealth] (0.9,1.2) -- (.6,.6);
\draw[->, semithick, >=stealth] (0,.6) node[below right] {$\scriptstyle h$}-- (.6,.6);
\node at (0,1) {$\scriptstyle f$};
\node at (0,.2) {$\scriptstyle c$};
\node at (.6,.2) {$\scriptstyle d$};
\node at (.6,1) {$\scriptstyle g$};
}
=\sum_{k\in\Irr(\mathcal{X})}F_{c,d}^{f,g}[h,k]
\tikzmath{
\draw[->, semithick, >=stealth] (0,0) -- (.6,.3);
\draw[->, semithick, >=stealth] (.6,.3) node[above right] {$\scriptstyle k$} -- (.6,.9);
\draw[->, semithick, >=stealth] (1.2,0) -- (.6,.3);
\draw[->, semithick, >=stealth] (0,1.2) -- (.6,.9);
\draw[->, semithick, >=stealth] (1.2,1.2) -- (.6,.9);
\node at (.2,.3) {$\scriptstyle c$};
\node at (1,.3) {$\scriptstyle d$};
\node at (.1,.9) {$\scriptstyle f$};
\node at (1,.9) {$\scriptstyle g$};
}\,.
\]
The operator $B_p^x$ is then given as follows on a hexagonal lattice, where the orientations are always from \emph{left to right}.
\begin{equation}
 \label{eq:Bpc}
 B_p^x\Bigg\lvert
 \tikzmath{
 \levinHex[]{0}{0}{.5}{black}
 \foreach \i in {1,...,6} {
 \node at (180+60*\i:.9cm) {$\scriptstyle g_{\i}$};
 }
 \foreach \i in {1,...,6} {
 \node at ($ (60*\i+150:.65cm) $) {$\scriptstyle d_{\i}$};
 }}
 \Bigg\rangle
 =
 \sum_{f_1,\dots,f_6}T(x,\vec{d},\vec{f},\vec{g})
 \Bigg\lvert
 \tikzmath{
 \levinHex[]{0}{0}{.5}{black}
 \foreach \i in {1,...,6} {
 \node at (180+60*\i:.9cm) {$\scriptstyle g_{\i}$};
 }
 \foreach \i in {1,...,6} {
 \node at ($ (60*\i+150:.65cm) $) {$\scriptstyle f_{\i}$};
 }}
 \Bigg\rangle
\end{equation}
where
\[
T(x,\vec{d},\vec{f},\vec{g})=F_{\overline{x},\overline{d_6}}^{f_1,\overline{g_6}}[\overline{f_6},\overline{f_1}]F_{x,d_1}^{f_2,\overline{g_1}}[f_1,\overline{d_2}]F_{x,d_2}^{f_3,g_2}[f_2,\overline{d_3}]F_{x,d_3}^{\overline{f_4},g_3}[f_3,d_4]F_{d_5,\overline{x}}^{f_4,g_4}[f_5,d_4]F_{d_6,\overline{x}}^{f_5,\overline{g_5}}[f_6,d_5].
\]
This complicated definition of $T$ comes from the need to reconcile the orientation of the loop of type $x$ inserted and the chosen orientation of our hexagonal lattice which was used to define the vertex terms.

\begin{rem}
 \label{rem:groupCase}
 When $\cX=\fdHilb(G,\omega)$ for a group $G$, for any choice of  $a_1,\ldots,a_n,b_1,\ldots,b_n\in G$,
 $\dim(\cX(a_1\cdots a_m \to b_1 \cdots b_n)) \in \{0,1\}$.
 Conversely, this condition on simple objects if $\cX$ is equivalent to $\cX$ being unitarily equivalent to some $\fdHilb(G,\omega)$.
 Indeed, under this condition, $\Irr(\cX)$ forms a group under tensor product, and the associator determines a $3$-cocycle $\omega$. 
 In this situation, we may put degrees of freedom on the links of any lattice (not just a trivalent lattice), in particular the square lattice as in Kitaev's quantum double model \cite{MR1951039}.

In this group case, we can adapt our general model to a rectangular lattice by adding one vertical edge for our condensable algebra $A$ emanating from each vertex out of the 2D plane, as follows.
$$
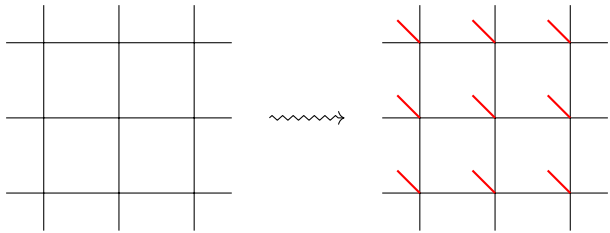
\tikzmath{
\draw[step=1.0,black,thin] (0.5,0.5) grid (3.5,3.5);
\draw [->, line join=round,
decorate, decoration={
    zigzag,
    segment length=4,
    amplitude=.9,post=lineto,
    post length=2pt
}]  (4,2) -- (5,2);
\draw[step=1.0,black,thin] (5.5,.5) grid (8.5,3.5);
\foreach \i in {6,7,8} {
\foreach \j in {1,2,3} {
        \draw[thick, red] (\i,\j) -- ($ (\i,\j) + (-.3,.3)$);
}}
}
$$
We can use the same Hilbert space $\mathbb{C}^{\Irr(\cX)}$ as before on black links, and use the Hilbert space $\bigoplus_{x\in\Irr(\cX)}(x,U(A))$ on red links.
This model will be applicable for the $\bbZ/2$-Toric Code, $\bbZ/n$-Toric Code, and doubled semion models.
\end{rem}

\subsection{\texorpdfstring{$\bbZ/2$}{Z/2} Toric code}

We begin by describing an anyon condensation in the simplest example of a Levin-Wen model, the $\bbZ/2$ toric code \cite{MR1951039}.
We denote by $X,Z$ the Pauli matrices
$$
X=
\begin{pmatrix}
0 & 1 \\
1 & 0
\end{pmatrix}
\qquad
\qquad
Z=
\begin{pmatrix}
1 & 0 \\
0 & -1
\end{pmatrix}.
$$
The system is defined on a square grid on a plane with an edge rising vertically from each vertex.
$$
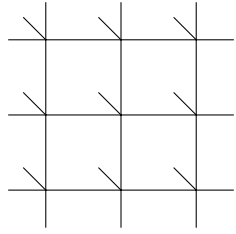
\tikzmath{
\draw[step=1.0,black,thin] (0.5,0.5) grid (3.5,3.5);
\foreach \i in {1,2,3} {
\foreach \j in {1,2,3} {
        \draw (\i,\j) -- ($ (\i,\j) + (-.3,.3)$);
}}
}
$$
To each link $\ell$ of the lattice,
we associate the Hilbert space $\cH_\ell = \bbC^2$.
The spaces $\cH_\ell$ on the 2D lattice are a direct sum $\bbC|0_\ell\rangle \oplus \bbC |1_\ell\rangle$.
Here, we view the state $|0_\ell\rangle$ as `off', `vacuum', or $0 \in \bbZ/2$, and $|1_\ell\rangle$ as `on' or $1\in \bbZ/2$.

Our preferred bases of the Hilbert spaces assigned to the vertical links depend on the particle to be condensed, and will be given in the subsequent sections.
In all cases, the full Hilbert space is the tensor product
$\cH = \bigotimes_{\ell} \cH_\ell$.

\subsubsection{Condensing \texorpdfstring{$e$}{e}}
\label{sssec:Z2ToricCode-Condensing-e}

We use $e$ to denote the vertex excitations in toric code; in other words, we say that there is an $e$-particle at a vertex $v$ in a state $\psi$ when $A_v|\psi\rangle=-1|\psi\rangle$.

The Hilbert spaces $\cH_\ell$ which we assign to vertical links are a direct sum $\bbC|\mathbbm{1}_\ell\rangle \oplus \bbC|e_\ell\rangle$, where $|\mathbbm{1}_\ell\rangle$ represents the unit of $D(\bbZ/2)$ and $|e_\ell\rangle$ represents the $e$-particle from $D(\bbZ/2)$.

We now describe the Hamiltonian for our lattice model.
For each vertex $v$, we define a \textbf{vertex term} $A_v$ by
$$
A_v:=
\tikzmath{
\draw[step=1.0, black, dotted] (-0.5,-0.5) grid (2.5,2.5);
\foreach \i in {0,1,2} {
\foreach \j in {0,1,2} {
        \draw[dotted] (\i,\j) -- ($ (\i,\j) + (-.3,.3)$);
}}
\draw[thick] (.7,1.3) node[above] {$\scriptstyle Z$} -- (1,1) node[below, xshift=.2cm] {$\scriptstyle v$};
\draw[thick] (0,1) node[left]{$\scriptstyle Z$} -- (2,1) node[right]{$\scriptstyle Z$};
\draw[thick] (1,0) node[below]{$\scriptstyle Z$} -- (1,2) node[above]{$\scriptstyle Z$};
}
$$
Here, $-A_v$ ensures that an even number of edges adjacent to $v$ are in the `on' position by imposing an energy penalty for an odd number of edges in the `on' position.

For each plaquette/face $p$, we define a \textbf{plaquette term} $B_p$ by
$$
B_p :=
\tikzmath{
\draw[step=1.0, black, dotted] (0.5,0.5) grid (2.5,2.5);
\foreach \i in {1,2} {
\foreach \j in {1,2} {
        \draw[dotted] (\i,\j) -- ($ (\i,\j) + (-.3,.3)$);
}}
\draw[thick] (1,1) -- node[left] {$\scriptstyle X$} (1,2) -- node[above] {$\scriptstyle X$} (2,2) -- node[right] {$\scriptstyle X$} (2,1) -- node[below] {$\scriptstyle X$} (1,1);
\node at (1.5,1.5) {$\scriptstyle p$};
}
$$
Here, $-B_p$ averages over states that are `off' and `on' by imposing an energy penalty.
It also ensures that a string to one side of the plaquette $p$ can be isotoped over $p$ to the other side at no energy cost in the ground state.

For each vertex $v$, we also have a new term, called the \textbf{unit term}:
$$
C_v:=
\tikzmath{
\draw[step=1.0, black, dotted] (0.5,0.5) grid (1.5,1.5);
\foreach \i in {1} {
\foreach \j in {1} {
        \draw[dotted] (\i,\j) -- ($ (\i,\j) + (-.3,.3)$);
}}
\draw[thick] (.7,1.3) node[above] {$\scriptstyle Z$} -- (1,1) node[below, xshift=.2cm] {$\scriptstyle v$};
}
$$
Here, $-C_v$ turns off the vertical edges by imposing an energy penalty on the state $|e_\ell\rangle$.

For each link $\ell$ on the 2D plane, we have a new term, the \textbf{condensation term}:
$$
D_\ell :=
\tikzmath{
\draw[step=1.0, black, dotted] (0.5,0.5) grid (2.5,1.5);
\foreach \i in {1,2} {
\foreach \j in {1} {
        \draw[dotted] (\i,\j) -- ($ (\i,\j) + (-.3,.3)$);
}}
\draw[thick] (.7,1.3) node[above] {$\scriptstyle X$} -- (1,1) -- node[above, xshift=-.2cm] {$\scriptstyle X$} node[below, xshift=-.2cm] {$\scriptstyle \ell$} (2,1) -- (1.7,1.3) node[above] {$\scriptstyle X$};
}
\quad
\text{or}
\quad
\tikzmath{
\draw[step=1.0, black, dotted] (0.5,0.5) grid (1.5,2.5);
\foreach \i in {1} {
\foreach \j in {1,2} {
        \draw[dotted] (\i,\j) -- ($ (\i,\j) + (-.3,.3)$);
}}
\draw[thick] (.7,1.3) node[left] {$\scriptstyle X$} -- (1,1) -- node[left, yshift=.2cm] {$\scriptstyle X$} node[right, yshift=.2cm] {$\scriptstyle \ell$} (1,2) -- (.7,2.3) node[left] {$\scriptstyle X$};
}.
$$
Here, $-D_\ell$ implements $m^*m$, where $m$ is the multiplication of the \'etale algebra $A=\mathbbm{1}\oplus e$, which we mean to condense.
Thus, the vertical links at each vertex and the horizontal links between two vertices each support a two-dimensional Hilbert space, but for different reasons: if $\ell$ is a horizontal link, the states $|0_\ell\rangle$ and $|1_\ell\rangle$ correspond to the elements $0$ and $1$ of $\mathbb{Z}/2$, while if $\ell$ is vertical, the states $|0_\ell\rangle$ and $|1_\ell\rangle$ correspond to summands $\mathbbm{1}$ and $e$ of $A$, respectively.

We define the Hamiltonian of the system for $t\in [0,1]$ to be
\begin{equation}
\label{eq:ToricCodeHamiltonian}
H_t
:=
-V\left(\sum_v A_v +\sum_p B_p \right)
-K\left((1-t)\sum_v C_v + t\sum_\ell D_\ell\right),
\end{equation}
where
$V>0$ is a constant
and
$K\gg V$ is a large constant.
Observe that the terms $A_v, B_p$ commute with one another, as well as with the $C_v$ and $D_\ell$ terms.
However, the $C_v$ and $D_\ell$ terms do not commute when $v$ is a source or target of $\ell$.

When $t=0$, the low energy physics of our model is equivalent to that of \cite{MR1951039}.
The fact that $K$ is large forces us into the ground state of the $C_v$ terms to analyze the low energy physics of the model.
In the ground state of the $C_v$ terms, each vertical edge is in the state $|\mathbbm{1}\rangle$.
Since $Z|\mathbbm{1}\rangle=|\mathbbm{1}\rangle$, on ground states of $C_v$, the $A_v$ terms agrees with
$$
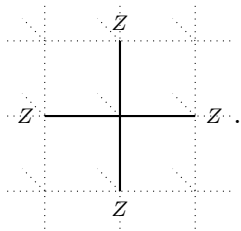
\tikzmath{
\draw[step=1.0, black, dotted] (-0.5,-0.5) grid (2.5,2.5);
\foreach \i in {0,1,2} {
\foreach \j in {0,1,2} {
        \draw[dotted] (\i,\j) -- ($ (\i,\j) + (-.3,.3)$);
}}
\draw[thick] (0,1) node[left]{$\scriptstyle Z$} -- (2,1) node[right]{$\scriptstyle Z$};
\draw[thick] (1,0) node[below]{$\scriptstyle Z$} -- (1,2) node[above]{$\scriptstyle Z$};
}\text{.}
$$
Hence, on the ground state of the $C_v$ terms, our Hamiltonian agrees exactly with the Hamiltonian for toric code given in \cite{MR1951039}, up to exchanging Pauli $X$ and $Z$.

When $t=1$, we may locally create and destroy individual $e$-particles at any vertex by applying the Pauli $X$ operator to the vertical edge, which commutes with $A_v$, $B_p$, and $D_\ell$ terms:
$$
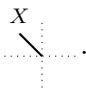
\tikzmath{
\draw[step=1.0, black, dotted] (0.5,0.5) grid (1.5,1.5);
\foreach \i in {1} {
\foreach \j in {1} {
        \draw[dotted] (\i,\j) -- ($ (\i,\j) + (-.3,.3)$);
}}
\draw[thick] (.7,1.3) node[above] {$\scriptstyle X$} -- (1,1);
}\text{.}
$$
Hence, when $t=1$, the $e$-particle is condensed.

Now suppose that we drop the assumption that $K\gg V$, and instead send $V\to\infty$ while $K$ remains fixed.
This pushes us into the ground state space of the usual Levin-Wen commuting projector Hamiltonian on the 2D lattice, leaving only degrees of freedom on the additional vertical edges.
The Hamiltonian on this reduced Hilbert space is then
$$
\widetilde{H}_t
:=
-K\left((1-t)\sum_v C_v + t\sum_\ell D_\ell\right)
=
-K\left((1-t)\sum_v Z_v + t\sum_{\ell} X_{s(\ell)}X_{t(\ell)}\right).
$$
Here, $s(\ell), t(\ell)$ denote the source and target of the edge $\ell$.
This Hamiltonian is just the 2D transverse-field Ising model \cite{PhysRev.65.117}, so tuning $t$ from $0$ to $1$ drives the system through a well-studied quantum phase transition \cite{1811.09275,1901.00278}.
Of course, when $K\gg V$ and $K$ and $V$ remain constant, the overall story is more complicated.
However, this special case motivates the analogy to the 2D transverse field Ising model made in \S~\ref{ssec:GSD}.

\subsubsection{Condensing \texorpdfstring{$m$}{m}}
\label{sssec:Z2ToricCode-Condensing-m}

In order to condense $m$ excitations, we slightly alter the previous model,  adapting certain terms in the Hamiltonian.
We relabel our preferred basis for $\mathcal{H}_\ell$ when $\ell$ is vertical link, replacing $|e_\ell\rangle$ with $|m_\ell\rangle$.
The operators $C_v$ remain unchanged, while $A_v$, $B_p$, and $D_\ell$ are adjusted as follows, with the changes highlighted in \textcolor{blue}{blue}.
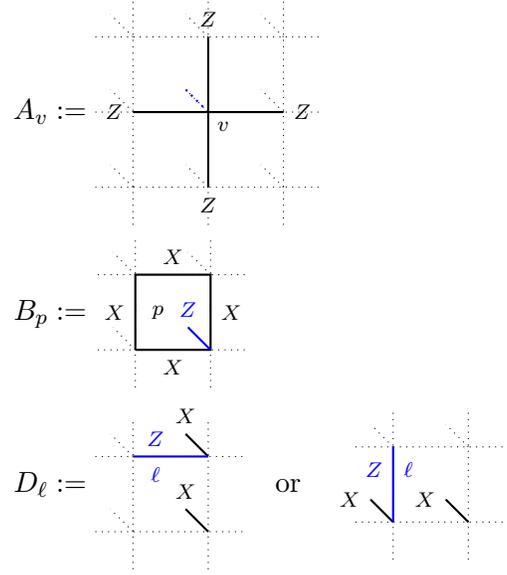
\begin{align*}
A_v &:=
\tikzmath{
\draw[step=1.0, black, dotted] (-0.5,-0.5) grid (2.5,2.5);
\foreach \i in {0,1,2} {
\foreach \j in {0,1,2} {
        \draw[dotted] (\i,\j) -- ($ (\i,\j) + (-.3,.3)$);
}}
\draw[dotted, thick, blue] (1,1) -- (.7,1.3);
\node at (1,1) [below, xshift=.2cm] {$\scriptstyle v$};
\draw[thick] (0,1) node[left]{$\scriptstyle Z$} -- (2,1) node[right]{$\scriptstyle Z$};
\draw[thick] (1,0) node[below]{$\scriptstyle Z$} -- (1,2) node[above]{$\scriptstyle Z$};
}
\\
B_p &:=
\tikzmath{
\draw[step=1.0, black, dotted] (0.5,0.5) grid (2.5,2.5);
\foreach \i in {1,2} {
\foreach \j in {1,2} {
        \draw[dotted] (\i,\j) -- ($ (\i,\j) + (-.3,.3)$);
}}
\draw[thick] (1,1) -- node[left] {$\scriptstyle X$} (1,2) -- node[above] {$\scriptstyle X$} (2,2) -- node[right] {$\scriptstyle X$} (2,1) -- node[below] {$\scriptstyle X$} (1,1);
\node at (1.3,1.5) {$\scriptstyle p$};
\draw[thick, blue] (2,1) -- (1.7,1.3) node[above] {$\scriptstyle Z$};
}
\\
D_\ell &:=
\tikzmath{
\draw[step=1.0, black, dotted] (0.5,0.5) grid (2.5,2.5);
\foreach \i in {1,2} {
\foreach \j in {1,2} {
        \draw[dotted] (\i,\j) -- ($ (\i,\j) + (-.3,.3)$);
}}
\draw[thick] (2,1) -- (1.7,1.3) node[above] {$\scriptstyle X$};
\draw[thick] (2,2) -- (1.7,2.3) node[above] {$\scriptstyle X$};
\draw[thick, blue] (1,2) -- node[above, xshift=-.2cm] {$\scriptstyle Z$} node[below, xshift=-.2cm] {$\scriptstyle \ell$} (2,2);
}
\quad
\text{or}
\quad
\tikzmath{
\draw[step=1.0, black, dotted] (0.5,0.5) grid (2.5,2.5);
\foreach \i in {1,2} {
\foreach \j in {1,2} {
        \draw[dotted] (\i,\j) -- ($ (\i,\j) + (-.3,.3)$);
}}
\draw[thick] (1,1) -- (.7,1.3) node[left] {$\scriptstyle X$};
\draw[thick] (2,1) -- (1.7,1.3) node[left] {$\scriptstyle X$};
\draw[thick, blue] (1,1) -- node[left, yshift=.2cm] {$\scriptstyle Z$} node[right, yshift=.2cm] {$\scriptstyle \ell$} (1,2);
}.
\end{align*}

Note that the operator $D_\ell$ in this example becomes the $D_\ell$ of the previous example, if we pass to the dual lattice and apply a change of basis exchanging $X$ and $Z$ operators to each link of the square grid.
Since the vertical spaces $\mathcal{H}_\ell$ have a basis $\{|1_\ell\rangle,|m_\ell\rangle\}$, and passage to the dual lattice exchanges $e$- and $m$-particles, we apply the same operator to the vertical links in both examples.
Consequently, we obtain toric code at $t=0$, $m$ is condensed at $t=1$, and the phase transition maps onto a 2D transverse-field Ising model, as in the previous example.

\subsection{\texorpdfstring{$\bbZ/n$}{Z/n} Toric code}

The next simplest example of a Levin-Wen model is the $\bbZ/n$ toric code.
The matrices $X$ and $Z$ are replaced by $n \times n$ Pauli matrices $X_n$ and $Z_n$.
For example, when $n=3$,
$$
X_3=
\begin{pmatrix}
0 & 0 & 1 \\
1 & 0 & 0 \\
0 & 1 & 0
\end{pmatrix}
\qquad
\text{ and }
\qquad
Z_3=
\begin{pmatrix}
1 & 0 & 0 \\
0 & e^{\frac{2\pi i}{3}} & 0 \\
0 & 0 & e^{\frac{4\pi i}{3}}
\end{pmatrix}
=
\begin{pmatrix}
\zeta^0 & 0 & 0 \\
0 & \zeta & 0 \\
0 & 0 & \zeta^2
\end{pmatrix}
$$
where $\zeta=e^{\frac{2\pi i}{3}}$ is a primitive cube root of unity.

The Hilbert space on each link is $H_\ell = \bbC^n$.
The spaces $H_\ell$ on the 2D lattice are direct sums $\bbC|0_\ell\rangle \oplus \bbC |1_\ell\rangle \oplus \cdots \oplus |(n-1)_\ell\rangle$
where $|0_\ell\rangle$ is `off' or `vacuum', while the remaining $n-1$ states are distinct `on' states in $\bbZ/n$.

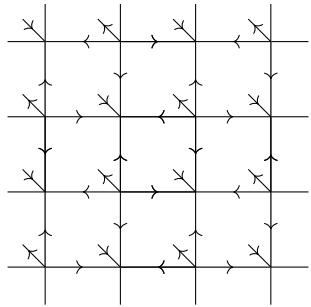
$$
\tikzmath{
\draw[step=1.0] (0.5,0.5) grid (4.5,4.5);
\filldraw[white] (1,1) rectangle (4,4);
\foreach \i in {1,3} {
\foreach \j in {1,3} {
        \draw[late>] (\i,\j) -- ($ (\i,\j) + (-.3,.3)$);
}}
\foreach \i in {2,4} {
\foreach \j in {2,4} {
        \draw[late>] (\i,\j) -- ($ (\i,\j) + (-.3,.3)$);
}}
\foreach \i in {1,3} {
\foreach \j in {2,4} {
        \draw[late<] (\i,\j) -- ($ (\i,\j) + (-.3,.3)$);
}}
\foreach \i in {2,4} {
\foreach \j in {1,3} {
        \draw[late<] (\i,\j) -- ($ (\i,\j) + (-.3,.3)$);
}}
\foreach \i in {1,3} {
\foreach \j in {1,3} {
        \draw[mid>] (\i,\j) -- ($ (\i+1,\j) $); 
        \draw[mid>] (\i,\j) -- ($ (\i,\j+1) $); 
        \draw[mid>] ($ (3,\j) $) -- ($ (2,\j) $); 
        \draw[mid>] ($ (\i,3) $) -- ($ (\i,2) $); 
        \draw[mid>] ($ (\i+1,\j+1) $) -- ($ (\i+1,\j) $); 
        \draw[mid>] ($ (2,\j+1) $) -- ($ (3,\j+1) $); 
        \draw[mid>] ($ (\i+1,2) $) -- ($ (\i+1,3) $); 
        \draw[mid>] ($ (\i+1,\j+1) $) -- ($ (\i,\j+1) $); 
}}
}
$$

\subsubsection{Condensing \texorpdfstring{$e^k$}{ek}}
\label{sec:ZnToricCode-Condensing-e}

Now fix $k \mid n$. We will condense the algebra $A = \mathbbm{1} + e^k + \cdots e^{k(n/k - 1)}$.
The vertical spaces $\cH_\ell=\bbC^{n/k}$ are therefore a direct sum
$\bbC|\mathbbm{1}_\ell\rangle \oplus \bbC |e^k_\ell\rangle \oplus \cdots \oplus |e^{k(n/k-1)}_\ell\rangle$ where $|\mathbbm{1}_\ell\rangle$ represents the unit of $D(\bbZ/n)$ and the states $|e^{jk}_\ell\rangle$ represent powers of the $e^k$-particle from $D(\bbZ/n)$.

We now modify the four types of operators $A_v, B_p, C_v, D_\ell$ from \S\ref{sssec:Z2ToricCode-Condensing-e}, and the Hamiltonian $H_t$ has the same formula \eqref{eq:ToricCodeHamiltonian}.
We define the vertex term by
$$
A_v :=
\tikzmath{
\draw[step=1.0, black, dotted] (-0.5,-0.5) grid (2.5,2.5);
\foreach \i in {0,1,2} {
\foreach \j in {0,1,2} {
        \draw[dotted] (\i,\j) -- ($ (\i,\j) + (-.3,.3)$);
}}
\draw[thick] (.7,1.3) node[above, xshift=-.1cm] {$\scriptstyle Z_{n/k}$} -- (1,1) node[below, xshift=.2cm] {$\scriptstyle v$};
\draw[thick] (0,1) node[left]{$\scriptstyle Z_n$} -- (2,1) node[right]{$\scriptstyle Z_n$};
\draw[thick] (1,0) node[below]{$\scriptstyle Z_n$} -- (1,2) node[above]{$\scriptstyle Z_n$};
}
+
\tikzmath{
\draw[step=1.0, black, dotted] (-0.5,-0.5) grid (2.5,2.5);
\foreach \i in {0,1,2} {
\foreach \j in {0,1,2} {
        \draw[dotted] (\i,\j) -- ($ (\i,\j) + (-.3,.3)$);
}}
\draw[thick] (.7,1.3) node[above, xshift=-.1cm] {$\scriptstyle Z^{\dag}_{n/k}$} -- (1,1) node[below, xshift=.2cm] {$\scriptstyle v$};
\draw[thick] (0,1) node[left]{$\scriptstyle Z^{\dag}_n$} -- (2,1) node[right]{$\scriptstyle Z^{\dag}_n$};
\draw[thick] (1,0) node[below]{$\scriptstyle Z^{\dag}_n$} -- (1,2) node[above]{$\scriptstyle Z^{\dag}_n$};
}
$$
Observe that given a simple tensor $x$ in $\otimes_{\ell \sim v}\cH_{\ell}$ in the standard basis, $A_v$ preserves $x$ if and only if the legs of $x$ sum to $0\mod n$.
Thus, rather than ensure an even number of links $\ell$ are in the state $|1_\ell\rangle$ at a given vertex,
$-A_v$ now ensures that the links surrounding $v$ sum to $0\mod n$.

The plaquette term $-B_p$ averages over all $n$ possible states on each link rather than only two; we define $B_p$ explicitly by the following:
$$
B_p:=
\tikzmath{
\draw[step=1.0, black, dotted] (0.5,0.5) grid (2.5,2.5);
\foreach \i in {1,2} {
\foreach \j in {1,2} {
        \draw[dotted] (\i,\j) -- ($ (\i,\j) + (-.3,.3)$);
}}
\draw[thick] (1,1) -- node[left] {$\scriptstyle X_n$} (1,2) -- node[above] {$\scriptstyle X_n^\dag$} (2,2) -- node[right] {$\scriptstyle X_n$} (2,1) -- node[below] {$\scriptstyle X_n^\dag$} (1,1);
\node at (1.5,1.5) {$\scriptstyle p$};
}
+
\tikzmath{
\draw[step=1.0, black, dotted] (0.5,0.5) grid (2.5,2.5);
\foreach \i in {1,2} {
\foreach \j in {1,2} {
        \draw[dotted] (\i,\j) -- ($ (\i,\j) + (-.3,.3)$);
}}
\draw[thick] (1,1) -- node[left] {$\scriptstyle X^\dag_n$} (1,2) -- node[above] {$\scriptstyle X_n$} (2,2) -- node[right] {$\scriptstyle X^\dag_n$} (2,1) -- node[below] {$\scriptstyle X_n$} (1,1);
\node at (1.5,1.5) {$\scriptstyle p$};
}
$$

The unit term is given by
$$
C_v:=
\tikzmath{
\draw[step=1.0, black, dotted] (0.5,0.5) grid (1.5,1.5);
\foreach \i in {1} {
\foreach \j in {1} {
        \draw[dotted] (\i,\j) -- ($ (\i,\j) + (-.3,.3)$);
}}
\draw[thick] (.7,1.3) node[above] {$\scriptstyle Z_{n/k}$} -- (1,1) node[below, xshift=.2cm] {$\scriptstyle v$};
}
+
\tikzmath{
\draw[step=1.0, black, dotted] (0.5,0.5) grid (1.5,1.5);
\foreach \i in {1} {
\foreach \j in {1} {
        \draw[dotted] (\i,\j) -- ($ (\i,\j) + (-.3,.3)$);
}}
\draw[thick] (.7,1.3) node[above] {$\scriptstyle Z^\dag_{n/k}$} -- (1,1) node[below, xshift=.2cm] {$\scriptstyle v$};
}
$$
Here, $-C_v$ turns off the vertical edges by imposing an energy penalty on every state except $|\mathbbm{1}_\ell\rangle$.

Finally, the condensation term is given by
$$
D_\ell:=
\tikzmath[scale=1.5]{
\draw[step=1.0, black, dotted] (0.5,0.5) grid (2.5,1.5);
\foreach \i in {1,2} {
\foreach \j in {1} {
        \draw[dotted] (\i,\j) -- ($ (\i,\j) + (-.3,.3)$);
}}
\draw[thick] (.7,1.3) node[above] {$\scriptstyle X_{n/k}$} -- (1,1) -- node[above, xshift=-.2cm] {$\scriptstyle(X_n^k)^\dag$} node[below, xshift=-.2cm] {$\scriptstyle \ell$} (2,1) -- (1.7,1.3) node[above] {$\scriptstyle X_{n/k}$};
}
+
\tikzmath[scale=1.5]{
\draw[step=1.0, black, dotted] (0.5,0.5) grid (2.5,1.5);
\foreach \i in {1,2} {
\foreach \j in {1} {
        \draw[dotted] (\i,\j) -- ($ (\i,\j) + (-.3,.3)$);
}}
\draw[thick] (.7,1.3) node[above] {$\scriptstyle X_{n/k}^\dag$} -- (1,1) -- node[above, xshift=-.2cm] {$\scriptstyle X_n^k$} node[below, xshift=-.2cm] {$\scriptstyle \ell$} (2,1) -- (1.7,1.3) node[above] {$\scriptstyle X_{n/k}^\dag$};
}
\qquad
\text{or}
\qquad
\tikzmath[scale=1.5]{
\draw[step=1.0, black, dotted] (0.5,0.5) grid (1.5,2.5);
\foreach \i in {1} {
\foreach \j in {1,2} {
        \draw[dotted] (\i,\j) -- ($ (\i,\j) + (-.3,.3)$);
}}
\draw[thick] (.7,1.3) node[left] {$\scriptstyle X_{n/k}$} -- (1,1) -- node[left, yshift=.2cm] {$\scriptstyle (X_n^k)^\dag$} node[right, yshift=.2cm] {$\scriptstyle \ell$} (1,2) -- (.7,2.3) node[left] {$\scriptstyle X_{n/k}$};
}
+\quad
\tikzmath[scale=1.5]{
\draw[step=1.0, black, dotted] (0.5,0.5) grid (1.5,2.5);
\foreach \i in {1} {
\foreach \j in {1,2} {
        \draw[dotted] (\i,\j) -- ($ (\i,\j) + (-.3,.3)$);
}}
\draw[thick] (.7,1.3) node[left] {$\scriptstyle X_{n/k}^\dag$} -- (1,1) -- node[left, yshift=.2cm] {$\scriptstyle X_n^k$} node[right, yshift=.2cm] {$\scriptstyle \ell$} (1,2) -- (.7,2.3) node[left] {$\scriptstyle X_{n/k}^\dag$};
}.
$$
Here, $D_\ell$ condenses $e^k$ by implementing the self-adjoint operator $m^*m + mm^*$, where $m$ is the algebra multiplication $AA \to A$.

As in previous examples, when $t=0$, we obtain a usual model \cite[\S5.3]{MR3157248} for $\mathbb{Z}/n$ toric code.
When $t=1$, our system exhibits $D(\mathbb{Z}/k)$ topological order.
By locally applying the operator $X_{n/k}$ (or $X_{n/k}^\dag$, depending on parity) on a vertical edge, we may create an $e^k$ particle at any vertex.
Therefore, while the charge particles $e^j$ still exist, the value of $j$ is only well-defined modulo $k$.
Similarly, applying the operator $Z_n$ (or $Z_n^\dag$, depending on parity) along an edge in the square grid still lets us excite two adjacent plaquettes, so the flux particle $m$ still exists.
However, $m$ anyons are now confined, since applying $Z_n$ on a link $\ell$ only commutes with $D_\ell$ up to a phase of $\zeta^{\pm k}$; $m$ is no longer a topological excitation.
In fact, $(Z_n^j)_{\ell}$ commutes with $D_\ell$ if and only if $n|jk$, i.e. if $j$ is a multiple of $n/k$.
In other words, the flux particles which are not confined are just powers of $\widetilde{m}:=m^{n/k}$.
In terms of fusion rules, $e$ and $\widetilde{m}$ are both bosons, and $e^k$ and $\widetilde{m}^k$ are the vacuum.
Also, $S_{e\widetilde{m}}=\zeta^{1\cdot{n/k}}$ is a primitive $k$-th root of unity.
This shows that $e$ and $\widetilde{m}$, which generate the topologically mobile vertex and plaquette excitations at $t=1$, generate the braided tensor category $D(\mathbb{Z}/k)$.

\subsubsection{Condensing \texorpdfstring{$m^k$}{mk}}

In \S\ref{sssec:Z2ToricCode-Condensing-m}, we made some basic changes in our model for condensing $e$ in \S\ref{sssec:Z2ToricCode-Condensing-e} for $\bbZ/2$-toric code.
Using similar modifications to \S\ref{sec:ZnToricCode-Condensing-e} above for condensing $e^k$, one gets a model for condensing $m^k$ for any $k\mid n$.
We leave the details to the interested reader.


\subsection{Doubled semion - condensing \texorpdfstring{$b$}{b}}
Our doubled semion model uses a modified hexagonal lattice.
Similar to the terminology of \cite{PhysRevB.71.045110}, given a hexagonal face/plaquette of our lattice, we call the six links bounding the plaquette \emph{edges} and the six links emanating outward from the vertices of the plaquette \emph{legs}:
$$
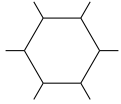
\tikzmath{
\levinHex{0}{0}{.5}{black}
}
$$
We now add vertical links on the face of the plaquette, close to the corner of the rightmost vertex of the plaquette.
We draw this extra vertical link approximately parallel to the north-east edge of the plaquette.
$$
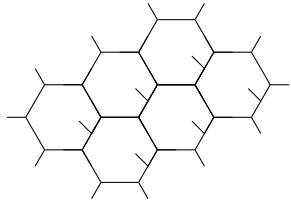
\tikzmath{
\levinHexGrid[black]{0}{0}{2}{3}{.5}{black}
}
$$
As in $\bbZ/2$ toric code, we define the space $H_\ell=\bbC^2=\bbC |0_\ell\rangle + \bbC |1_\ell\rangle$ on each link in the plane.

The ``doubled semion'' modular tensor category is the Drinfel'd center of the fusion category $\Vec[\mathbb{Z}/2,-1]$, the category of $\mathbb{Z}/2$-graded vectorspaces where $\alpha_{g,g,g}=-1$, where $g\in\mathbb{Z}/2$ is the nontrivial group element.
Excitations in the usual doubled semion model come in three forms: the semions $\sigma,\overline{\sigma}$ and the boson $b$.
The boson $b$ is an excitation of the plaquette operator, which exists in a state $\psi$ when $B_p|\psi\rangle=-|\psi\rangle$, but $\psi$ is in the ground state of $A_v$ for each vertex $v$ of $p$.
The semions $\sigma$ and $\overline{\sigma}$ are both excitations of vertex terms, and correspond to the simple objects in $Z(\Vec[\mathbb{Z}/2,-1])$ with underlying object $g$ and half-braiding with $g$ given by $\pm i$.


In our model, the requisite operators $A_v, B_p, C_{p'}$ and $D_\ell$ are given by
\footnote{In defining $B_p$, we extend the definition originally appearing in \cite[Section VI A]{PhysRevB.71.045110}, rather than implementing the general definition we give in \S~\ref{ssec:lwModel}.
This emphasizes the fact that the strategy for modifying lattice models to accomplish anyon condensation described in Section~\ref{sec:condensationModels} is based on features of the topological order, and therefore robust to small changes in the details of the original model.}
\begin{align*}
A_v &:=\quad
\tikzmath{
\levinHex[black, dotted]{0}{0}{.5}{black, dotted}
\levinHex[black, dotted]{.75}{.433}{.5}{black, dotted}
\levinHex[black, dotted]{.75}{-.433}{.5}{black, dotted}
\draw (0:.5cm) node[below, xshift=.1cm] {$\scriptstyle v$} -- (60:.5cm) node[above] {$\scriptstyle Z$};
\draw (0:.5cm) -- (-60:.5cm) node[below] {$\scriptstyle Z$};
\draw (0:.5cm) -- (0:1cm) node[right] {$\scriptstyle Z$} ;
}
\displaybreak[1]\\
B_p &:=\quad
\tikzmath{
\nhex[black]{0}{0}{.5}{black}
\levinHex[black, dotted]{-.75}{.433}{.5}{black, dotted}
\levinHex[black, dotted]{-.75}{-.433}{.5}{black, dotted}
\levinHex[black, dotted]{.75}{.433}{.5}{black, dotted}
\levinHex[black, dotted]{.75}{-.433}{.5}{black, dotted}
\levinHex[black, dotted]{0}{.866}{.5}{black, dotted}
\levinHex[black, dotted]{0}{-.866}{.5}{black, dotted}
\foreach \i in {0,60,120,180,240,300} {
\draw (\i:.5cm) -- (\i:1cm);
\node at (\i:1.2cm) {$\scriptstyle W$};
}
\foreach \i in {0,60,180,240,300} {
\node at ($ (\i+30:.633cm) $) {$\scriptstyle X$};
}
\node at (150:.3) {$\scriptstyle X$};
\node at (.075,-.2) {$\scriptstyle Z$};
}\cdot \mathbf{P} \quad \text{ where } \quad W := \begin{bmatrix}1 &0 \\ 0 & i \end{bmatrix}
\displaybreak[1]\\
C_p &:=\quad
\tikzmath{
\levinHex[black]{0}{0}{.5}{black, dotted}
\node at (.075,-.2) {$\scriptstyle Z$};
}
\displaybreak[1]\\
D_\ell &:=\quad
\tikzmath{
\levinHex[black]{0}{0}{.5}{black, dotted}
\levinHex[black]{.75}{.433}{.5}{black, dotted}
\draw (0:.5cm) -- node[left] {$\scriptstyle \ell$} node[right] {$\scriptstyle Z$} (60:.5cm);
\node at (.075,-.2) {$\scriptstyle X$};
\node at (1,.5) {$\scriptstyle X$};
}
\quad,\quad
\tikzmath{
\levinHex[black]{0}{0}{.5}{black, dotted}
\levinHex[black]{0}{.866}{.5}{black, dotted}
\draw (60:.5cm) -- node[below] {$\scriptstyle \ell$} node[above] {$\scriptstyle Z$} (120:.5cm);
\node at (.075,-.2) {$\scriptstyle X$};
\node at (.075,1) {$\scriptstyle X$};
}
\quad,\text{ or}\quad
\tikzmath{
\levinHex[black]{0}{0}{.5}{black, dotted}
\levinHex[black]{.75}{-.433}{.5}{black, dotted}
\draw (0:.5cm) -- node[left] {$\scriptstyle \ell$} node[right] {$\scriptstyle Z$} (-60:.5cm);
\node at (.2,.2) {$\scriptstyle X$};
\node at (.9,-.7) {$\scriptstyle X$};
}
\end{align*}

$A_v$ imposes a zero-sum around a given vertex, while $B_p$ averages over a state and its opposite up to a phase imposed by the number of outgoing legs which are in the state $|1\rangle$.
The term $B_p$ inserts a loop around the plaquette $p$, as in toric code, while applying a phase to account for the nontrivial associator.
Because inserting a loop only makes sense in the context of diagrammatic calculus, we precompose with $\mathbf{P}$, the projection into the ground state of $\prod_v A_v$.
Thus, $B_p$ only inserts a loop when no vertex of $p$ contains an excitation; otherwise, $B_p$ is a scalar multiple of the identity.
Since $-\mathbf{P}$ (and hence $-B_p$) favors a lack of vertex excitations, and $A_v$ and $B_p$ terms still commute, this does not affect the physics, as asserted in \cite{PhysRevB.71.045110}.

As defined above, an excitation of type $b$ occurs on a plaquette $p$ when $B_p$ has the eigenvalue $-1$.
Thus, when $t=1$, we may locally apply the operator $X$ to a vertical edge, locally creating or destroying a $b$-excitation on any plaquette, just as when we condensed $m$ in $\mathbb{Z}/2$ toric code in \S\ref{sssec:Z2ToricCode-Condensing-m}.
Similarly, the action of $D_\ell$ on a state containing a $\sigma$ excitation will produce the superposition $\sigma\oplus\overline{\sigma}$.
String operators attempting to move a $\sigma$ particle more than a single link will anticommute with some $D_\ell$ term, so the semions are now confined, leaving no nontrivial topological excitations.

\subsection{Doubled Fibonacci}
The unitary fusion category $\Fib$ has simple objects $1$ and $\tau$, and the fusion rule $\tau\otimes\tau\cong1\oplus\tau$.
This makes $\Fib$ the smallest fusion category (in several senses) which is not pointed, i.e. where the simple objects do not form a group.
$\Fib$ can be constructed in several ways, including as the semisimple quotient of the Temperly-Lieb category at index $\phi$ \cite{MR1280463}.
In \cite{MR2889539}, the associator and braidings on categories with these fusion rules are determined algebraically; see also \cite{10.1143/PTPS.176.384}, although they pick the opposite braiding.
The fusion category $\Fib$ consists of the following data.
\begin{itemize}
 \item Simple objects (edge labels): $\{1,\tau\}$.
 \item Quantum dimension $d_1=1$, $d_\tau=\phi=\frac{1+\sqrt{5}}{2}$.
 \item Fusion rules $1\otimes 1\cong 1$, $1\otimes\tau\cong\tau\otimes 1\cong\tau$, $\tau\otimes\tau\cong1\oplus\tau$.
 \item $F$-symbols:
  \begin{align*}
   F_{\tau,\tau}^{\tau,\tau}[\tau,\tau] &= -\phi^{-1}\\
   F_{\tau,\tau}^{\tau,\tau}[1,1] &= \phi^{-1}\\
   F_{\tau,\tau}^{\tau,\tau}[\tau,1] &= \phi^{-1/2}\\
   F_{\tau,\tau}^{\tau,\tau}[1,\tau] &= \phi^{-1/2}
  \end{align*}
  with all other $F$-symbols being $1$.
  Here, $\phi=\frac{1+\sqrt{5}}{2}$ is the golden ratio, and $-\phi^{-1}=\frac{1-\sqrt{5}}{2}$ is the other root of $x^2-x-1$.
\end{itemize}

  The fusion category $\Fib$ can be made into a UMTC with two different braidings, which are reverse to one another; we arbitrarily denote one by $\Fib$, so that the other is $\overline{\Fib}$, where $\overline{\cdot}$ refers to the fact that the braiding is reversed.
  The UMTC $\Fib$ has the following additional data.
  \begin{itemize}
   \item Braiding and $R$-symbols:
    The braiding between $\tau$ and $\tau$ is given by
    \[\beta_{\tau,\tau}=q^3~\tikzmath[scale=.33]{
     \draw (-.5,1)--(0,.5);
     \draw (.5,1)--(0,.5);
     \draw (-.5,-1)--(0,-.5);
     \draw (.5,-1)--(0,-.5);
     \draw (0,.5)--(0,-.5);
     }+q^6\phi^{-1}~\tikzmath[scale=.33]{
     \draw (-.5,1)--(0,.5);
     \draw (.5,1)--(0,.5);
     \draw (-.5,-1)--(0,-.5);
     \draw (.5,-1)--(0,-.5);}\text{,}
    \]
    where $q=e^{2\pi i/10}$ is a primitive $10$th-root of unity such that $\phi=q+q^{-1}$.
    Equivalently, we list the $R$-symbols
    $
    \tikzmath{
    \draw (0,0) -- (0,-.4) node[below]{$\scriptstyle c$};
    \draw (-.2,.4) node[above]{$\scriptstyle a$} .. controls +(-90:.2) and +(45:.4) .. (0,0);
    \draw[knot] (.2,.4) node[above]{$\scriptstyle b$} .. controls +(-90:.2) and +(135:.4) .. (0,0);
    \filldraw (0,0) circle (.05cm);
    }
    =
    R_{c}^{ab}
    \tikzmath{
    \draw (0,0) -- (0,-.4) node[below]{$\scriptstyle c$};
    \filldraw (0,0) circle (.05cm);
    \draw (-.2,.4) node[above]{$\scriptstyle a$} -- (0,0) -- (.2,.4) node[above]{$\scriptstyle b$};
    }
    $.
    For $\Fib$, $R_{\tau}^{\tau,\tau}=q^3$, $R_1^{\tau,\tau}=q^6$, and all other $R$-symbols are $1$.
   \item $S$-matrix: $S=\frac{1}{\sqrt{1+\phi^2}}\left(\begin{array}{cc}1&\phi\\\phi&-1\end{array}\right)$
   \item $T$-matrix: $T=\left(\begin{array}{cc}1&0\\0&q^4\end{array}\right)$
  \end{itemize}
  To obtain the UMTC $\overline{\Fib}$, we replace $q$ with $q^{-1}$, to obtain the following.
  \begin{itemize}
   \item $R$-symbols: $R_{\tau}^{\tau,\tau}=q^7$, $R_1^{\tau,\tau}=q^4$, all other $R$ symbols are $1$.
   \item $S$-matrix: identical to $\Fib$.
   \item $T$-matrix: $T=\left(\begin{array}{cc}1&0\\0&q^6\end{array}\right)$.
  \end{itemize}

  By \cite[Rem.~4.3]{MR1990929}, we have a braided tensor equivalence $Z(\Fib)\cong\Fib\boxtimes\overline{\Fib}$, allowing us to derive the data of $Z(\Fib)$ from that of $\Fib$ and $\overline{\Fib}$.
  Explicitly, the UMTC $Z(\Fib)$ has the following data.
  \begin{itemize}
   \item Simple objects (anyon types): $\{1\boxtimes 1,\tau\boxtimes 1,1\boxtimes\tau,\tau\boxtimes\tau\}$.
    For brevity, we rename the anyons as follows: $1:=1\boxtimes 1$, $\tau:=\tau\boxtimes1$, $\overline{\tau}:=1\boxtimes\tau$, $b:=\tau\boxtimes\tau$.
    Here, $\overline{\tau}$ does not denote the dual of $\tau$ but rather the fact that $\overline{\tau}\in\overline{\Fib}\subseteq Z(\Fib)$ carries the inverse half-braiding.
    We choose the name $b$ for $\tau\boxtimes\overline{\tau}$ because, as we will see, $b$ is a non-Abelian boson.
   \item Fusion rules:
    \[\begin{array}{c|c|c|c|c}
     \otimes & 1 & \tau & \overline{\tau} & b
     \\\hline
     1 & 1 & \tau & \overline{\tau} & b
     \\\hline
     \tau & \tau & 1\oplus\tau & b & \overline{\tau}\oplus b
     \\\hline
     \overline{\tau} & \overline{\tau} & b & 1\oplus\overline{\tau} & \tau\oplus b
     \\\hline
     b & b & \overline{\tau}\oplus b & \tau\oplus b & 1\oplus\tau\oplus\overline{\tau}\oplus b
    \end{array}\]
   \item The $F$-symbols, $R$-symbols, $S$-matrix and $T$-matrix can be obtained by tensoring the matrices for $\Fib$ and $\overline{\Fib}$.
    In particular, we can see that $b$ is a boson.
    First, $b$ has trivial twist $\theta_b=1$, because
    \[\theta_b=\theta_\tau\theta_{\overline{\tau}}\beta_{\tau,\overline{\tau}}\beta_{\overline{\tau},\tau}=(q^4\cdot q^6)\id_{\tau\overline{\tau}}=\id_b\text{,}\]
    since the $\tau$ and $\overline{\tau}$ particles are transparent to one another.
    We can also compute the $R$-symbols $R_b^{b,b}=R_{\tau}^{\tau,\tau}R_{\overline{\tau}}^{\overline{\tau},\overline{\tau}}=1$ and $R_1^{b,b}=R_1^{\tau,\tau}R_1^{\overline{\tau},\overline{\tau}}=1$.
    Note that $b$ is not bosonic in all fusion channels: $R_\tau^{b,b}=R_\tau^{\tau,\tau}R_1^{\overline{\tau},\overline{\tau}}=q^3$, and $R_{\overline{\tau}}^{b,b}=R_1^{\tau,\tau}R_{\overline{\tau}}^{\overline{\tau},\overline{\tau}}=q^7$.
  \end{itemize}

The double $Z(\Fib)$ contains a single nontrivial connected \'etale algebra $A=1\oplus b$, which is the canonical Lagrangian algebra \cite[\S3.2]{MR3022755} $A=\textbf{1}\oplus\tau\boxtimes\overline{\tau}$ in $\Fib\boxtimes\overline{\Fib}$.
Since $A$ is Lagrangian, $Z(\Fib)_A^{\loc}\cong\Vec$, and the condensed phase has trivial topological order.
However, this $A$ is a minimal example where the underlying object of $A$ is not just the direct sum of invertible objects.

Since all fusion spaces in $\Fib$ are $1$-dimensional, i.e. $\Fib$ is multiplicity free, we can implement the condensation of $A$ on the hexagonal lattice, with spins associated to each link.
(One could instead associate a spin to each vertex, but this would lead to higher dimensional local Hilbert spaces, which would make the Hamiltonian more complicated to write down).
Each ordinary link receives a $2$-dimensional Hilbert space, with basis vectors $|1\rangle$ and $|\tau\rangle$ labeled by the simple objects of $\Fib$, while the vertical red link receives a $3$-dimensional Hilbert space, with basis vectors $|1\rangle$, $|\iota\rangle$, and $|\theta\rangle$, where $\iota$ and $\theta$ are basis vectors for the $1$-dimensional spaces $\Fib(1\to F(\tau\boxtimes\overline{\tau}))$ and $\Fib(\tau\to F(\tau\boxtimes\overline{\tau}))$ respectively.

The original Levin-Wen Hamiltonian (without the added vertical links) associated to $\Fib$ is described in \cite[\S~\RN{6}.B]{PhysRevB.71.045110}, and our Hamiltonian will be a modification.
The $A_v$ term projects onto the subspace where $0$, $2$, or $3$ of the links which meet at $v$ are labelled by $|\tau\rangle$ or $|\theta\rangle$.
For a vertical link $\ell$, $C_\ell$ projects onto the subspace spanned by $|1\rangle$.

The $B_p$ and $D_{p,q}$ terms are more complicated.
As in \cite{PhysRevB.71.045110}, we will not write out the $B_p$ term explicitly, because, as is generally the case for objects in a fusion category with dimension greater than $1$, $B_p^\tau$ does not factor as a tensor product of operators local to a smaller region, meaning that the final description is not more concise or enlightening than \eqref{eq:Bpc}.
For the same reasons, we will not write out the entire $D_{p,q}$ term.
However, we have given definitions of $B_p$ and $D_{p,q}$ as linear combinations of tensor products of operators local at each vertex, and we will explicitly compute those local operators which involve the half-braiding, multiplication, and separator of $A$.

First, we begin with the $B_p^\tau$ term, which involves the half-braiding of $A$ when the inserted $\tau$-loop crosses under the vertical link of $p$.
When resolving the $B_p^\tau$ term \eqref{eq:newBps} as in \eqref{eq:bpExpansion}, we end up with local operators of the following form at the vertex $v$ incident to the vertical link.
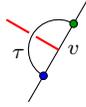
\[\tikzmath{
 \draw (60:-.8) -- (60:.8);
 \draw[thick,red] (0:0) -- (150:.8);
 \draw[knot] (60:-.4) arc (240:60:.4cm);
 \node at (-.53,-.05) {$\scriptstyle\tau$};
 \node at (.2,0) {$\scriptstyle v$};
 \draw[fill=blue] (60:-.4) circle (.05cm);
 \draw[fill=DarkGreen] (60:.4) circle (.05cm);
}\]
The action of this operator on the local Hilbert space $\cH_v=\cC(x\to Ay)$ is given by
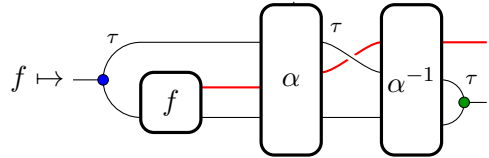
\[
 f\mapsto
 \tikzmath{
  \draw (0,0) -- (.4,0);
  \draw (2.5,.5) -- (.9,.5) arc (90:270:.5cm);
  \node at (.53,.55) {$\scriptstyle\tau$};
  \draw[fill=blue] (.4,0) circle (.07cm);
  \draw[red,thick] (1.7,-.1) -- (2.5,-.1);
  \draw (2.9,1,-.1) -- (2.5,-.1);
  \draw (1.1,-.5) -- (2.5,-.5);
  \roundNbox{fill=white}{(1.3,-.3)}{.4}{0}{0}{$f$}
  \draw[red,thick] (3.3,.1) .. controls ++(0:.2cm) and ++(0:-.2cm) .. (4.1,.5);
  \node at (3.5,.7) {$\scriptstyle\tau$};
  \draw[knot] (3.3,.5) .. controls ++(0:.2cm) and ++(0:-.2cm) .. (4.1,.1);
  \draw (3.3,-.5) -- (4.1,-.5);
  \roundNbox{fill=white}{(2.9,0)}{1}{-.6}{-.6}{$\alpha$}
  \draw (4.9,0) arc (90:-90:.3cm); 
  \node at (5.3,0) {$\scriptstyle\tau$};
  \draw (5.2,-.3) -- (5.5,-.3);
  \draw[fill=DarkGreen] (5.2,-.3) circle (.07cm);
  \draw[red,thick] (4.9,.5) -- (5.5,.5);
  \roundNbox{fill=white}{(4.5,0)}{1}{-.6}{-.6}{$\alpha^{-1}$}
 }
\]

If the vertical link is labelled by $|1\rangle$, i.e. on states in the image of $C_\ell$, these operators trivialize to $\delta_{\mBasisCircle{.07}{blue}=\mBasisCircle{.07}{DarkGreen}}$.
The subspace where the on the vertical link is in $\spann\{|\iota\rangle,|\theta\rangle\}$, corresponding to the summand $b\subseteq A$, is also preserved.
For each choice of $\mBasisCircle{.07}{blue}$ and $\mBasisCircle{.07}{DarkGreen}$, we thus get a matrix $M[\mBasisCircle{.07}{DarkGreen},\mBasisCircle{.07}{blue}]\in\cC(A,A)\cong M_2(\mathbb{C})$.

We can compute the half-braiding of $b$ under $\tau$ using the hexagon equation.
In the same basis as the associator/$F$-symbols, the half-braiding $e_\tau:\tau b\to b\tau$ is given by
\[e_\tau^\tau=\left(\begin{array}{cc}-\phi^{-2}&q^{-3/2}\sqrt{1-\phi^{-4}}\\q^{3/2}\sqrt{1-\phi^{-4}}&\phi^{-2}\end{array}\right),\qquad\qquad e_\tau^1=\left(1\right)\] 
where $q^{1/2}=e^{2\pi i/20}$.

There are $3$ possible choices for each of the vertices $\mBasisCircle{.07}{blue}$ and $\mBasisCircle{.07}{DarkGreen}$:
$\gamma_\tau^{\tau,\tau}=\tikzmath{\draw(0:0)--(30:.4);\draw(0:0)--(150:.4);\draw(0:0)--(-90:.4);}:\tau\to\tau\tau$,
$\gamma_\tau^{\tau,1}=\tikzmath{\draw[dotted](0:0)--(30:.4);\draw(0:0)--(150:.4);\draw(0:0)--(-90:.4);}:\tau\to\tau1$, and
$\gamma_1^{\tau,\tau}=\phi^{-1/2}\tikzmath{\draw(0:0)--(30:.4);\draw(0:0)--(150:.4);\draw[dotted](0:0)--(-90:.4);}:1\to\tau\tau$.
Thus, there are $9$ matrices to compute. We list each in the basis $\{|\theta\rangle,|\iota\rangle\}$.
\begin{align*}
 M\left[\tvttt,\tvttt\right] &= \left(\begin{array}{cc}2\phi^{-3}&q^{7/2}\sqrt{7\phi^{-1}-4}\\q^{-7/2}\sqrt{7\phi^{-1}-4}&\phi^{-2}\end{array}\right)\\
 M\left[\tvttt,\tvitt\right] &= \left(\begin{array}{cc}\sqrt{5\phi^5}&q^{-3/2}\sqrt{\phi^{-2}+\phi^{-4}}\\0&0\end{array}\right)
 \displaybreak[1]\\
 M\left[\tvttt,\tvtti\right] &= \left(\begin{array}{cc}-\phi^{-2}&0\\q^{3/2}\sqrt{1-\phi^{-4}}&0\end{array}\right)
 \displaybreak[1]\\
 M\left[\tvtti,\tvttt\right] &= \left(\begin{array}{cc}\sqrt{5\phi^{-5}}&0\\q^{3/2}\sqrt{\phi^{-2}+\phi^{-4}}&0\end{array}\right)
 \displaybreak[1]\\
 M\left[\tvitt,\tvttt\right] &= \left(\begin{array}{cc}-\phi^{-2}&0\\0&0\end{array}\right)
 \displaybreak[1]\\
 M\left[\tvtti,\tvitt\right] &= \left(\begin{array}{cc}\phi^{-4}&0\\0&0\end{array}\right)
 \displaybreak[1]\\
 M\left[\tvtti,\tvtti\right] &= \left(\begin{array}{cc}0&q^{-3/2}\sqrt{1-\phi^{-4}}\\0&\phi^{-2}\end{array}\right)
 \displaybreak[1]\\
 M\left[\tvitt,\tvitt\right] &= \left(\begin{array}{cc}0&0\\q^{3/2}\sqrt{1-\phi^{-4}}&\phi^{-2}\end{array}\right)
 \displaybreak[1]\\
 M\left[\tvitt,\tvtti\right] &= \left(\begin{array}{cc}1&0\\0&0\end{array}\right)
\end{align*}
Each $0$ in these matrices appears because of the requirement that $B_p^\tau=(B_p^\tau)^\dag=0$ on states which excite $A_v$.

The other difficulty is in resolving the $D_{p,q}$ term.
Most of the term is analogous to string operators, and can be written as local operators in a similar way, as shown in \eqref{eq:DTerm}.
The novel ingredient is the condensation morphism
\[\tikzmath{\draw[red,thick](0,0) arc (180:0:.3cm);\draw[red,thick](.3,.3)--(.3,.6);\draw[red,thick](0,.9) arc (-180:0:.3cm);}=m^\dag m:AA\to AA\text{,}\]
viewed as a morphism in $\cX$.
In other words, we want the corresponding matrices $\mu^\tau\in\End(\cX(\tau\to AA))\cong M_5(\mathbb{C})$ and $\mu^1\in\End(\cX(1\to AA))\cong M_5(\mathbb{C})$.

To give a basis for $\cX(\tau\to AA)$, we consider that any morphism $\tau\to AA$ factors as $(f\otimes g)\circ\gamma$, where $\gamma:\tau\to xy$ is a fusion channel, $xy\in\Irr(\cX)$ are edge labels, and $f:x\to A$ and $g:y\to A$.
Since there is at most one fusion channel $\tau\to xy$, it suffices to pick $f$ and $g$ in the basis $\{|1\rangle,|\iota\rangle,|\theta\rangle\}$ of $\cX(1\to A)\oplus\cX(\tau\to A)$ previously chosen.
This gives the basis
\[\{|\theta\theta\rangle,|\theta\iota\rangle,|\iota\theta\rangle,|\theta1\rangle,|1\theta\rangle\}\subseteq\cX(\tau\to AA)\text{.}\]
In the same way, we obtain a basis
\[\{|\theta\theta\rangle,|\iota\iota\rangle,|\iota1\rangle,|1\iota\rangle,|11\rangle\}\subseteq\cX(1\to AA)\text{.}\]
In the above bases, the matrices are as follows.
\begin{align*}
 \mu^\tau &= \frac{1}{3}\left(
 \begin{array}{*5{>{\scriptstyle}c}}
  \phi^{-4}+\phi^{-2} & q^{-3/2}\sqrt{\phi^{-5}+\phi^{-7}} & q^{-3/2}\sqrt{\phi^{-5}+\phi^{-7}} & q^{7/2}\sqrt{\phi^{-2}+\phi^{-4}} & q^{7/2}\sqrt{\phi^{-2}+\phi^{-4}}\\
  q^{3/2}\sqrt{\phi^{-5}+\phi^{-7}} & \phi^{-3} & \phi^{-3} & -\phi^{-3/2} & -\phi^{-3/2}\\
  q^{3/2}\sqrt{\phi^{-5}+\phi^{-7}} & \phi^{-3} & \phi^{-3} & -\phi^{-3/2} & -\phi^{-3/2}\\
  q^{-7/2}\sqrt{\phi^{-2}+\phi^{-4}} & -\phi^{-3/2} & -\phi^{-3/2} & 1 & 1\\
  q^{-7/2}\sqrt{\phi^{-2}+\phi^{-4}} & -\phi^{-3/2} & -\phi^{-3/2} & 1 & 1
 \end{array}\right)\\
 \mu^1 &= \frac{1}{6}\left(
 \begin{array}{ccccc}
  \phi^2 & q^2\phi^{-3/2} & q^22\phi^{-1} & q^22\phi^{-1} & q^23\phi^{-1/2}\\
  q^{-2}\phi^{-3/2} & 3-\phi^{-1} & -2\phi^{-1/2} & -2\phi^{-1/2} & 3\phi^{-1}\\
  q^{-2}2\phi^{-1} & -2\phi^{-1/2} & 2 & 2 & 0\\
  q^{-2}2\phi^{-1} & -2\phi^{-1/2} & 2 & 2 & 0\\
  q^{-2}3\phi^{-1/2} & 3\phi^{-1} & 0 & 0 & 3
 \end{array}\right)
\end{align*}

We now consider the fates of anyons in the condensed phase, i.e.~at $t=1$.
We begin by computing the category $Z(\Fib)_A$, making use of the free-forgetful adjunction.
Because $A\cong1\oplus b$, we have $\tau A\cong \tau\oplus\overline{\tau}\oplus b\cong\overline{\tau}A$.
Therefore, the free modules $\tau A_A$ and $\overline{\tau}A_A$ are simple and isomorphic.
Finally, $bA_A\cong b\oplus1\oplus\tau\oplus\overline{\tau}\oplus b$, so $bA_A\cong A_A\oplus\tau A_A$.
Thus, $\Irr(Z(\Fib)_A)$ contains two simple objects: the vacuum $A_A$, and a single species of excitation $\tau A_A$.

The module $\tau A_A$ is not local, and hence does not correspond to a topological excitation at $t=1$.
By the $R$-symbols $R_b^{b,\tau}=R_b^{\tau,b}=q^3$ and $R_{\overline{\tau}}^{b,\tau}=R_{\overline{\tau}}^{\tau,b}=q^6$, we can see that the double-braiding between $b$ and $\tau$ is given by
 \[\tikzmath[scale=.375]{
  \draw (.5,0) node[below]{$\scriptstyle b$} -- (-.5,1);
  \draw[knot] (-.5,0) node[below]{$\scriptstyle\tau$} -- (.5,1) -- (-.5,2);
  \draw[knot] (-.5,1) -- (.5,2);
 }
 =q^6~
 \tikzmath[scale=.4]{
  \draw (-.5,1) --(0,.5);
  \draw (.5,1)--(0,.5);
  \draw (-.5,-1) node[below]{$\scriptstyle\tau$} --(0,-.5);
  \draw (.5,-1) node[below]{$\scriptstyle b$} --(0,-.5);
  \draw (0,.5)-- node[right]{$\scriptstyle b$} (0,-.5);
 }
 +q^2~
 \tikzmath[scale=.4]{
  \draw (-.5,1)--(0,.5);
  \draw (.5,1)--(0,.5);
  \draw (-.5,-1) node[below]{$\scriptstyle\tau$} --(0,-.5);
  \draw (.5,-1) node[below]{$\scriptstyle b$} --(0,-.5);
  \draw (0,.5)-- node[right]{$\scriptstyle\overline{\tau}$} (0,-.5);
 }
 \]
 where the trivalent vertices on the right-hand side are chosen so that
 \[\id_{\tau b}=
 \tikzmath[scale=.4]{
  \draw (-.5,1) --(0,.5);
  \draw (.5,1)--(0,.5);
  \draw (-.5,-1) node[below]{$\scriptstyle\tau$} --(0,-.5);
  \draw (.5,-1) node[below]{$\scriptstyle b$} --(0,-.5);
  \draw (0,.5)-- node[right]{$\scriptstyle b$} (0,-.5);
 }
 +~
 \tikzmath[scale=.4]{
  \draw (-.5,1)--(0,.5);
  \draw (.5,1)--(0,.5);
  \draw (-.5,-1) node[below]{$\scriptstyle\tau$} --(0,-.5);
  \draw (.5,-1) node[below]{$\scriptstyle b$} --(0,-.5);
  \draw (0,.5)-- node[right]{$\scriptstyle\overline{\tau}$} (0,-.5);
 }
 \]
 Thus, the defect operators $\sigma^\tau_r$ will not commute with $D_{p,q}$ terms which cross $r$, and hence create an excitation which is not topologically mobile.

\bibliographystyle{amsalpha}
{\footnotesize{
\bibliography{../../bibliography/bibliography}

\newcommand{\etalchar}[1]{$^{#1}$}
\providecommand{\bysame}{\leavevmode\hbox to3em{\hrulefill}\thinspace}
\providecommand{\MR}{\relax\ifhmode\unskip\space\fi MR }
\providecommand{\MRhref}[2]{%
  \href{http://www.ams.org/mathscinet-getitem?mr=#1}{#2}
}
\providecommand{\href}[2]{#2}
\begin{thebibliography}{BMW{\etalchar{+}}17}

\bibitem[ALW19]{1709.01941}
David Aasen, Ethan Lake, and Kevin Walker, \emph{Fermion condensation and super
  pivotal categories}, Journal of Mathematical Physics \textbf{60} (2019),
  no.~12, 121901, \doi{10.1063/1.5045669} \arxiv{1709.01941}.

\bibitem[BD12]{MR2889539}
Thomas Booker and Alexei Davydov, \emph{Commutative algebras in {F}ibonacci
  categories}, J. Algebra \textbf{355} (2012), 176--204, \mathscinet{MR2889539}
  \doi{10.1016/j.jalgebra.2011.12.029} \arxiv{1103.3537}. \MR{2889539}

\bibitem[BD19]{MR4051062}
Alex Bullivant and Clement Delcamp, \emph{Tube algebras, excitations statistics
  and compactification in gauge models of topological phases}, J. High Energy
  Phys. (2019), no.~10, 216, 76, \mathscinet{MR4051062}
  \doi{10.1007/jhep10(2019)216} \arxiv{1905.08673}. \MR{4051062}

\bibitem[BL23]{2303.07291}
Fiona Burnell and Chien-Hung Lin, \emph{Anyon condensation in the string-net
  models}, 2023, \arxiv{2303.07291}.

\bibitem[BMW{\etalchar{+}}17]{MR3614057}
N.~Bultinck, M.~Mari\"{e}n, D.~J. Williamson, M.~B. \c{S}ahino\u{g}lu,
  J.~Haegeman, and F.~Verstraete, \emph{Anyons and matrix product operator
  algebras}, Ann. Physics \textbf{378} (2017), 183--233, \mathscinet{MR3614057}
  \doi{10.1016/j.aop.2017.01.004} \arxiv{1511.08090}. \MR{3614057}

\bibitem[BS09]{PhysRevB.79.045316}
F.~A. Bais and J.~K. Slingerland, \emph{Condensate-induced transitions between
  topologically ordered phases}, Phys. Rev. B \textbf{79} (2009), 045316,
  \doi{10.1103/PhysRevB.79.045316} \arxiv{0808.0627}.

\bibitem[BSH09]{MR2516228}
F.~A. Bais, J.~K. Slingerland, and S.~M. Haaker, \emph{Theory of topological
  edges and domain walls}, Phys. Rev. Lett. \textbf{102} (2009), no.~22,
  220403, 4, \mathscinet{MR2516228} \doi{10.1103/PhysRevLett.102.220403}
  \arxiv{0812.4596}. \MR{2516228}

\bibitem[BSS11]{PhysRevB.84.125434}
F.~J. Burnell, Steven~H. Simon, and J.~K. Slingerland, \emph{Condensation of
  achiral simple currents in topological lattice models: Hamiltonian study of
  topological symmetry breaking}, Phys. Rev. B \textbf{84} (2011), 125434,
  \doi{10.1103/PhysRevB.84.125434} \arxiv{1104.1701}.

\bibitem[CHPJP22]{MR4419534}
Quan Chen, Roberto Hern\'{a}ndez~Palomares, Corey Jones, and David Penneys,
  \emph{Q-system completion for {$\rm C^*$} 2-categories}, J. Funct. Anal.
  \textbf{283} (2022), no.~3, Paper No. 109524, \mathscinet{MR4419534}
  \doi{10.1016/j.jfa.2022.109524} \arxiv{2105.12010}. \MR{4419534}

\bibitem[DMNO13]{MR3039775}
Alexei Davydov, Michael M{\"u}ger, Dmitri Nikshych, and Victor Ostrik,
  \emph{The {W}itt group of non-degenerate braided fusion categories}, J. Reine
  Angew. Math. \textbf{677} (2013), 135--177, \mathscinet{MR3039775}
  \arxiv{1009.2117}. \MR{3039775}

\bibitem[DNO13]{MR3022755}
Alexei Davydov, Dmitri Nikshych, and Victor Ostrik, \emph{On the structure of
  the {W}itt group of braided fusion categories}, Selecta Math. (N.S.)
  \textbf{19} (2013), no.~1, 237--269, \mathscinet{MR3022755}
  \doi{10.1007/s00029-012-0093-3} \arxiv{1109.5558}. \MR{3022755}

\bibitem[FFRS06]{MR2187404}
J\"{u}rg Fr\"{o}hlich, J\"{u}rgen Fuchs, Ingo Runkel, and Christoph Schweigert,
  \emph{Correspondences of ribbon categories}, Adv. Math. \textbf{199} (2006),
  no.~1, 192--329, \mathscinet{MR2187404} \doi{10.1016/j.aim.2005.04.007}
  \arxiv{math/0309465}. \MR{2187404}

\bibitem[GJ16]{MR3447719}
Shamindra~Kumar Ghosh and Corey Jones, \emph{Annular representation theory for
  rigid {$C^*$}-tensor categories}, J. Funct. Anal. \textbf{270} (2016), no.~4,
  1537--1584, \mathscinet{MR3447719} \doi{10.1016/j.jfa.2015.08.017}
  \arxiv{1502.06543}. \MR{3447719}

\bibitem[HBFL16]{PhysRevB.94.235136}
Chris Heinrich, Fiona Burnell, Lukasz Fidkowski, and Michael Levin,
  \emph{Symmetry-enriched string nets: Exactly solvable models for {SET}
  phases}, Phys. Rev. B \textbf{94} (2016), 235136,
  \doi{10.1103/PhysRevB.94.235136} \arxiv{1606.07816}.

\bibitem[HBJP22]{2208.14018}
Peter Huston, Fiona Burnell, Corey Jones, and David Penneys, \emph{Composing
  topological domain walls and anyon mobility}, 2022, \arxiv{2208.14018}.

\bibitem[HGW18]{PhysRevB.97.195154}
Yuting Hu, Nathan Geer, and Yong-Shi Wu, \emph{Full dyon excitation spectrum in
  extended {L}evin-{W}en models}, Phys. Rev. B \textbf{97} (2018), 195154,
  \doi{10.1103/PhysRevB.97.195154} \arxiv{1502.03433}.

\bibitem[HMH20]{1811.09275}
Tomohiro Hashizume, Ian~P. McCulloch, and Jad~C. Halimeh, \emph{Dynamical phase
  transitions in the two-dimensional transverse-field ising model}, 2020,
  \arxiv{1811.09275}.

\bibitem[{Hon}09]{0907.2204}
Seung-Moon {Hong}, \emph{{On symmetrization of 6j-symbols and Levin-Wen
  Hamiltonian}}, July 2009, \arxiv{0907.2204}, p.~arXiv:0907.2204.

\bibitem[HP17]{MR3663592}
Andr\'e Henriques and David Penneys, \emph{Bicommutant categories from fusion
  categories}, Selecta Math. (N.S.) \textbf{23} (2017), no.~3, 1669--1708,
  \mathscinet{MR3663592} \doi{10.1007/s00029-016-0251-0} \arxiv{1511.05226}.
  \MR{3663592}

\bibitem[HSW12]{PhysRevB.85.075107}
Yuting Hu, Spencer~D. Stirling, and Yong-Shi Wu, \emph{Ground-state degeneracy
  in the levin-wen model for topological phases}, Phys. Rev. B \textbf{85}
  (2012), 075107, \doi{10.1103/PhysRevB.85.075107} \arxiv{1105.5771}.

\bibitem[Izu00]{MR1782145}
Masaki Izumi, \emph{The structure of sectors associated with {L}ongo-{R}ehren
  inclusions. {I}. {G}eneral theory}, Comm. Math. Phys. \textbf{213} (2000),
  no.~1, 127--179, \mathscinet{MR1782145} \doi{10.1007/s002200000234}.
  \MR{1782145}

\bibitem[{Kir}11]{1106.6033}
Alexander {Kirillov Jr}, \emph{{String-net model of Turaev-Viro invariants}},
  2011, \arxiv{1106.6033}.

\bibitem[Kit03]{MR1951039}
A.~Yu. Kitaev, \emph{Fault-tolerant quantum computation by anyons}, Ann.
  Physics \textbf{303} (2003), no.~1, 2--30, \mathscinet{MR1951039}
  \doi{10.1016/S0003-4916(02)00018-0} \arxiv{quant-ph/9707021}. \MR{1951039}

\bibitem[Kit06]{cond-mat/0506438}
Alexei Kitaev, \emph{Anyons in an exactly solved model and beyond}, Annals of
  Physics \textbf{321} (2006), no.~1, 2--111, January Special Issue.
  \doi{10.1016/j.aop.2005.10.005} \arxiv{cond-mat/0506438}.

\bibitem[KK12]{MR2942952}
Alexei Kitaev and Liang Kong, \emph{Models for gapped boundaries and domain
  walls}, Comm. Math. Phys. \textbf{313} (2012), no.~2, 351--373,
  \mathscinet{MR2942952} \doi{10.1007/s00220-012-1500-5} \arxiv{1104.5047}.
  \MR{2942952}

\bibitem[KKR10]{MR2726654}
Robert Koenig, Greg Kuperberg, and Ben~W. Reichardt, \emph{Quantum computation
  with {T}uraev-{V}iro codes}, Ann. Physics \textbf{325} (2010), no.~12,
  2707--2749, \mathscinet{MR2726654} \doi{10.1016/j.aop.2010.08.001}
  \arxiv{1002.2816}. \MR{2726654}

\bibitem[KL94]{MR1280463}
Louis~H. Kauffman and S{\'o}stenes~L. Lins, \emph{Temperley-{L}ieb recoupling
  theory and invariants of {$3$}-manifolds}, Annals of Mathematics Studies,
  vol. 134, Princeton University Press, Princeton, NJ, 1994,
  \mathscinet{MR1280463}.

\bibitem[Kon14a]{MR3204497}
L.~Kong, \emph{Some universal properties of {L}evin-{W}en models}, X{VII}th
  {I}nternational {C}ongress on {M}athematical {P}hysics, World Sci. Publ.,
  Hackensack, NJ, 2014, \mathscinet{MR3204497}
  \doi{10.1142/9789814449243\_0042} \arxiv{1211.4644}, pp.~444--455.
  \MR{3204497}

\bibitem[Kon14b]{MR3246855}
Liang Kong, \emph{Anyon condensation and tensor categories}, Nuclear Phys. B
  \textbf{886} (2014), 436--482, \mathscinet{MR3246855}
  \doi{10.1016/j.nuclphysb.2014.07.003} \arxiv{1307.8244}. \MR{3246855}

\bibitem[KZ21]{1912.01760}
Liang Kong and Hao Zheng, \emph{A mathematical theory of gapless edges of 2d
  topological orders. part {II}}, Nuclear Physics B \textbf{966} (2021),
  115384, \doi{10.1016/j.nuclphysb.2021.115384} \arxiv{1912.01760}.

\bibitem[LLB21]{PhysRevB.103.195155}
Chien-Hung Lin, Michael Levin, and Fiona~J. Burnell, \emph{Generalized
  string-net models: A thorough exposition}, Phys. Rev. B \textbf{103} (2021),
  195155, \doi{10.1103/PhysRevB.103.195155} \arxiv{2012.14424}.

\bibitem[LR97]{MR1444286}
R.~Longo and J.~E. Roberts, \emph{A theory of dimension}, $K$-Theory
  \textbf{11} (1997), no.~2, 103--159, \mathscinet{MR1444286}
  \doi{10.1023/A:1007714415067} \arxiv{funct-an/9604008}. \MR{1444286}

\bibitem[LVHV20]{PhysRevLett.124.120601}
Laurens Lootens, Robijn Vanhove, Jutho Haegeman, and Frank Verstraete,
  \emph{Galois conjugated tensor fusion categories and nonunitary conformal
  field theory}, Phys. Rev. Lett. \textbf{124} (2020), 120601.

\bibitem[LW05]{PhysRevB.71.045110}
Michael~A. Levin and Xiao-Gang Wen, \emph{String-net condensation: A physical
  mechanism for topological phases}, Phys. Rev. B \textbf{71} (2005), 045110,
  doi{10.1103/PhysRevB.71.045110} \arxiv{cond-mat/0404617}.

\bibitem[LW14]{PhysRevB.90.115119}
Tian Lan and Xiao-Gang Wen, \emph{Topological quasiparticles and the
  holographic bulk-edge relation in $(2+1)$-dimensional string-net models},
  Phys. Rev. B \textbf{90} (2014), 115119, \doi{10.1103/PhysRevB.90.115119}
  \arxiv{1311.1784}.

\bibitem[M{\"u}g03a]{MR1966524}
Michael M{\"u}ger, \emph{From subfactors to categories and topology. {I}.
  {F}robenius algebras in and {M}orita equivalence of tensor categories}, J.
  Pure Appl. Algebra \textbf{180} (2003), no.~1-2, 81--157,
  \mathscinet{MR1966524} \doi{10.1016/S0022-4049(02)00247-5}
  \arXiv{math.CT/0111204}.

\bibitem[M{\"u}g03b]{MR1966525}
\bysame, \emph{From subfactors to categories and topology. {II}. {T}he quantum
  double of tensor categories and subfactors}, J. Pure Appl. Algebra
  \textbf{180} (2003), no.~1-2, 159--219, \mathscinet{MR1966525}
  \doi{10.1016/S0022-4049(02)00248-7} \arXiv{math.CT/0111205}.

\bibitem[M{\"u}g03c]{MR1990929}
\bysame, \emph{On the structure of modular categories}, Proc. London Math. Soc.
  (3) \textbf{87} (2003), no.~2, 291--308, \mathscinet{MR1990929}
  \doi{10.1112/S0024611503014187}. \MR{MR1990929 (2004g:18009)}

\bibitem[NSS{\etalchar{+}}08]{MR2443722}
Chetan Nayak, Steven~H. Simon, Ady Stern, Michael Freedman, and Sankar
  Das~Sarma, \emph{Non-abelian anyons and topological quantum computation},
  Rev. Modern Phys. \textbf{80} (2008), no.~3, 1083--1159,
  \mathscinet{MR2443722} \doi{https://doi.org/10.1103/RevModPhys.80.1083}
  \arxiv{0707.1889}.

\bibitem[Ons44]{PhysRev.65.117}
Lars Onsager, \emph{Crystal statistics. i. a two-dimensional model with an
  order-disorder transition}, Phys. Rev. \textbf{65} (1944), 117--149,
  \doi{10.1103/PhysRev.65.117}.

\bibitem[Ost03]{MR1976459}
Victor Ostrik, \emph{Module categories, weak {H}opf algebras and modular
  invariants}, Transform. Groups \textbf{8} (2003), no.~2, 177--206,
  \mathscinet{MR1976459} \arXiv{math/0111139}. \MR{MR1976459 (2004h:18006)}

\bibitem[Pac12]{MR3157248}
Jiannis~K. Pachos, \emph{Introduction to topological quantum computation},
  Cambridge University Press, Cambridge, 2012, \mathscinet{MR3157248}
  \doi{10.1017/CBO9780511792908}. \MR{3157248}

\bibitem[Pen20]{MR4133163}
David Penneys, \emph{Unitary dual functors for unitary multitensor categories},
  High. Struct. \textbf{4} (2020), no.~2, 22--56, \mathscinet{MR4133163}
  \arxiv{1808.00323}. \MR{4133163}

\bibitem[Reu23]{MR4538281}
David Reutter, \emph{Uniqueness of {U}nitary {S}tructure for {U}nitarizable
  {F}usion {C}ategories}, Comm. Math. Phys. \textbf{397} (2023), no.~1, 37--52,
  \mathscinet{MR4538281} \doi{10.1007/s00220-022-04425-7} \arxiv{1906.09710}.
  \MR{4538281}

\bibitem[SRD{\etalchar{+}}21]{1901.00278}
Markus Schmitt, Marek~M. Rams, Jacek Dziarmaga, Markus Heyl, and Wojciech~H.
  Zurek, \emph{{Quantum phase transition dynamics in the two-dimensional
  transverse-field Ising model}}, Jun 2021, \arxiv{1901.00278}.

\bibitem[TTWL08]{10.1143/PTPS.176.384}
Simon Trebst, Matthias Troyer, Zhenghan Wang, and Andreas W.~W. Ludwig,
  \emph{{A Short Introduction to Fibonacci Anyon Models}}, Progress of
  Theoretical Physics Supplement \textbf{176} (2008), 384--407,
  \doi{10.1143/PTPS.176.384} \arxiv{0902.3275}.

\bibitem[XLLC21]{2110.06079}
Wenjie Xi, Ya-Lei Lu, Tian Lan, and Wei-Qiang Chen, \emph{A lattice realization
  of general three-dimensional topological order}, 2021, \arxiv{2110.06079}.

\bibitem[Yam04]{MR2091457}
Shigeru Yamagami, \emph{Frobenius duality in {$C^*$}-tensor categories}, J.
  Operator Theory \textbf{52} (2004), no.~1, 3--20, \mathscinet{MR2091457}.
  \MR{2091457 (2005f:46109)}

\bibitem[Zha17]{YanbaiZhang}
Yanbai Zhang, \emph{From the {T}emperley-{L}ieb categories to toric code},
  2017, Undergraduate honors thesis, available at
  \url{https://tqft.net/web/research/students/YanbaiZhang/thesis.pdf}.

\bibitem[ZHW{\etalchar{+}}22]{2209.12750}
Yu~Zhao, Shan Huang, Hongyu Wang, Yuting Hu, and Yidun Wan, \emph{{Exactly
  solvable Hamiltonian model of the doubled Ising and $\mathbb{Z}_2$ toric code
  topological phases separated by a gapped domain wall via anyon
  condensation}}, 9 2022, \arxiv{2209.12750}.

\end{thebibliography}
}}
\end{document}